\documentclass[12pt,twoside]{book}
 \usepackage{amsmath,amsfonts,amssymb,latexsym,graphicx,setspace}

\newtheorem{theorem}{Theorem}[section]
\newtheorem{definition}{Definition}[section]

\newtheorem{proposition}{Proposition}[section]
\newtheorem{remark}{Remark}[section]

\newtheorem{example}{Example}[section]
\newenvironment{proof}[1][Proof]{\textsc{#1.} }{\ \rule{0.5em}{0.5em}}
\numberwithin{equation}{section}
\def\be{\begin{equation}}
\def\ee{\end{equation}}
\def\bq{\begin{eqnarray}}
\def\eq{\end{eqnarray}}
\def\beq{\begin{eqnarray*}}
\def\eeq{\end{eqnarray*}}

\def\a{\alpha}
\def\b{\beta}

\def\d{\delta}

\begin{document}
\pagenumbering{roman}
\begin{titlepage}

\vspace{4cm}

\begin{center}
{\LARGE \bf{The nature of cosmological singularities in isotropic universes and
braneworlds}}

\vspace{3cm}

by \vspace{0.5cm}

{\Large Ifigeneia Klaoudatou}

\end{center}

\vspace{2cm} {\large
\begin{center}
A thesis defended before the advisory and examining committee

In partial fulfillment
of the requirements for the degree of \\

Doctor of Philosophy

Department of Information and Communication Systems Engineering

University of the Aegean

\end{center}
\vspace{0.5cm}
\begin{center}
\textbf{Samos 2008}
\end{center}}
\vspace{0.5cm}
\end{titlepage}
\thispagestyle{empty}
\newpage
\pagebreak
\newpage
\;\;\; \vfill
\parindent=0.75cm
\vfill
\begin{center}
\large\em\null\vskip1in \vfill \
\end{center}
\pagebreak
\newpage
\thispagestyle{empty}
\newpage
\begin{titlepage}
\begin{center}
{\small
The advisory committee \\
\vspace{0.5cm}
Spiros Cotsakis, Professor, Supervisor\\
Department of Information and Communication Systems Engineering

George Flessas, Professor, Member\\
Department of Information and Communication Systems Engineering

Ioannis Miritzis, Assistant Professor, Member\\
Department of Marine Science

\vspace{0.3cm}
University of the Aegean

\vspace{1.5cm}

The examining committee \\
\vspace{0.5cm}
Spiros Cotsakis, Professor\\
Department of Information and Communication Systems Engineering,
University of the Aegean

George Flessas, Professor\\
Department of Information and Communication Systems Engineering,
University of the Aegean

Ioannis Miritzis, Assistant Professor\\
Department of Marine Science, University of the Aegean

George Vouros, Professor\\
Department of Information and Communication Systems Engineering,
University of the Aegean

Konstantinos Vagionakis, Professor\\
Department of Physics, University of Ioannina

Konstantinos Kokkotas, Associate Professor\\
Department of Physics, Aristotle University of Thessaloniki

Nikolaos Stergioulas, Assistant Professor\\
Department of Physics, Aristotle University of Thessaloniki}
\end{center}
\end{titlepage}
\thispagestyle{empty}
\newpage
\pagebreak
\newpage
\;\;\; \vfill
\parindent=0.75cm
\vfill
\begin{center}
\large\em\null\vskip1in \vfill \
\end{center}
\pagebreak
\newpage
\thispagestyle{empty}
\newpage
\begin{center}
{\Large{\bf Acknowledgements}}
\end{center}
I owe a great dept of gratitude to my supervisor Spiros Cotsakis for his
encouragement and enthusiasm that were the driving force for starting and
completing my PhD. I am also grateful to him for his contribution and guidance
in research and in the writing up of our work. This thesis has been
influenced dramatically both in style and content from his comments and
suggestions.

I would like to acknowledge the financial support from the Ministry of
Education grant `Herakleitos'.

I thank the members of the advisory committee George Flessas and John Myritzis
for their support. I am also grateful to the members of our group Manolis Delis
and Antonios Tsokaros for useful discussions.

I am very grateful to Ignatios Antoniadis for making it possible for me to stay
at CERN, where part of our research was done.
I am also grateful to him for all the useful discussions we had during the
making and writing up of our joint work with Spiros Cotsakis.

I would also like to thank Mihalis Dafermos for inviting me to the
Seventh British Gravitational Conference which gave me the opportunity to
present part of my work and discuss it with other British relativists.

Many thanks to the Department of Information and Communication Systems
Engineering of the University of the Aegean for offering
ideal conditions for doing a research degree and for giving me the opportunity
to work as a tutor. Special thanks to the Head of the
Department Stephanos Gritzalis.

I would also like to thank the Library of the Eugenidion Foundation.
I have benefited of its inspiring environment and
hospitality since my school years. I wish to thank especially
Eleni Tsalaga, staff of the Library, for her help that was beyond her typical
obligations.

I am also grateful to my brother-in-law Giorgos Vourdachas for his help in
organising and technical issues.

I particularly thank my friends Georgia Kotsala, Betty Dertili, Orestis
Anavaloglou and Elisavet Konstantinou. The times of fun and joy that we
shared are mostly appreciated. Special thanks to Danai Konstanta for her
encouragement and for sharing her humor and idiosyncratic ideas on various
`existential' matters throughout this project and life in general.
I am looking forward to attend SEP classes with her and my sister someday.

I am deeply grateful to my sister Eleni for her support and
constructive criticism in all stages of my PhD. She has helped me
overcome all the difficulties I met through the years of my study.
In various occasions of frenetic rhythms of work she has been the `voice
of wisdom' suggesting that I should simply `stop and smell the roses on the way'.

Finally, I dedicate this thesis to my mother Tzeni.
\newpage

\begin{center}
{\bf {\Large Abstract}}
\end{center}
In this thesis we explore the character of finite-time
singularities that are possible
to arise in FRW universes and in specific braneworld configurations.

We begin with a review of the structure and properties of RW-spaces. We
then give an adaptation of a completeness theorem 
and analyse complete model universes from the recent literature.

We continue with the construction of a classification plan of singularities
that may appear in FRW universes. Our tool for this quest is the
contrapositive of the completeness theorem 
which translates in
\emph{necessary} conditions for singularities. The classification derived in this
way is based on the behaviour of the Hubble rate and offers a first insight
into the nature of singularities.
We further extend and complete this classification with the use of the
Bel-Robinson energy. We also associate the existence of a closed trapped
surface with a divergence in the Bel-Robinson
energy and use this relation to determine the evolution of various specific
cosmologies.

Next we study models from the recent literature that offer
the appropriate physical conditions for the appearance of singularities
predicted by our classification.
This study shows that our classification not only
accommodates all known types of singularities but it also predicts new ones.

Finally, we focus on braneworld models consisting of a three-brane embedded in a
five-dimensional bulk space that is inhabited by a scalar field or a perfect
fluid. Our goal is to study the singularities that may occur
within a finite distance from the brane. Our analysis that is performed with the
method of asymptotic splittings shows that these braneworlds can
exhibit three main types of finite-distance singularities.
\newpage
\tableofcontents

\listoftables
\newpage
\;\;\;
\vfill
\parindent=0.75cm
\vfill
\begin{center}
\large\em\null\vskip1in
\vfill \
\end{center}
\pagebreak
\newpage
\pagenumbering{arabic}
\chapter{Introduction} 
Ever since the fundamental discovery of the big bang and big crunch singularities
in the general relativistic FRW universes of dust and
radiation a lot of effort has been focused on finding ways to avoid these
philosophically and physically puzzling features of spacetime structure.
Indeed, spacetime
singularities constitute a challenging issue philosophically since they imply
that the immense universe emerged a finite time ago through a region of spacetime
of infinite curvature and zero size, as well as physically
since general relativity fails to describe such a state of spacetime.
A common way out of this situation is to modify general relativity to
obtain theories of gravity that do not possibly contain singularities. A
different attitude is to accept that singularities are essential features of our
universe and investigate further conditions that trigger their existence
as well as those that describe their nature.

The singularity theorems of general relativity \cite{he73} gave appropriate
conditions that led to the existence of singularities in generic spacetimes;
they did not, however, provide detailed information about their nature.
Uncovering the nature of spacetime singularities remains therefore an open issue
in cosmology and this thesis focuses on the precise pathologies
that may determine the character of singularities in the class of isotropic
spacetimes. Before we analyse more carefully our goals, let us briefly review
the development of recent research on the issue of finite-time singularities
that partly motivates our own research reported herein.

An investigation of the recent literature actually shows that
interest 
in the issue of spacetime singularities has been renewed
through the construction of various cosmological models with finite-time
singularities.
Among these are those similar to the
standard big bang and big crunch singularities ones albeit
appearing in universes containing different and even `exotic' forms of matter,
as well as other kinds not met in the standard cosmological models studied in
the usual textbook expositions of cosmology.

Interestingly, these novel types of singularities
sometimes involve the presence of quite an exotic type of matter
(cf. \cite{ph_1}-\cite{ph_42}), a characteristic example of such a singularity
being the so-called big-rip singularity (cf. \cite{cald03}, \cite{ch04},
\cite{ph_33}, \cite{steph},
\cite{sud_10}, \cite{sud_12}, \cite{gonzales}, \cite{noj1}, \cite{cald99}).
The study of some of these newer
types of singularities started off as soon as data from the WMAP
(cf. \cite{page}, \cite{spergel}) became available.
These data indicate that the universe is
currently undergoing an accelerating phase typically attributed to a form of
\emph{dark energy} that may be described by a linear equation of state with a
parameter $w$ converging, however, towards the value $w=-1$.
The case with $w<-1$ is not ruled out by the data \cite{cald03} and this provides
a surprising implication for the future of the universe which had not been
appreciated earlier:
once such a dark energy component becomes
dominant, the universe is destined to face up with a violent end towards
which the scale factor, the Hubble expansion rate, the horizon distance, the
density and the pressure all experiencing a divergence. In this sense every bound
structure will be \emph{ripped apart}.

The type of dark energy that leads to the emergence of the big-rip type of
singularity mentioned above has some awkward characteristics
(see Section 4.3) but it nonetheless sparkled new interest on
finite-time singularities. In particular, research
focused (cf. \cite{sud_10}, \cite{sud_12}, \cite{ba04}-\cite{sud_11}) in
finding paradigms of cosmologies that exhibited future finite-time
singularities potentially distinct from big crunches and rips but still feasible
to arise in an expanding phase of a model universe which did not require exotic
forms of matter. This led to the discovery of \emph{sudden} and other milder
types of singularities. These singularities are characterised by the
fact that the second or higher order derivatives of the scale factor diverge
while the scale factor itself and the Hubble rate remain finite.

The models described in the preceding paragraphs
played the important role of providing the appropriate
physical circumstances for the appearance of such singularities. But still they
provide mere \emph{examples} not general theorems that could capture
the unifying reasons for the occurrence of such spacetime singularities.
In fact, the study of an increasing number of models accommodating
finite-time singularities leaves one with the impression that different
singular behaviours are totally determined by (or connected to) the
particularities of each model.

This situation, however, provides us with a fertile ground for the formulation
of some basic questions such as: What do all these different singular
cosmologies have in common? Is there some underline reason that
unifies all these `phenomenologically' different singular behaviours?
Are there other types of finite-time singularities which are possible in these
models but have not yet been uncovered?
What are the general criteria that predict and describe the possible kinds and
nature of finite-time singularities?

This thesis aims to provide answers to such questions.
We will show that singular isotropic universes are characterised by the basic
property that their Hubble expansion rate \emph{fails} in one way or another
to be an integrable function of the proper time for an infinite time-interval.
The various ways in which this
non-integrability can be implemented leads to a classification of
finite-time singularities that consists of three possible types which
accommodate all the recently introduced singularities mentioned above
\cite{ck05}, \cite{iklaoud_2}.
This result offers first evidence for a more detailed description of the
nature of singularities in isotropic spacetimes.

Note that such a `classification' of singularities is totally effected by the
geometry of the spacetime in question; in no way do the various matter fields
influence its construction. A more subtle examination of the character
of a singularity implies however that matter fields may strongly affect the
asymptotic behaviour of these model universes and hence play a role in changing
the nature of their singularities. This type of classification needs therefore
to be refined by somehow incorporating the various matter fields into the picture.
This refinement results in a more detailed classification scheme
that sharpens our view of the nature of singularities
\cite{iklaoud_3}-\cite{iklaoud_5}.

Our approach in constructing a classification of singularities is a geometric one. In particular,
for the Hubble rate classification mentioned above we make use of the
contrapositive of the completeness theorem of \cite{chc02}, \cite{chc02pr},
while for the second we use the notion of the Bel-Robinson energy
(an invariant geometric quantity representing the slice-projected energy of the
gravitational field).
This latter approach makes our classification conducive
to more general spacetimes as well as more flexible since now more functions
describe the possible different types of singularities and trace the properties
and structure of spacetime.

After the completion of the classification of singularities in isotropic
universes, this thesis focuses on the analysis of singularities that occur in the
different context of \emph{braneworld universes}. The study of these models
(originally motivated by string theory) offers many interesting scenarios in
cosmology and particle physics (see Section 5.1 for a short review).
Our interest in this recently issued research area originates from the fact that
these models exhibit finite-\emph{distance} singularities that have many
analogies with the finite-time singularities predicted in our classification
and not met in simpler, more standard models.

For instance, the specific braneworld model of \cite{nima}
captured our interest precisely because it involves the appearance of
a finite-distance singularity. This model consists of a flat three-brane sitting
at a fixed position of an extra spatial dimension, and a scalar field propagating
throughout the whole of the enveloping five-dimensional space.
At the vicinity of the singularity occurring within a finite-distance from
the brane, the energy density of the scalar field becomes infinite and this may
offer an explanation to the small, observed value of the vacuum energy
(actualised in this model by the scalar field).

We further explore whether this singularity continues to
appear when we consider a curved instead of a flat brane. Our preliminary
results indicate that the singularity is in fact removed at an infinite distance from the brane
\cite{iklaoud_6}.
This work is further extended here by examining other types of finite-distance
singularities that become possible once we consider a different form of bulk
matter such as a perfect fluid. 

By the use of the method of asymptotic splittings developed
in \cite{skot}, we find that there are three
possible types of finite-distance singularities, triggering an infinite leak of
energy of the matter that inhabits the background space.

A more detailed plan of this thesis is as follows:

In Chapter 2, after a review of the basic properties of Robertson-Walker spaces
we give an adaptation of the proof of the general completeness theorem
of \cite{chc02} for the specific case of such spaces and analyse complete model
universes from the recent literature.

In Chapter 3, we use the contrapositive of the completeness theorem \cite{chc02},
\cite{chc02pr} to derive \emph{necessary} conditions for singularities in
isotropic spacetimes. These in turn provide us
with a first classification plan of singularities. This plan is later extended
and completed with the use of the Bel-Robinson energy. We also associate the
existence of a closed trapped surface with a divergence in the Bel-Robinson
energy and use this relation to determine the evolution of various specific
isotropic models.

In Chapter 4, we exploit model universes appearing in the recent literature to
demonstrate how some particular members of our classification indeed arise
during the evolution of such models. It follows from the analysis in this
chapter that our classification scheme accommodates all known types of
singularities. In fact, an examination of the literature shows that many new
types of singularities -predicted by our classification- have not yet been
embodied in appropriate cosmological models.

In Chapter 5, we present an asymptotic analysis of braneworld models that
consist of a three-brane embedded in a five-dimensional bulk space filled
with a scalar field or a perfect fluid. Our analysis, performed with
the method of asymptotic splittings, shows that these models
are constrained to exhibit only
three main types of singularities within a finite distance from the position
of the brane which is taken to be at the origin.

In Chapter 6 we conclude and discuss projects of future work.
This thesis has also five appendices.
In Appendix A, we explain the notation and conventions that we use throughout,
and in the appendices that follow we give some background of
some of the basic concepts that we use in this work such as the extrinsic
curvature and the Gauss-Codazzi formalism (Appendix B), the Bel-Robinson energy
(Appendix C), the energy-momentum tensor of a scalar field (Appendix D) and the
method of asymptotic splittings (Appendix E).

We declare that all statements in this work correspond to original research
conducted by the author except in all those places where it is otherwise stated
and proper citations are then given.

\chapter{Completeness and Robertson-Walker geometry}
In this Chapter we explore the character of Robertson-Walker spaces, the
basic spacetimes used and studied in this work.
We start by reviewing, in Section 2.1, the motivation behind the
adaptation of such spaces, their basic properties and structure as well
as some simple features of singularities that they may exhibit.
In Section 2.2, we first review the
assumptions of the completeness theorem stated and proved in \cite{chc02}
for generic spacetimes and then give a version of that proof for the
case of a Robertson-Walker space.
The key idea of the proof of that theorem is to study the geodesic equations
satisfied by every causal geodesic; this is a different approach from the one
taken in the proof of the `antipodal' singularity theorems, where the idea is to
study instead the Raychaudhouri equation which describes the expansion of a
bundle of geodesics.
The completeness theorem applied to the case of a Robertson-Walker space
guarantees that the integrability
of the norm of the extrinsic curvature, or equivalently of the Hubble expansion
rate, on an infinite-proper time interval suffices for the completeness of such
spaces. Finally, in Section 2.3 we analyse complete universes from the recent
literature.
\section{Structure and properties of Robertson-Walker spaces}
Since observations on the large-scale structure of the universe are limited,
the construction of cosmological models includes at a most basic level
philosophical ideas that we accept as principles. The
\emph{cosmological principle} plays a central role in the construction
of model universes and can be viewed as a generalisation of the
Copernican principle according to which the earth does not occupy a
privileged position in the solar system. The cosmological principle
assumes that in turn the solar system, our galaxy, and our local
group of galaxies do not occupy a privileged position in the universe.

The cosmological principle therefore implies that the universe on a large
enough scale ought to be spatially homogeneous. A spacetime is spatially
homogeneous if at any instant of time, represented by a spacelike hypersurface
$\Sigma_{t}$, any two points on $\Sigma_{t}$ are equivalent. More precisely,
a spatially homogeneous spacetime can be foliated by a one-parameter family of
spacelike hypersurfaces $\Sigma_{t}$
such that for each $t$ and for any pair of points $p,q$ on $\Sigma_{t}$,
there exists an isometry that maps $p$ into $q$ (cf. \cite{wald}, pp. 92-93).
This implies that the curvature of $\Sigma_{t}$ cannot vary from point to
point i.e., the spacelike slices $\Sigma_{t}$ are surfaces of constant curvature.

Apart from imposing homogeneity, the cosmological principle also requires isotropy
about every point in spacetime; this means that there are no privileged
directions, or equivalently, that the universe is spherically symmetric about
every point.

The strongest evidence for isotropy is provided by the discovery
of Penzias and Wilson \cite{penzias} that the universe is filled with a
microwave background radiation of about $3$ K that is highly isotropic.
Since isotropy implies homogeneity \cite{walker}, this fundamental discovery
provided evidence for the validity of the cosmological principle.

The requirement of isotropy further implies that there exists a class of
privileged observers, the \emph{isotropic observers}, that move along a
congruence of timelike geodesics with tangent vectors orthogonal to the
homogeneous slices $\Sigma_{t}$. These timelike geodesics do not intersect,
except possibly at singular points, so that through every non-singular point in
spacetime passes at \emph{most} one isotropic observer. The orthogonality between these geodesics
and the slices $\Sigma_{t}$ is a necessary feature of isotropic spacetimes:
if these geodesics did not cut the surfaces $\Sigma_{t}$ orthogonally, then
we would be able to construct a preferred spatial direction in violation of
isotropy (cf. \cite{wald}, p. 93).

The suitable metric for the description of an isotropic spacetime is
the Robertson-Walker metric given by
\be
d{s}^{2}=-dt^{2}+a^{2}(t)d\Omega^{2},
\ee
where 
$d\Omega^{2}$ is the metric of the three-spaces of constant negative,
zero or positive curvature. The non-zero constant curvature can be
normalised to be $\pm 1$, (so that $d\Omega^{2}$ is the metric of the unit
three-sphere or of the unit hyperbolic three-space respectively, and
the trichotomy of possibilities reads as follows:
\be
\label{frw_metric}
d{s}^{2}=-dt^{2}+a^{2}(t)(dr^{2}+f^{2}(r)(d\theta^{2}+
  \sin^{2}\theta d\phi^{2})),
\ee
where $a(t)$ is the \emph{scale factor}, and
\be
\label{f(r)}
f(r)=\left\{\begin{array}{ll}
  r, & k=0 \\
  \sin r, & k=1  \\
   \sinh r,& k=-1,
\end{array}
\right.
\ee
with $r\in [0,\infty)$ if $k=0,-1$, whereas $r\in[0,2\pi]$ if $k=+1$.
We call a universe with constant curvature $k$ equal to $-1,0$, or $+1$, an open,
flat, or closed universe respectively.
Note that the spacelike slices of a flat or open universe are of an infinite
extent 
in contrast to those of a closed universe which are 
finite. 

As we have already remarked already in a isotropic universe we are always able to
construct a privilege class of isotropic observers. Each of these observers can
be thought of as a particle of a cosmic fluid filling spacetime; the world lines of the
observers become then the flow lines of the fluid. Because of the orthogonality
between these flow lines and the spacelike slices $\Sigma_{t}$, the spacelike
coordinates of a particle in the fluid are constant -\emph{co-moving}
coordinates- along these lines.
This cosmic fluid models the matter content of the universe and the scale factor
describes the separation of neighboring flow lines (cf. \cite{he73} p.136).
In the simplest case the energy-momentum tensor of the cosmic fluid is thus
taken to be that of a perfect fluid with
density $\mu$ and pressure $p$ which are functions only of the (proper) time. The
resulting Einstein-Friedman equations can be written in the form
\bq
\label{frw_H}
3H^{2} &=&\mu+3\frac{k}{a^{2}} ,\\
\label{accfrw}
\frac{\ddot{a}}{a}&=&-\frac{1}{6}(\mu+3p),
\eq
while the conservation of the energy-momentum tensor gives
\be
\label{cons_fluid}
\dot{\mu}+3H(\mu+p)=0,
\ee
where $H=\dot{a}/a$ is the Hubble rate.

In the absence of exotic forms of matter, we assume that $\mu >0$ and $p\geq 0$
for all $t$.
Then Eq. (\ref{accfrw}) shows that $\dot{a}$ cannot be constant and thus the
universe is either expanding or contracting. The discovery of the expansion of
the universe by Hubble in 1929 offered a very interesting implication for its
history: because its density decreases as the
universe expands (see Eq. (\ref{cons_fluid})), it follows that it should have
been infinite in the past as $a(t)\rightarrow 0$; the universe must have
therefore began its existence with a singularity. In fact, the singularity
theorems predict that this is inevitable in all models satisfying
$\mu+3p>0$.

The first cosmological models with a RW geometry
that were found to possess such
a singularity were the dust and radiation-filled universes. Dust is a
pressureless fluid that provides a good approximation
for a matter component comprised solely by non-interacting galaxies. This follows
because the observed relative velocities of galaxies in
groups are small and hence galaxies are interpreted as particles of
a pressureless fluid. On the other hand, an adequate approximation to the matter
content of the universe in an early epoch is provided by radiation
which is described by taking $p=\mu/3$. 
Away from an initial singularity
open and flat universes 
typically expand forever, whereas, closed
ones end their existence with another singularity that arises within a
finite time from the initial big bang.
\section{Completeness theorems}
Although we are going to focus on exclusively isotropic spacetimes, we
refer briefly in this Section to the basic properties of general spacetimes used
in the formulation of the completeness theorem of \cite{chc02}.
In the end of this Section, we give an adapted version of the original
theorem and proof of \cite{chc02}, for our case of interest
-isotropic spacetimes. We will refer to the original version of the
completeness theorem (Theorem \ref{compl_arb} below) later in Chapter 3 when we
shall discuss the possible extensions of our work to more general spacetimes.

We begin by giving some basic definitions.
Consider a spacetime $(\mathcal{V},g)$ with
$\mathcal{V}=\mathcal{M}\times \mathcal{I},\;$ $\mathcal{I}
=(t_{0},\infty )$, where
$\mathcal{M}$ is a smooth manifold of dimension $n$ and
$^{(n+1)}g$ a Lorentzian metric which in the usual $n+1$
splitting, reads
\begin{equation}
\label{slicedspace}
^{(n+1)}g\equiv-N^{2}dt^{2}+g_{ij}(dx^{i}+\beta ^{i}dt)(dx^{j}+\beta ^{j}dt)
\end{equation}
Here $N=N(t,x^{i})$ is called the \emph{lapse function}, $\beta
^{i}(t,x^{j})$ is called the \emph{shift function} and the spatial
slices $\mathcal{M}_{t}\,(=\mathcal{M}\times \{t\})$ are spacelike
submanifolds endowed with the time-dependent spatial metric
$g_{t}\equiv g_{ij}dx^{i}dx^{j}$. We call such a spacetime a
\emph{sliced space} \cite{c03}. A sliced space is time-oriented
by increasing $t$ and we choose $\mathcal{I} =(t_{0},\infty )$ (or
$\mathcal{I} =(-\infty,t_{0})$)
when we study the future (past) singularity behaviour
of an expanding (contracting) universe with a singularity in the past (future),
for instance at $t=0<t_{0}$ ($t=0>t_{0}$). However, since $t$ is just a
coordinate, our study could apply as well
to any interval $\mathcal{I}\subset\mathbb{R}$.

A natural causal assumption for $(\mathcal{V},g)$ is that it is
\emph{globally hyperbolic}. This implies the existence of a time function
on $(\mathcal{V},g)$.
In a globally hyperbolic space,
spacetime splits as above with each spacelike
slice $\mathcal{M}_{t}$ being a Cauchy surface; this means that
each timelike and null curve without end points cuts $\mathcal{M}_{t}$
exactly once \cite{ge70}. In such a spacetime therefore
the future and history of the universe can be predicted or retrodicted from
conditions at the instant of time represented by $\mathcal{M}_{t}$ (cf. \cite{wald}, pp.200-202).

Following the definitions given in \cite{chc02}, we say that a sliced space
has \emph{uniformly bounded lapse} if the lapse function $N$ is bounded below
and above by positive numbers $N_{m}$ and $N_{M}$,
\begin{equation}
0<N_{m}\leq N\leq N_{M}. \label{h1}
\end{equation}
A sliced space has \emph{uniformly bounded shift} if the $g_{t}$
($t\in\mathcal{I}$) norm of the shift vector $\beta$, projection on the tangent
space to $\mathcal{M}_{t}$ of the tangent to the lines $\{x\}\times
\mathcal{I}$, is uniformly bounded by a number $B.$

A sliced space has \emph{uniformly bounded spatial metric} if the
time-dependent metric $g_{t}\equiv g_{ij}dx^{i}dx^{j}$ is
uniformly bounded below for all $t\in \mathcal{I}$ by a
metric $\gamma =g_{t_0}$, that is there exists a number $A>0$ such
that for all tangent vectors $v$ to $\mathcal{M}$ it holds that
\begin{equation}
A\gamma _{ij}v^{i}v^{j}\leq g_{ij}v^{i}v^{j}.  \label{h2}
\end{equation}
\begin{definition}
\label{reg_hyp}
A sliced space $(\mathcal{V},g)$ with uniformly bounded
lapse, shift and spatial metric is called regularly sliced.
\end{definition}

Denoting by $\nabla N$ the space gradient of the lapse $N$, by
$K_{ij}$ 
the extrinsic curvature of $\mathcal{M}_{t}$, and by $|K|^2_{g}$ the product
$g^{am}g^{bn}K_{ab}K_{mn}$, we have the following theorem of \cite{chc02} which
gives sufficient conditions for geodesic completeness:
\begin{theorem}\label{compl_arb}
Let  $(\mathcal{V},g)$ be a sliced space such that the following
assumptions hold:
\begin{description}
\item[C1] $(\mathcal{V},g)$ is globally hyperbolic
\item[C2] $(\mathcal{V},g)$ is regularly sliced
\item[C3] For each finite $t_{1},$ the space gradient of the lapse, $|\nabla N|_{g}$,
is bounded by a function of $t$ which is integrable on $[t_{1},+\infty )$
\item[C4] For each finite $t_{1},$  $|K|_{g}$ is bounded by a function of $t$ which is
integrable on $ [t_{1},+\infty )$.
\end{description}
Then  $(\mathcal{V},g)$ is future timelike and null geodesically
complete.
\end{theorem}
It was later proved in \cite{c03} that in a regularly sliced spacetime,
condition $C1$ is in fact \emph{equivalent} to
the condition that each slice of $(\mathcal{V},g)$ is a complete
Riemannian manifold.

Turning back now to isotropic spacetimes we see from the definition of a sliced
space 
and the form of the RW metric (Eqs. (\ref{frw_metric}), (\ref{f(r)})), that an
isotropic spacetime is a sliced space with $N=1$, $\beta =0$. Following the
Definition (\ref{reg_hyp}) we note that:

\begin{remark}
An isotropic spacetime is regularly sliced iff the scale factor is a bounded
from below function of time.
\end{remark}

We can now reformulate the completeness
theorem \ref{reg_hyp} so as to apply directly to the case of an isotropic
spacetime.
\begin{theorem}[Completeness of RW spacetimes]
\label{frwcomp}
Every globally hyperbolic, regularly sliced RW space
such that for each finite $t_{1}$ the norm of the extrinsic curvature is
integrable on $ [t_{1},+\infty )$, is future timelike and null geodesically
complete.
\end{theorem}
\begin{proof}

Consider a causal geodesic in an isotropic universe. The tangent vector to
these geodesics $u^{a}=dx^{a}/ds$,
where $x^{a}=(x^{0},x^{1},x^{2},x^{3})=(t,r,\theta,\phi)$ and $s$ is the
proper time or an affine parameter measured along a timelike or null geodesic
respectively,
satisfies the geodesic equations
\be
\frac{d^{2}x^{a}}{ds^{2}}+\Gamma_{mn}^{a}\frac{dx^{m}}{ds}\frac{dx^{n}}{ds}=0.
\ee
The $a=0$ component of these equations for this type of universe is
\be
\label{frw_geod}
\frac{d^{2}t}{ds^{2}}+\Gamma_{ij}^{0}\frac{dx^{i}}{ds}
\frac{dx^{j}}{ds}=0,
\ee
where $i,j=1,2,3$ (the rest of the Christoffel symbols $\Gamma_{0i}^{0}$ and
$\Gamma_{00}^{0}$ vanish).

Dividing Eq. (\ref{frw_geod}) by $(dt/ds)^{2}$ we find,
\be
\frac{ds}{dt}\frac{d}{dt}\left(\frac{dt}{ds}\right)+\Gamma_{ij}^{0}
\frac{dx^{i}}{dt}\frac{dx^{j}}{dt}=0.
\ee
Setting $y=dt/ds$ and denoting by $\dot{y}$ the derivative of $y$ with
respect to $t$, the above equation is written in the form
\be
\frac{\dot{y}}{y}=-\Gamma_{ij}^{0}\frac{dx^{i}}{dt}\frac{dx^{j}}{dt}.
\ee
Integrating over the interval $[t_{1},t]$, we find
\be
\label{length}
\ln\frac{y(t)}{y(t_{1})}=
-\int_{t_{1}}^{t}\Gamma_{ij}^{0}\upsilon^{i}\upsilon^{j}dt,
\ee
where we have set $\upsilon^{i}=dx^{i}/dt$.
Since $\Gamma^{0}_{ij}=K_{ij}$
\footnote{\label{footnote}Using Eqs. (\ref{K11})-(\ref{K33}) 
we can verify that $\Gamma_{11}^{0}=\dot{a}a=K_{11},$
$\Gamma_{22}^{0}=\dot{a}af^{2}=K_{22},$
$\Gamma_{33}^{0}=\dot{a}af^{2}\sin^{2}\theta=K_{33}$.}
we can write Eq. (\ref{length}) as follows:
\be
\ln\frac{y(t)}{y(t_{1})}=
-\int_{t_{1}}^{t}K_{ij}^{0}\upsilon^{i}\upsilon^{j}dt.
\ee
Substitution then in Eq. (3.16) of \cite{chc02} gives
\be
\ln\frac{y(t)}{y(t_{1})}
<\int_{t_{1}}^{t}|K|_{g}dt=
\sqrt{3}\int_{t_{1}}^{t}\left|\frac{\dot{a}}{a}\right|dt,
\ee
where in the last equation we have used Eq. (\ref{extr_frw}) from Appendix B.
Using the hypothesis of the integrability of $|K|_{g}$, we see that $y=dt/ds$
is uniformly bounded. This means that $ds/dt$ is bounded away from zero
and hence the length of the geodesic given by
\footnote{As it was shown in \cite{chc02} a future-directed
causal curve $C:s\mapsto C(s)$ can be reparametrised by $t$:
if $n=(-1,0)$ is the timelike normal to $\mathcal{M}_{t}$ then
$g\left(\frac{dC}{ds},n\right)=-\frac{dt}{ds}<0$ and hence $\frac{dt}{ds}>0$.}
\be
\int_{t_{1}}^{\infty}\frac{ds}{dt}dt
\ee
is infinite. This implies that the spacetime is causally geodesically complete.
\end{proof}

Past completeness is derived by exchanging future with past and
$ [t_{1},+\infty )$ with $(-\infty,t_{1}]$.
It follows from the above theorem that a globally hyperbolic and regularly
sliced isotropic spacetime is complete provided that
the spacelike slices $\Sigma_{t}$ do not curve too much as they
are seen by an observer located outside $\Sigma_{t}$ in spacetime, for all
$t\in [t_{1},+\infty )$.

Now since the norm of the extrinsic curvature in an isotropic spacetime is
proportional to the Hubble expansion rate $H(t)$ (see Appendix B, Eq.
(\ref{extr_frw})), we conclude that such a spacetime is complete if it exhibits an
expansion rate that is an integrable function of proper time on $[t_{1},+\infty)$.

In the following Section, we will see that there exist complete FRW models that
have an expansion rate which is non-integrable on $[t_{1},+\infty)$.
This implies that the integrability of the Hubble rate on an infinite proper-time
interval
is a sufficient but \emph{not} a necessary condition for completeness.
In order to prove completeness of these universes we take a different approach in
which we make use of the
Bel-Robinson energy.

As we discuss in Appendix C, the Bel-Robinson energy
is a kind of energy of the gravitational field \emph{projected} in a sense to a
slice in spacetime. In particular, for a RW metric the Bel-Robinson energy at
time t, $\mathcal{B}(t)$, is given by
\be
\label{bre_frw}
\mathcal{B}(t)=\frac{C}{2}\left(|E|^{2}+|D|^{2}\right),
\ee
where $C$ is the constant volume \emph{of} (or \emph{in}, in the case of a
non-compact space) the three dimensional slice at time $t$ and $|E|$, $|D|$ are
the norms of the two electric tensors, namely,
\bq
|E|^{2}&=&3\frac{\ddot{a}^{2}}{a^{2}},\\
|D|^{2}&=&3\left(\frac{\dot{a}^{2}}{a^{2}}+\frac{k}{a^{2}}\right)^{2}.
\eq

It can be proved that a closed or flat, expanding at $t_{\ast}$, FRW universe
that has a bounded Bel-Robinson energy is causally geodesically complete
\cite{skotmg}. Therefore complete FRW models which fail to have a Hubble expansion
rate that is integrable on an infinite proper-time interval may still manage to be
complete provided that they have a bounded Bel-Robinson energy.
In the following Section we will illustrate this result with the analysis of
various models.
\section{Complete FRW cosmologies}
Completeness is not a property only shared by `normal' cosmologies.
Although later in Section 4.3, we meet isotropic universes filled with dark energy
that end their existence in a big-rip singularity, this need not always be the
case. There are many examples of such `exotic' universes containing dark energy
that still manage to evade the big-rip singularity and exist forever.
This is achieved either by assuming that the dark energy component satisfies
simultaneously two different equations of state, or that it
satisfies a generalised equation of state that includes matter ranging from
Chaplygin gasses to perfect fluids.
We shall presently analyse examples of complete universes belonging to both of
these 
categories in turn.

For the first case, consider the flat FRW universe studied in
\cite{sr04} (see also \cite{go_1 04}). In this case the dark energy component
satisfies a `phantom' equation of state of the form
$$p=w\mu, \quad \textrm{with} \quad w=w(t) \quad \textrm{and} \quad w<-1,$$
as well as an equation of state of a \emph{Chaplygin gas}
(cf. \cite{gibbons1}, \cite{gibbons2})
$$p=-\frac{A}{\mu}, \quad \textrm{where} \quad A>0,
\quad \textrm{constant},$$
(similar results will hold when $w\in (-1,-1/3)$ -$k$-essence
models, see, e.g., \cite{go_1 04}). Substituting only the latter equation of
state in the continuity equation, Eq. (\ref{cons_fluid}), and integrating we find
\be
\label{8.4}
\mu^{2}(t)=A+(\mu_{0}^{2}-A)\left(\frac{a_{0}}{a(t)}\right)^{6}.
\ee
We can now take into account the phantom equation of state and reform the above
equation. This is done by calculating, with the use of
both equations of state, the time-dependent parameter $w(t)$,
\be
w(t)=-\frac{A}{\mu^{2}}.
\ee
Substituting this back in Eq. (\ref{8.4}), we find the following form for the
energy density,
\newpage
\be
\mu(t)={\mu_{0}}^{2}{[-w_{0}+(1+w_{0}){(a_{0}/a(t))}^{6}]}^{1/2},
\ee
where $A=-w_{0}{\mu_{0}}^{2}$ and $w_{0}<-1$, which we can now insert in the
Friedmann equation, Eq. (\ref{frw_H}), and derive the following solution
for the scale factor:
\be
a(t)=(Ce^{C_{1}(t-t_{0})})^{1/6},
\ee
where $C$ and $C_{1}$ are positive constants. This model admits
cosmic acceleration, $\ddot{a}>0$, and $a(t)\rightarrow \infty$ as
$t\rightarrow\infty$ and it is geodesically complete. The Hubble parameter takes
the constant value,
\be
H=\frac{C_{1}C}{6}.
\ee
Here also the Bel-Robinson energy remains always finite and hence by the
completeness theorem in \cite{skotmg} we conclude that the model is
geodesically complete.

As we mentioned in the beginning of this section, 
the second technique to trigger the avoidance of a big-rip singularity in models containing dark energy is to
assume that the dark energy satisfies a \emph{generalised} Chaplygin equation
of state. This was done in \cite{go_1 04} for a flat universe. In
particular, the equation of state is here chosen so that
\be
p=-{\mu}^{-\alpha}[C+(\mu^{1+\alpha}-C)^{\alpha/(1+\alpha)}],
\ee
where $C=A/(1+w)-1$, $1+\alpha=1/(1+w)$, and
$\alpha$ is a real parameter. Note that in the limit $w\rightarrow 0$ the gas
reduces to a Chaplygin gas satisfying $p=-A\mu^{-\alpha}$, whereas as
$A\rightarrow 0$ it takes the form of a perfect fluid with equation of state
$p=w\mu$. The scale factor is given by the form
\be
a(t)=(C_{1}e^{-C_{3}(t-t_{0})}+C_{2}e^{C_{3}(t-t_{0})})^{2/3},
\ee
where $C_{1}$, $C_{2}$ and $C_{3}$ are parameters depending on
$a_{0}$ (the initial value of the scale factor), $A$, $w$ and $B$,
with $B$ being a positive integration constant.
At the asymptotic limits $t\rightarrow t_{0}$ and $t\rightarrow \infty$,
we find that $H$ tends to suitable constants, that is it remains finite
on $[t_{0},\infty)$. 
The complete character of this model is 
guaranteed by a clearly bounded Bel-Robinson energy.

Assuming the existence of more conventional matter contents such as radiation
and a cosmological constant 
in a closed FRW universe, we obtain the following solution originally given in
\cite{Harris},
\be
a(t)=a_{i}\left[ 1+\exp(\sqrt{2}t/a_{i})\right]^{1/2},
\ee
where $a_{i}$ is a constant that bounds $a(t)$ from below making this universe
regularly hyperbolic. 
One then finds that $H$ 
becomes zero as $t\rightarrow -\infty$, it is integrable on $(-\infty,0]$
and the Bel-Robinson energy is bounded. Thus
we conclude that this universe is geodesically complete, either using the result
of \cite{skotmg} or alternatively by the completeness theorem \ref{frwcomp}.
\chapter{Classification of singularities}
\section{Introduction}
Up to now we studied situations where geodesic completeness
plays a major role. From now on we shall focus on the other important issue,
namely, that of gravitational singularities.
Spacetime singularities characterise solutions
describing cosmological and gravitational collapse situations in general
relativity and in other metric theories of gravity. Understanding the various
aspects of spacetime singularities constitutes the so-called singularity problem.
Viewed from a mathematical perspective the singularity problem apart from
being appealing because of its fundamental nature, it is also a pressing one,
ever-present
in manifolds and metrics of arbitrary form under general
topological and geometric assumptions
(see Section 3.4). From a physical perspective the singularity problem attracts
attention because of the wide variety of cosmological models in general
relativity and other metric theories of gravity that exhibit singularities that
need to be analysed and understood (see Sections 4.2-4.5 for a detailed analysis
of singular cosmological models in general relativity).

Even though spacetime singularities represent an extensively studied area from
both a mathematical as well as a physical point of view, there are still many
important unresolved issues in this general area. It is essential to
find appropriate tools that will serve to define and analyse the structure and
nature of spacetime singularities. This will hopefully lead to a more
basic understanding of all those qualities that make up the character of
singularities along with a more detailed classification of them.
But there exists a risk in performing such a project: mapping
out only a part of the field comprised of all possible
types of singularities would provide us with misleading conclusions about the
generality of the types found since we would then be shedding light only on
some aspects of the field while others would remain concealed \cite{collins}.
It is thus important to attempt to plot a more or less \emph{complete}
classification scheme. Any such scheme ought to acknowledge both the
mathematical and the physical
qualities of a spacetime singularity. Among the mathematical qualities
the most important one is that of geodesic completeness/incompleteness and it
should therefore play a key role in the planning of a classification.
On the other hand, physical qualities
describe the way matter fields behave on approach to the singularity
and are expected to contribute significantly to our understanding of its nature.

We begin the analysis of singularities in this chapter by focusing on
those types of singularity that may arise in the simplest
class of isotropic spacetimes.
The tools that we use are of a geometric nature. Our geometric approach on one
hand makes the resulting classification scheme easily adjustable
to a further future study of more complicated spacetimes, and on the other 
it helps us find complementary results and generic properties of
the class of spacetimes being analysed.

%
Our attack to a realisation of a complete classification scheme as discussed
above is implemented in two stages. In the \emph{first stage},
we exploit the contrapositive of the completeness theorem Theorem
\ref{compl_arb} in the effort to derive \emph{necessary} conditions for the
existence of singularities in the case of isotropic cosmological models. This
provides us with a classification of singularities that is solely based
on the behaviour of the Hubble expansion rate $H$. 

For various reasons to be given later, the classification depicted in the first
stage does not provide us with all
details necessary for a complete classification scheme referred to above.
We therefore need a complementary tool in order to bring to surface all the other
missing factors and complete our plan.

This leads us to the \emph{second stage} which consists of a classification scheme
based on the introduction of a new invariant, the
Bel-Robinson energy, 
a geometric quantity built out of combinations
of various curvature components that
traces successfully all those required qualities missing from the first
classification.

Apart from being involved in the second stage of our classification scheme,
we show that the use of the Bel-Robinson energy in combination with the notion
of a closed trapped surface induces important information about the dynamical
evolution of isotropic models. In particular, we prove that the existence of a
closed trapped surface can be traced in the dynamical behaviour of the
Bel-Robinson energy. We show by analysing various examples from the literature
that when the existence of a closed trapped surface leads to a singularity the
Bel-Robinson energy also diverges, whereas in the opposite case the Bel-Robinson
energy remains finite and signatures the complete character of the model.

The plan of this chapter is as follows: In Section 3.2, we classify
singularities according to the behaviour of the Hubble rate. In Section 3.3
we analyse the necessity for refining our classification scheme: In
Subsection 3.3.1 we derive a complete classification scheme based
on the possible behaviours of three functions the Hubble rate, the scale factor
and the Bel-Robinson energy.
In Section 3.4 we analyse the interplay of closed trapped surfaces and the
dynamical behaviour of the Bel-Robinson energy:
In the Subsection 3.4.1 we show how the existence of a
closed trapped surface is linked to a possible divergence in the Bel-Robinson
energy and in Subsection 3.4.2 we illustrate this result for various
model universes. Finally, in Section 3.5 we discuss our conclusions.
\section{Hubble parameter singularities}
Theorem \ref{frwcomp} provides us with sufficient conditions for completeness
and we can therefore use its contrapositive to derive necessary conditions for
the existence of singularities. This will also give us an insight into the
character of possible singularities. 
In particular, consider a singular, globally hyperbolic and regularly
sliced
isotropic universe. Then according to Theorem \ref{frwcomp}, there is a finite
time $t_1$ for which $H$ fails to be integrable on the proper time interval
$[t_1,\infty)$. In turn, this non-integrability of the expansion rate $H$ can
be implemented in different ways and we arrive at the following result which
provides first evidence for the nature of finite-time singularities,
based entirely on the behaviour of the Hubble parameter
\cite{ck05}, \cite{iklaoud_2}.
\begin{theorem}\label{sing_H}
Necessary conditions for the existence of finite time singularities in
globally hyperbolic, regularly sliced isotropic universes are:
\begin{description}
\item[$\mathbf{S_{1}}$] For each finite $t$, H is non-integrable on $[t_1,t]$,
or
\item[$\mathbf{S_{2}}$] H blows up in a finite future time, or
\item[$\mathbf{S_{3}}$] H is defined and integrable for only a finite
proper time interval.
\end{description}
\end{theorem}
Condition $S_{1}$ describes a big bang type of singularity when
$H$ blows up at $t_{1}$ since then it is not integrable on any
interval of the form $[t_{1},t]$, $t>t_{1}$ (the condition of regular slicing
is violated in this case but the scale factor is
bounded from above). However, under $S_{1}$ we can have other types of
singularities: Since $H(\tau)$ is integrable on an interval
$[t_1,t]$, if $H(\tau)$ is defined on $[t_1,t]$, continuous on
$(t_1,t)$ and the limits $\lim_{\tau\rightarrow t_1^+}H(\tau)$ and
$\lim_{\tau\rightarrow t^-}H(\tau)$ exist, the violation of
\emph{any} of these conditions leads to a singularity that is not
of the big bang type discussed previously.

Condition $S_{2}$ describes a future singularity at $t_{s}$ ($t_{s}>t_{1}$)
characterised by the divergence of $H$ at $t_{s}$.
Condition $S_{3}$ may lead to a future singularity at $t_{s}$
with $H$ being finite on $[t_{a},t_{s}]$ where $t_{a}\geq t_{1}$
but for this to be a genuine type of singularity, in the sense of geodesic
incompleteness, one needs to demonstrate that the metric is
non-extendible to a larger interval.

Note that these three conditions are not overlapping:
clearly $S_{1}$ and $S_{2}$ are not overlapping with $S_{3}$ but also
$S_{1}$ is not overlapping with $S_{2}$. Indeed $S_{1}$ is not implied by $S_{2}$
for if $H$ blows up at some finite time $t_s$ after $t_1$, then it may still be
integrable on $[t_1,t]$, $t_1<t<t_s$.

There are many examples of singularities met in the literature that belong to
the types predicted by the above theorem; they are
analyzed in Chapter 4 in detail. They usually describe flat isotropic
universes with various components of matter. For example, the phantom dark
energy models of \cite{gonzales} exhibit a future
singularity which falls in the $S_{2}$ category. The standard big bang
singularities as well as the sudden singularity of \cite{ba04} fall in the
$S_{1}$ category. Other sudden singularities (which do not have a blow up of
$H$ at $t=0$ (cf. \cite{noj1})) and the inflation model of \cite{borde} both
fall in the $S_{3}$ category
(for a detailed analysis of these models see Chapter 4).
\section{The $(S,N)$ pairs}
Although the classification based on the Hubble parameter
is a first broad step to clearly distinguish between the various types
of singularities that can occur in isotropic universes, it does not
bring out some of the essential features of the dynamics that may
differ from singularity to singularity. To see this we note that
$S_{3}$ encloses altogether dissimilar types of singularities
each characterised by a distinct pathology, for example in higher order
derivatives of $H$, and as we have already mentioned we also accommodate distinct
singularity types under $S_{1}$.
In addition, condition $S_{2}$ includes both a collapse singularity, where the
scale factor $a\rightarrow 0$ as $t\rightarrow t_{s}$, \emph{and} a
singularity with $a\rightarrow \infty$ as $t\rightarrow t_{s}$.
This type of degeneration is unwanted since it impedes the construction of
a complete classification of singularities.
It is therefore necessary to extend and refine this classification.

We begin by including in our classification the behavior of the scale factor.
That it is only necessary to include this behaviour is seen most
clearly by noticing that only in this way one may consistently
distinguish between initial and final singularities, or between
past and future ones, as it has been repeatedly emphasised by
Penrose in his Weyl curvature hypothesis (see for example
\cite{pen04} pp. 765-769). According to this physical conjecture, the Weyl
curvature which acts as a distorting source for matter is constrained to be
almost zero at initial singularities. In contrast, at final singularities it
diverges to infinity.

Taking into account the behaviour of the scale factor in our singularity
classification scheme results in describing
the various singularity types by a pair $(S_{i}, N_{j})$, where the first
component reflects the behaviour of the Hubble parameter according to Theorem
\ref{sing_H}, while the second component symbolises one of the following three
possible behaviours of the scale factor on approach to the finite-time singularity
at $t_{s}$ namely,
\begin{description}
\item [$\mathbf{N_{1}}$] $a\rightarrow 0$ -\emph{collapse} singularity or

\item [$\mathbf{N_{2}}$] $a\rightarrow a_{s}\neq 0$ -\emph{sudden} singularity or

\item [$\mathbf{N_{3}}$] $a\rightarrow \infty$ -\emph{rip} singularity.
\end{description}

We thus end up with nine possible singularity types given by the pairs
$(S_{i},N_{j})$ where $i,j=1,2,3$.
We further note that in the $(S_{2}, N_{3})$ category we can accommodate
singularities having a Hubble rate which either diverges more rapidly or more
slowly than the scale factor, as well as singularities with these two functions
diverging almost or exactly at the same rate.
Making such a further or more refined distinction between singularities that
otherwise belong to the same
$(S,N)$ type means that we allow for the relative behaviours of the Hubble rate
and the scale factor to emerge as we approach the singularity. This refinement
can be accomplished by introducing in our analysis the following notion of \emph{asymptotic strength}.

First recall that for the relative asymptotic behaviour of two
functions $f,g$ on approach to a point $t_{\ast}$, we say that:
\begin{enumerate}
\item $f(t)$ is \textit{much smaller} than $g(t)$, $f(t)<<g(t)$, if and only if \\
$\lim_{t\rightarrow t_{\ast}} f(t)/g(t)=0$

\item $f(t)
$ is \textit{similar} to $g(t)$, $f(t) \sim  g(t)$, if and only if
$0\neq\lim_{t\rightarrow t_{\ast}} f(t)/g(t)<\infty$

\item $f(t)$ is \textit{asymptotic} to  $g(t)$, $f(t)  \leftrightarrow g(t)$,
if and only if $\lim_{t\rightarrow t_{\ast}} f(t)/g(t)=1$,
\end{enumerate}

There are three specific $(S,N)$ pairs that can be further analysed
according to the possible relative asymptotic behaviours of $H$ and $a$:
\begin{table}[h!b!p!]
\centering
\begin{tabular}{c|c}
\hline
  \textrm{Singularity type} & \textrm{Relative asymptotic strength}\\
  \hline \hline
  $(S_{1},N_{3})$ & $H>>a$, $H<< a$, $H\leftrightarrow a$, $H\sim a$\\
  $(S_{2},N_{3})$ &  $H>>a$, $H<< a$, $H\leftrightarrow a$, $H\sim a$\\
  $(S_{3},N_{2})$ & $H\leftrightarrow a$, $H\sim a$ \\
  \hline
\end{tabular}
\end{table}

We note that while $(S_{1},N_{3})$ and $(S_{2},N_{3})$ singularities can
accommodate all possible relative behaviours between the Hubble rate and the
scale factor, an $(S_{3},N_{2})$ singularity allows only for an asymptotic or
similar relative behaviour between them.

We have so far managed to derive a new classification that consists of nine types
of singularities described by the pairs $(S_{i},N_{j})$ where the first component
reflects the behaviour of the Hubble rate (according to Theorem \ref{sing_H}),
while the second one traces the behaviour of the scale factor. This new
classification can also be further refined with the use of the notion of
asymptotic strength.
\section{Bel-Robinson energy and singularities}
Even though we have at this point taken care of the two main
functions that describe a singularity, we still have
to face another aspect of the problem of classifying singularities.
We have not yet taken into account the relative behaviour of the various matter
components as we approach the time singularities.
A clear picture of the asymptotic behaviour of the matter fields will provide us
with the physical qualities of the singularity
contributing significantly to the understanding of its
nature. In order to describe the behaviour of the various
matter components and include it in our geometric approach, we now exploit the
use of the Bel-Robinson energy
by combining the behaviour of the electric parts of the Bel-Robinson
energy with that of the Hubble parameter in the
following way: Suppose for example that $S_{2}$ holds. Then
from Eq. (\ref{d_3}) we see that for a flat or closed universe $|D|$ necessarily
becomes divergent whereas for an open universe it can be finite since the term
$k/a^{2}$ may, in the case of some collapse singularities, exactly counter-balance
the divergence of the $H^{2}$ term.
These different cases can in turn be combined with a finite or an infinite $|E|$.
Similarly, we can deduce the possible behaviours of $|E|$ and $|D|$ for a singular
universe satisfying $S_{1}$ or $S_{3}$.
Thus the necessary conditions for singularities based on the behaviour of the
Hubble parameter given by Theorem \ref{sing_H} can now be rephrased in terms of
the electric parts of Bel-Robinson energy, to read:
\begin{theorem}
\label{sing_br}
Necessary conditions for null and timelike geodesically incomplete
globally hyperbolic and regularly sliced isotropic universes are that
at a finite time:
\begin{description}
\item [$\mathbf{B_{1}}$] $|E|\rightarrow \infty$ and $|D|\rightarrow
\infty$

\item [$\mathbf{B_{2}}$] $|E|<\infty$ and $|D|\rightarrow
\infty$

\item [$\mathbf{B_{3}}$] $|E|\rightarrow \infty$ and $|D|<\infty$

\item [$\mathbf{B_{4}}$] $|E|< \infty$ and $|D|<\infty$
\end{description}
\end{theorem}
We can now move on to list all possible types of singularities
that are formed in an isotropic universe during its cosmic
evolution and enumerate the possible types that result from the
different combinations of the three main functions in the problem,
namely, the scale factor $a$, the Hubble expansion rate $H$ and
the Bel-Robinson energy $\mathcal{B}.$ These types will by
necessity entail a possible blow up in the functions $|E|$, $|D|$.

If we suppose that a given model has a finite time singularity at
$t=t_s$, then the possible behaviours of the functions in the
triplet $\left(H,a,(|E|,|D|)\right)$ in accordance with Theorems
\ref{sing_H}, \ref{sing_br} are as follows:
\begin{description}
\item [$\mathbf{S_{1}}$] $H$ non-integrable on $[t_{1},t]$ for every $t>t_{1}$

\item [$\mathbf{S_{2}}$] $H\rightarrow\infty$  at $t_{s}>t_{1}$

\item [$\mathbf{S_{3}}$] $H$ otherwise pathological
\end{description}

\begin{description}
\item [$\mathbf{N_{1}}$] $a\rightarrow 0$

\item [$\mathbf{N_{2}}$] $a\rightarrow a_{s}\neq 0$

\item [$\mathbf{N_{3}}$] $a\rightarrow \infty$
\end{description}

\begin{description}
\item [$\mathbf{B_{1}}$] $|E|\rightarrow\infty,\, |D|\rightarrow \infty$

\item [$\mathbf{B_{2}}$] $|E|<\infty,\, |D|\rightarrow \infty $

\item [$\mathbf{B_{3}}$] $|E|\rightarrow\infty,\, |D|< \infty $

\item [$\mathbf{B_{4}}$] $|E|<\infty,\, |D|< \infty $.
\end{description}
The nature of any prescribed singularity is thus described completely by
specifying the components in a triplet of the form \[(S_{i},N_{j},B_{l}),\]
with the indices $i,j,l$ taking their respective values as above.

Note that there are a few types that cannot occur. For instance,
we cannot have an $(S_{2},N_{2},B_{3})$ singularity because that
would imply having $a<\infty$ ($N_{2}$) and $H\rightarrow\infty$
$(S_{2})$, while $3\left(
(\dot{a}/{a})^{2}+k/a^{2}\right)^{2}<\infty$ $(B_{3})$, at $t_{s}$
which is impossible since $|D|^{2}\rightarrow\infty$ at $t_{s}$
($k$ arbitrary).

A complete list of the impossible singularities is model-dependent and
is generally given by triplets  $(S_{i},N_{j},B_{l})$ where the
indices in the case of a flat or a closed universe take the values
$i=1,2$, $j=1,2,3$, $l=3,4$, whereas in the case of an open
universe they take the values $i=1,2$, $j=2,3$, $l=3,4$
(here by $S_{1}$ we denote for simplicity only the big bang case
in the  $S_{1}$ category). We thus see that some singularities
which are impossible for a flat or a closed universe become
possible for an open universe. Consider for example the triplet
$(S_{2},N_{1},B_{3})$ which means having
 $H\rightarrow\infty$, $a\rightarrow 0$ and
$$|D|^2=3\left(\left(\dot{a}/{a}\right)^{2}+k/{a^{2}}\right)^{2}<\infty$$
at $t_{s}$. This behaviour is valid only for some cases of an open universe.
With these exceptions, all other types of finite time singularities can
in principle be formed during the evolution of isotropic matter-filled models
in general relativity or other metric theories of gravity.

This classification allows us to speak consistently of strong, mild or weak
singularities. Note that the standard dust or
radiation-filled big bang singularities fall under the
\emph{strongest} singularity type, namely, the type
$(S_{1},N_{1},B_{1})$, since for example, in a flat universe filled with
dust we have
\bq
a(t)&\propto& t^{2/3}\rightarrow 0,
\quad (N_{1}),\\ H&\propto& t^{-1}\rightarrow\infty, \quad
(S_{1}),\\ |E|^{2}&=&3/4H^{4}\rightarrow\infty,\quad
|D|^{2}=3H^{4}\rightarrow\infty, \quad (B_{1}).
\eq
Our classification scheme is organized in such  a way that the character of the
singularities (i.e., the behaviour of the defining functions)
becomes milder as the indices of $S$, $N$ and $B$ increase. Milder
singularities in isotropic universes are thus expected to occur as
one proceeds down the singularity list.

It is the purpose of this classification to apply both to vacuum as
well as matter dominated models. In fact, the Bel-Robinson energy takes
care in a neat way the matter case. For instance, in fluid-filled
models the various behaviours $B_{i}$, $i=1,2,3,4$, of the Bel-Robinson energy
density can be instead given by asymptotic conditions imposed on the density and
pressure of the cosmological fluid:
\begin{description}
\item [$B_{1}$] $\Leftrightarrow$ $\mu\rightarrow \infty$ and $|\mu+3p|\rightarrow\infty$

\item [$B_{2}$] $\Leftrightarrow$ $\mu\rightarrow \infty$ and $|\mu+3p|<\infty$

\item [$B_{3}$] $\Leftrightarrow$ $\mu<\infty$ and $|\mu+3p|\rightarrow\infty$
$\Leftrightarrow$ $\mu<\infty$ and $|p|\rightarrow\infty$

\item [$B_{4}$] $\Leftrightarrow$ $\mu< \infty$ and $|\mu+3p|<\infty$
$\Leftrightarrow$ $\mu<\infty$ and $|p|<\infty$.
\end{description}
We can translate these conditions to asymptotic behaviours
in terms of $a,H$, depending on the value of $k$, for example,
\begin{enumerate}
\item If $k=0$, $\mu<\infty$ $\Rightarrow$ $H^{2}<\infty$, $a$ arbitrary

\item If $k=1$, $\mu<\infty$ $\Rightarrow$ $H^{2}<\infty$ and $a\neq 0$

\item If $k=-1$, $\mu<\infty$ $\Rightarrow$ $H^{2}-1/a^{2}<\infty$.
\end{enumerate}

It is useful at this point to name five particular members of our classification
of singularities
that are of special interest.
\begin{definition}
\label{sing_names}
Let $(S_{i},N_{j},B_{k})$ be a triplet describing the type of a singularity.\\
We call the triplet:

\item[i)] $(S_{1},N_{1},B_{1})$ ($(S_{2},N_{1},B_{1})$)
a past-collapse (future-collapse) singularity.

\item[ii)] $(S_{2},N_{3},B_{1})$ a big-rip singularity.

\item[ii)] $(S_{3},N_{2},B_{3})$ a sudden singularity.

\item[iv)] $(S_{3},N_{2},B_{1})$ a sudden-rip singularity.
\end{definition}
These types of singularities arise in some of the most popular cosmologies
appearing in the recent literature (a detailed analysis of such cosmologies
is given in Chapter 4).
%
%

In order to gain a deeper insight of
the dynamics on approach to the singularity we can make use of the
notion of asymptotic strength defined in the previous Section.
The use of asymptotic strength apart from offering us a clear
view of the way that the three main functions, designating the nature of a
singularity, compete asymptotically,
and helping us decide about how strong one singularity type is compared to
another, it also provides us with a refinement since it distinguishes
singularities that fall in the same type.
%
%
For example, as we have mentioned already all
standard big-bang type singularities met in the literature belong to our
past-collapse type $(S_{1},N_{1},B_{1})$. Finding and comparing the asymptotic strengths of
the singularities that may be present in dust and radiation filled universes we
find that their precise nature is different since the radiation ones are
characterised by the asymptotic strength $a<<H<<(|E|\leftrightarrow |D|)$,
whereas the dust ones have $a<<H<<(|E| \sim  |D|)$. We shall give many
examples of this later in Chapter 4.

\section{Closed trapped surfaces and Bel-Robinson energy}
In this section, we focus on the possible implications of the existence
of a closed trapped surface in an isotropic universe and show how its existence
can be traced in the behaviour of the Bel-Robinson energy. We
also illustrate through various examples how this association
can be used to shed some light on certain aspects of the dynamical evolution of
isotropic models.

A closed trapped surface is a two dimensional spherical surface $S$
in spacetime with the characteristic property that the
expansion $\theta$ of both the incoming and outgoing bundle of null geodesics
orthogonal to $S$ is negative as evaluated on $S$. That is, the light rays that
are emitted perpendicularly from and towards $S$ converge. This is feasible when
the surface $S$ encloses enough matter to attract all the incoming and outgoing
rays of light. Since nothing can travel faster than light, the matter inside $S$
becomes trapped within the boundary of this surface. If moreover the energy
conditions hold (gravity remains attractive), the radius of $S$ will decrease
gradually 
until a singularity is formed \cite{he73} p. 2-3.

In particular, sufficient conditions for the formation of a singularity are given
by the well-known singularity theorems and include the existence of a closed
trapped surface as well as appropriate energy and
causality conditions imposed on the geometrical and topological structure of the
spacetime (cf. \cite{he73}, p. 263-267): 

\begin{theorem}\label{sing_HE1}
Spacetime  $(\mathcal{V},g)$ cannot be null geodesically complete if:
\begin{description}
\item[(1)] $R_{ab}K^{a}K^{b}\geq 0$ for all null vectors $K^{a}$
\item[(2)] there is a non-compact Cauchy surface in $\mathcal{V}$
\item[(3)] there is a closed trapped surface in $\mathcal{V}$.
\end{description}
\end{theorem}
\begin{theorem}\label{sing_HE2}
Spacetime  $(\mathcal{V},g)$ is not timelike and null geodesically complete if:
\begin{description}
\item[(1)] $R_{ab}K^{a}K^{b}\geq 0$ for every non-spacelike vector $K^{a}$
\item[(2)] every non-spacelike geodesic contains a point at which
$K_{[a}R_{b]cd[e}K_{f]}K^{c}K^{d}\neq 0$,
where $K^{a}$ is the tangent vector to the geodesic
\item[(3)] the chronology condition holds on $\mathcal{V}$ (i.e. there are no
closed timelike curves)
\item[(4)] there exists 
a closed trapped surface.
\end{description}
\end{theorem}
We see that the formation
of a closed trapped surface is triggered by the presence
of a strong gravitational field.
A situation where this is possible is during the gravitational
collapse of a star. In particular, when a star of mass twice that of the sun
has burned out all the nuclear fuel which was supporting it against the influence
of its own gravity and starts collapsing, a singularity becomes
inevitable after a closed trapped surface is formed.
A closed trapped surface is believed to have occurred even in the past
of the universe itself. Looking back in time the universe enclosed enough matter
to provoke a past trapped surface that led to a past-collapse singularity in the
beginning of the universe.
\subsection{Existence}
Consider a past-directed bundle of radial null geodesics in an isotropic
universe with metric given by Eqs. (\ref{frw_metric}), (\ref{f(r)}). Then
\be
ds^{2}=\frac{d\theta}{dv}=\frac{d\phi}{dv}=0,
\ee
where $v$ is an affine parameter along these curves.
The geodesic equation
\be
\frac{d^{2}x^{s}}{dv^{2}}+\Gamma_{ab}^{s}\frac{dx^{a}}{dv}\frac{dx^{b}}{dv}=0,
\ee
with $x^{s}=(t,r,\theta,\phi)$, for this type
of geodesics has the following non-zero components:
\bq
\frac{d^{2}t}{dv^{2}}+\Gamma_{rr}^{t}\left(\frac{dr}{dv}\right)^{2}&=&0,\\
\frac{d^{2}r}{dv^{2}}+2\Gamma_{rt}^{r}\frac{dr}{dv}\frac{dt}{dv}&=&0.
\eq
After substitution of the Christoffel symbols we obtain
\bq
\label{nullgeod_t}
\frac{d^{2}t}{dv^{2}}+\dot{a}(t)a(t)\left(\frac{dr}{dv}\right)^{2}&=&0,\\
\label{nullgeod_r}
\frac{d^{2}r}{dv^{2}}+2\frac{\dot{a}(t)}{a(t)}\frac{dr}{dv}\frac{dt}{dv}&=&0.
\eq
Eq. (\ref{nullgeod_r}) can be written in the form
\be
\frac{d^{2}r/dv^{2}}{dr/dv}dv=-2\frac{\dot{a}(t)}{a(t)}dt,
\ee
and integrated out to give
\be
\frac{dr}{dv}=\pm\frac{C}{a^{2}(t)},
\ee
where $C$ is an arbitrary constant and the sign depends on whether the null
geodesics are incoming $(+)$, or outgoing $(-)$.
Inserting this in Eq. (\ref{nullgeod_t}) (with $C=1$),
multiplying by $dt/dv$ and integrating, we find
\be
\frac{dt}{dv}=\pm\frac{1}{a(t)},
\ee
from which we choose the minus sign since we are looking for the
\emph{past directed} family of
null geodesics. The resulting tangent vector field is then
\be
k^{a}=\frac{dx^{a}}{dv}=\left(-\frac{1}{a(t)},\pm \frac{1}{a^{2}(t)},0,0\right),
\ee
which gives
\be
g_{ab}\frac{dx^{a}}{dv}\frac{dx^{b}}{dv}=-\left(\frac{dt}{dv}\right)^{2}+a^{2}(t)
\left(\frac{dr}{dv}\right)^{2},
=0
\ee
as desired. The divergence of the family of geodesics is calculated to be
\be
k^{a}_{;a}=\frac{1}{\sqrt{-g}}\frac{\partial}{\partial x^{a}}(\sqrt{-g}k^{a})=
\frac{1}{a^{3}(t)f^{2}(r)\sin\theta}\left(
\frac{\partial}{\partial t} (-a^{2}(t)f^{2}(r)\sin\theta)+
\frac{\partial}{\partial r} (\pm a(t)f^{2}(r)\sin\theta)\right)
\ee
i.e.,
\be
\label{cts}
k^{a}_{;a}=\frac{2}{a^{2}(t)}\left(-\dot{a}(t)\pm \frac{f'(r)}{f(r)}\right),
\ee
where $f(r)=\sin r,r$, or $\sinh r$ for $k=1,0$, or $-1$ respectively,
\cite{ellis}. A closed trapped surface is then formed if the divergence of
both families of null geodesics becomes negative, that is if
\be
\label{cts}
\dot{a}(t)>\left|\frac{f'(r)}{f(r)}\right|,
\ee
for some values of $r$, $t$.

An interesting fact for our purposes is that the above condition for the
existence of a closed trapped surface has a decisive effect on the electric part $|D|$
of the Bel-Robinson
energy. This is seen clearly if we use Eqs. (\ref{d_2})-(\ref{d_3}) to write
$|D|$ in the form
\be
\label{d_4}
|D|=\frac{\sqrt{3}}{a^{2}(t)}
\left|\left(\dot{a}(t)-\frac{f'(r)}{f(r)}\right)
\left(\dot{a}(t)+\frac{f'(r)}{f(r)}\right)+
\frac{1}{f^{2}(r)}\right|,
\ee
and then use the closed trapped surface condition (\ref{cts}).
In this case, the relation (\ref{cts}) becomes equivalent to the inequality
\be
\label{d_5}
|D|>\frac{\sqrt{3}}{a^{2}(t)f^{2}(r)},
\ee
from which it follows that collapse singularities (as predicted by the existence
of a trapped surface) are characterised by a divergent Bel-Robinson energy.

In the next Subsection we illustrate how some of the results found here help to
determine the evolution of various model universes.
\subsection{Examples}
Our first three examples illustrate the result that in the case when
the CTS condition is combined with the null energy condition in spatially
flat universes, an initial or future singularity is predicted to occur
and this is reflected in the behaviour of the Bel-Robinson energy.
In contrast, our last example \ref{cts_ds} shows that
even though the existence of a closed trapped surface is usually associated with
the formation of a singularity this need not always be the case.
This was recently shown by Ellis in \cite{ellis}
who pointed out that a closed universe can escape the initial singularity
provided that the non-spacelike energy condition does not hold. The null energy
condition jointly with the condition for the existence of a closed trapped surface are not
sufficient to imply a singularity in the case of a closed universe
since Theorem \ref{sing_HE1} applies \emph{only} in the case of an open or a flat
universe.
\newpage
\begin{example}[Graduated inflation]
\end{example}
Consider the flat graduated inflationary model developed in
\cite{ba}, a universe filled with a fluid with an equation of state
$\mu+p=\gamma \mu^{3/4}$ and $\gamma>0$
(for more details on this family of models see Subsection 4.1.3). The null energy
condition\footnote{This condition leads to Eq.
(\ref{nullcondfrw}) which is equivalent to $\mu+p\geq 0$ for the case of a
fluid of density $\mu$ and pressure $p$.}
is then automatically satisfied. The scale factor for this universe is given by
\be
a(t)=\exp\left(-\frac{16}{3^{3/2}\gamma^{2}t}\right),
\ee
(this is derived by setting $\gamma=3/4$ in Eq. (\ref{scale_grad}) of
Subsection 4.1.3), and the inequality Eq. (\ref{cts}) for this case is
equivalent to
\be
\label{gradinfcts}
\dot{a}(t)=\frac{16}{3^{3/2}\gamma^{2}}\frac{1}{t^{2}}
\exp\left(-\frac{16}{3^{3/2}\gamma^{2}t}\right)>\frac{1}{r}.
\ee
Hence we find that the condition for a CTS
becomes valid for large values of $r$ leading to the existence of \emph{past} CTSs.
Therefore the assumptions of the singularity theorem \ref{sing_HE1} are
fulfilled making this model null geodesically incomplete in the past.
\begin{example}[Massless scalar field]
\end{example}
A similar type of singularity is studied in \cite{fo} where we meet a flat model
filled with a massless scalar field (see Subsection 4.1.1). If we consider
a scalar field with $\dot{\phi}^{2}\geq 0$ then the null energy condition
is satisfied.
The inequality (\ref{cts}) for this case reads
\be
\dot{a}(t)=\frac{1}{3t^{2/3}}>\frac{1}{r},
\ee
which is again true for large values of $r$. Then using the singularity
theorem \ref{sing_HE1} we conclude that this model is singular, and by the
Bel-Robinson energy technique we arrive at a
past-collapse singularity of type $(S_{1},N_{1},B_{1})$.
\begin{example}[Sudden singularity]
\end{example}
A different situation is described by Barrow in \cite{ba04} where he considers
a model with scale factor
\be
a(t)=\left(\frac{t}{t_{s}}\right)^{q}(a_{s}-1)+1-
\left(1-\frac{t}{t_{s}}\right)^{n},
\ee
with $1<n<2$, $0<q\leq 1$ and $t_{s}$ arbitrary but fixed from the beginning.
This model experiences a sudden singularity
(see Definition \ref{sing_names} ii) and Section 4.4)
at the future time $t_{s}$ with $a(t)\rightarrow a_{s}$,
$\dot{a}(t)\rightarrow \dot{a}_{s}$ and
$\ddot{a}(t)\rightarrow -\infty$. For this model we find that the non-spacelike
energy condition\footnote{As we saw for a null vector this condition is
equivalent to $\mu+p\geq 0$.
In order to find an analogous inequality for a timelike vector we
multiply the Einstein equations by a unit timelike vector $k^{a}$:
$R_{ab}k^{a}k^{b}=(T_{ab}-(1/2) T g_{ab})k^{a}k^{b}=T_{ab}k^{a}k^{b}+(1/2) T$.
Then $R_{ab}k^{a}k^{b}\geq 0\Leftrightarrow T_{ab}k^{a}k^{b}\geq -1/2 T$.
For a perfect fluid this translates to $\mu+3p=-6\ddot{a}/a\geq 0$. Therefore the
non-spacelike energy condition holds iff $\mu+3p\geq 0$ and $\mu+p\geq 0$.}
is valid because of $\ddot{a}(t)\rightarrow -\infty$,
while for the case of a closed universe there exists a future CTS (the
analogous inequality of (\ref{cts}) is $\dot{a}(t)<-|f'(r)/f(r)|$)
since the relation
\be
\dot{a}(t)=\frac{q}{t_{s}^{q}}{t^{q-1}}(a_{s}-1)+
\frac{n}{t_{s}^{n}}(t_{s}-t)^{n-1}<-|\cot r|,
\ee
is satisfied for $r=\pi/2$, $a_{s}<1$ and $0<t<t_{s}$. The singularity theorem
\ref{sing_HE2} now applies and guarantees that this model is
timelike and null geodesically incomplete in the future, while
$|E|\rightarrow \infty$ and $|D|<\infty$. 
\begin{example}[De Sitter universe]
\label{cts_ds}
\end{example}
The possibility of evading the singularity is feasible in a de Sitter
universe for which $a(t)=A\cosh H t$, $k=1$, $f(r)=\sin r$,
so that the existence condition given by Eq. (\ref{cts}) is satisfied if and
only if
\be
\label{tr.1}
AH\sinh H t>|\cot r|,
\ee
which holds true for any $t>0$ and $r=\pi/2$ \cite{ellis}. This implies the existence of more
than one CTSs: every $2$-sphere ($t,r=$ const.) with
$t>0$ and area greater than
\be
A_{s}=4\pi A^{2}\cosh^{2}Ht\sin^{2}(r_{s}),
\ee
where $\cot r_{s}=AH\sinh H t$, is a CTS. Nevertheless the spacetime is complete.
The timelike energy condition is violated, $(\mu+3p)=-\ddot{a}(t)/a(t)<0$, and so
Theorem \ref{sing_HE2} 
does not apply. However, the null energy
condition (\ref{nullcond})
holds true 
but Theorem \ref{sing_HE1} does not apply because the spatial slices are now
compact. Then we find that since
\be
|E|=\sqrt{3}\left|\frac{\ddot{a}(t)}{a(t)}\right|=\sqrt{3}H^{2},
\ee
and
\be
|D|=\sqrt{3}\left|\frac{\dot{a}^{2}(t)}{a^{2}(t)}+
\frac{1}{a^{2}(t)}\right|=\sqrt{3}
\left|H^{2}\tanh^{2} H t+\frac{1}{A^{2}\cosh^{2} H t}\right|,
\ee
they always remain finite indicating the complete character of the spacetime in
this case (note that the assumptions of the completeness theorem of
\cite{skotmg} which we discussed in the end of Section 2.2 are satisfied).
\section{Discussion}
We started off in this chapter by giving a classification
of the singularities that can occur in an isotropic universe based on the
behaviour of the Hubble parameter. However, such a classification cannot lead
to complete results and we have found it necessary to extend this scheme by
including succeedingly the behaviours of the scale factor and of the Bel-Robinson
energy. In this case, we have found that the resulting
behaviours of the three functions, $H$, $a$ and $\mathcal{B}$
\emph{taken together} exhaust the types of singularities that are
possible to form during the evolution of an isotropic universe.

The resulting classification is described by triplets of the form
$(S_{i},N_{j},B_{l})$ where the $S$ category monitors the
asymptotic behaviour of the expansion rate, closely related to the
extrinsic curvature of the spatial slices, the $N$ category that of the scale
factor, describing in a sense what the whole of space eventually
does, while the $B$ category describes how the various matter fields contribute
to the evolution of the geometry on approach to the singularity. We know
(cf. \cite{chc02,ycb02}) that all these quantities need to be
uniformly bounded to produce geodesically complete universes.
Otherwise, the whole situation can be very complicated and we have
exploited what can happen in such a case in the
relatively simple geometry considered here.

Our scheme not only covers all the recently discovered types
of singularities but it also predicts many possible new ones. For
example, in the case of a flat isotropic universe the
classification of \cite{noj1} provides  us with four main types
of singularities. These four types can be identified with four
particular $(S_{i},N_{j},B_{k})$ triplets of our scheme: Namely the `big
rip' type of \cite{noj1} characterised by $a\rightarrow\infty$,
$\mu\rightarrow\infty$ and $|p|\rightarrow\infty$ at $t_{s}$ is an
$(S_{2},N_{3},B_{1})$ singularity; their `sudden' singularity
described by $a\rightarrow a_{s}<\infty$,
$\mu\rightarrow\mu_{s}<\infty$ and $|p|\rightarrow\infty$ at
$t_{s}$ is an $(S_{3},N_{2},B_{3})$; further, their type
III singularity, namely $a\rightarrow a_{s}<\infty$,
$\mu\rightarrow\infty$ and $|p|\rightarrow\infty$ at $t_{s}$, is
clearly an $(S_{2},N_{2},B_{1})$ type singularity, while type IV
singularities of \cite{noj1} with $a\rightarrow a_{s}<\infty$, $\mu\rightarrow 0$
and $|p|\rightarrow 0$ at $t_{s}$, all belong to the
$(S_{3},N_{2},B_{4})$ category.

Having all possible singularity types in the unified form expounded in
this chapter has the added advantage that we can consistently
compare the various different types as we asymptotically approach the
time singularity. In this case we saw in detail that the relative
\emph{strength} of the
functions describing the singularity type becomes an important
factor in order to distinguish between all possible behaviours.

It is natural to consider the extension of our classification
scheme in the context of more general classes of cosmological spacetimes.
We believe that an analysis of this more complicated case is still
feasible using the techniques of the present work.
In particular, the first step towards that direction would be to examine
the family of homogeneous but anisotropic universes. The nature
of a singularity arising in such a universe would be given by a
triplet analogous to that of the isotropic case but now the number of
possible types of singularities will increase. This will result from the
existence of more than one scale factors describing the model and also from the
fact that the magnetic parts of the Bel-Robinson energy will not be zero as they
were in the isotropic case studied here.

A further step in that direction would be to examine inhomogeneous spacetimes
where the space gradient of the lapse will be a function not only of the
time but also of the \emph{space} variables and hence it will not
vanish automatically. The nature of a singularity would then be given by
a quadruplet with the new component describing the possible behaviours of the
space gradient of the lapse. These behaviours are predicted from the
contrapositive of the completeness theorem \ref{compl_arb} and they include
both a diverging $|\nabla N|_{g}$ in a finite time (giving rise to a
possible \emph{blow-up lapse singularity}) as well as a $|\nabla N|_{g}$
which will be finite only
for a finite interval of proper time corresponding to a
\emph{sudden lapse singularity}. The analysis of all these cases lies beyond the
work in the present thesis.
\chapter{Singularities in isotropic cosmologies}
\section{Introduction}
In the previous chapter we derived a complete classification of spacetime
singularities that are in principle possible to arise in isotropic cosmologies. The purpose
of this chapter is to look for physical circumstances that give rise to four
particularly important members of our classification given by
Definition \ref{sing_names}; ideally we would be content with prototypes of
cosmologies with every possible singularity type that is predicted by our
classification scheme.

The analysis of the various cosmologies appearing in this chapter is organised
as follows. In
Section 4.2 we study collapse singularities that arise in
recently studied cosmologies containing a massless scalar field (Subsection 4.2.1),
a scalar filed with an exponential potential (Subsection 4.2.2), or a graduated
inflationary potential (Subsection 4.2.3). Section 4.3 provides an
analysis of big-rip singularities which are met in models containing dark energy
in the form of phantom fields (Subsection 4.3.1), or Chaplygin gases
(Subsection 4.3.2).
This type of singularity offers a new possibility of a violent end for the
universe towards which every bound existing structure gets `ripped apart'.
In Section 4.4, we explore the character of sudden singularities that may arise in
cosmologies filled with a fluid having either \emph{unconnected} pressure and
density (Subsection 4.4.1), or alternatively, obeying an equation of state
(Subsection 4.4.2).
In Subsection 4.4.3, we analyse models that exhibit milder than sudden
singularities. Finally, in Section 4.5, we discuss sudden-rip type
singularities.

In each one of these sections we prove theorems giving necessary and sufficient
conditions for the occurrence of the various types of singularity that are
analysed.
We expect the proofs of all these theorems to be quite
straightforward for we have now \textit{already} identified  the
type of singularity that we are looking for in accordance to our
classification. Proving such theorems \emph{without} this knowledge
would have been a problem of quite a different order.

The usefulness of these theorems lies in that they
answer the question of whether or not the behaviours met in known
cosmological models described by exact solutions (which as a rule
have a \textit{flat} spatial metric ($k=0$)) continue to be valid
in universes having nonzero values of $k$, or described by solutions
which are more general than being exact in the sense that some or all of the
arbitrary constants present in a solution remain arbitrary.
We will see that whenever the behaviour met in flat models passes over to curved
ones
either the curvature term in the Friedman equation turns out to be
subdominant compared to the density term, or, in any case, cannot significantly
alter the behaviour of the Hubble parameter $H$.
%

It is worth mentioning that all the different example cases of cosmological
models analysed in this Chapter have been a primary source of motivation for
finding more general theoretical criteria which determine their evolution
and mark their basic characteristics.
%
\section{Collapse singularities}
Following the Definition \ref{sing_names} given in the previous chapter,
by past or future collapse singularities we mean those that arise in
solutions with a scale factor tending to zero at an initial time (big-bang time),
or within a finite future
time (big-crunch time) respectively, while the Hubble expansion rate and the
electric parts of the Bel-Robinson energy diverge. Originally, collapse
singularities were discovered in simple FRW universes filled with dust or
radiation (cf. \cite{wald}, pp. 96-101). Here, we will instead focus on \emph{recently} analysed models that
involve the presence of a scalar field (cf. \cite{fo},
\cite{mel}, \cite{ba}, \cite{ba06}).

Scalar fields play a substantial role in cosmology. In the framework
of inflation, scalar fields having a self-interaction potential are
used to model the matter content of the early universe and their dynamical
evolution monitors the mechanism of inflation (cf., e.g., \cite{early},
pp. 261-317). However, the role of scalar fields
is not limited to early universe scenarios. They are also crucial in the
present phase of evolution of the universe since scalar fields
actualise a possible dark-energy component that is believed to drive the
presently observed accelerating
expansion of the universe (cf. \cite{gonzales}). Furthermore, in grand unified theories scalar fields
are met in the spontaneous symmetry breaking (cf. \cite{early}, pp. 195-255,
\cite{coleman}) and other fundamental mechanisms proposed within these theories
(cf. \cite{early}).

In a homogeneous universe, scalar fields are functions of the proper time only.
The energy-momentum tensor of a scalar field is equivalent to that of a perfect
fluid with density $\mu_{\phi}$ and pressure $p_{\phi}$ given by (cf. Appendix D),
\be
\mu_{\phi}=\frac{\dot{\phi}^{2}}{2}+V(\phi),\quad
p_{\phi}=\frac{\dot{\phi}^{2}}{2}-V(\phi).
\ee
The relation between the pressure and the energy density of the scalar field is
therefore 
$$ p_{\phi}=\mu_{\phi}-2V,$$
which is not an equation of state in the usual sense \cite{mad}.

However, in the special case of a massless scalar field, $V(\phi)=0$,
we have an equation of state of the form $p_{\phi}=\mu_{\phi}$, that is
it behaves as stiff matter.
If on the other hand, the energy of the scalar field is purely potential, or
dominates over the kinetic energy, then the scalar field is approximately
described by the equation of state $p_{\phi}=-\mu_{\phi}$, and behaves therefore
like a cosmological constant. Within the framework of inflationary cosmology for
example, the energy of the scalar field becomes
\emph{mainly} potential as soon as the field enters the slow-roll phase. This
means that $\dot{\phi}^{2}\ll V(\phi)$, or equivalently $\mu_{\phi}+3p_{\phi}<0$, and
hence the dominant energy condition is violated in this case.
Therefore the singularity theorem \ref{sing_HE1} does not apply
and the initial singularity is not an inevitable consequence of inflationary
cosmologies.
Nevertheless, there exist interesting, non-trivial scalar field
cosmologies that cannot avoid the initial singularity in this way and in the
next subsections we shall focus on the analysis of such cosmologies.
\subsection{Massless scalar field}
An example of a singular scalar field cosmology that possesses a past-collapse
singularity can be found in \cite{fo}. This model is a flat FRW
universe with a massless scalar field, and
the Friedman and continuity equations read
\be
3H^{2}=\frac{1}{2}\dot{\phi}^{2}, \quad
\ddot{\phi}+3H\dot{\phi}=0,
\ee
from which we readily find:
\be
\label{fo_1}
H=\frac{1}{3t}, \quad
\phi=\pm\sqrt{\frac{2}{3}}\ln\frac{t}{c}.
\ee
Indeed, this universe emerges from a past-collapse singularity
at $t=0$, since
$a\propto t^{1/3}\rightarrow 0$, 
$H\rightarrow \infty$, while 
$|E|^{2}=\dot{\phi}^{4}/3\rightarrow\infty$ and
$|D|^{2}=3H^{4}\rightarrow\infty$. 

In \cite{fo} Foster also found a condition that should be satisfied so that
solution (\ref{fo_1}),
approximates that with an arbitrary potential $V(\phi)$ in the limit
$t\rightarrow 0$. This condition is derived by simply substituting the solution
given by Eq. (\ref{fo_1})
in the Friedmann equation that is satisfied by a scalar field with
potential $V(\phi)$, i.e. $3H^{2}=V(\phi)+\dot{\phi}^{2}/2$, and equating the
resulting terms. One then finds that
\be
\lim_{t\rightarrow 0}t^{2}V\left(\pm\sqrt{\frac{2}{3}}\ln\frac{t}{c}\right)=0,
\ee
or
\be
\lim_{\phi\rightarrow\pm\infty}e^{-\sqrt{6}|\phi|}V(\phi)=0.
\ee
Hence the solution given by Eq. (\ref{fo_1})
approximates the solution for a scalar field with potential $V(\phi)$ when the
potential decreases slower than $e^{-\sqrt{6}|\phi|}$.
The following theorem implies that this type of singularity may also arise
in more general, curved universes \cite{iklaoud_3}, \cite{iklaoud_4}.

\begin{theorem} \label{6}
A necessary and sufficient condition for the occurrence of a past-collapse
singularity of type $(S_{1},N_{1},B_{1})$ at time $t_{1}$ in an isotropic universe with a
massless scalar field is that $\dot{\phi}\rightarrow\infty$ at $t_{1}$.
\end{theorem}
\begin{proof}
We first prove necessity.
Integrating the continuity equation
$
\ddot{\phi}+3H\dot{\phi}=0,
$
we find that $\dot{\phi}\propto a^{-3}$. Since
$\dot{\phi}\rightarrow\infty$, we have that $a\rightarrow 0$. Because
\be
H^{2}=\frac{\mu}{3}-\frac{k}{a^{2}}\rightarrow\infty,
\quad \textrm{as} \quad t\rightarrow t_{1},
\ee
$H$ becomes unbounded at $t_{1}$. In addition, since
\be
|D|^{2}=\frac{\mu^{2}}{3}=\frac{\dot{\phi}^{4}}{12}\rightarrow\infty,
\ee and
\be
|E|^{2}=\frac{1}{12}(\mu+3p)^{2}=\frac{\dot{\phi}^{4}}{3}\rightarrow\infty,
\ee
as $t\rightarrow t_{1}$, both $|D|$ and $|E|$ diverge there.
Conversely, assuming a past-collapse singularity of type $(S_{1},N_{1},B_{1})$
at $t_{1}$, we have from condition $B_{1}$
($\Leftrightarrow \mu\rightarrow\infty, \, |\mu+3p|\rightarrow \infty$) that
$\mu\rightarrow\infty$ and so $\dot{\phi}^{2}\rightarrow\infty$ as
$t\rightarrow t_{1}$.
\end{proof}

It follows that the asymptotic strength of the singularity in this model
(described by both the exact solution and the above theorem) is
given by: $a<<H<<(|E| \sim  |D|)$.
\subsection{Multiple exponential potential}
Since there is not as yet any unique well-established fundamental theory
dictating and predicting an
exact form of the potential $V(\phi)$, there are many possible candidate
potentials currently studied depending on the underlying theory they intend to
describe.
For example, in \cite{mel} the authors consider a scalar field with the `multiple'
exponential potential of the form
$$V(\phi)=W_{0}-V_{0}\sinh(\sqrt{3/2}\phi),$$
where $W_{0}$ and $V_{0}$ are arbitrary constants
The scalar field inhabits a flat FRW universe
containing also a separately conserved pressureless fluid. In this context, the
field equations become
\bq
\label{fr_mel}
3H^{2}=\left(\frac{\dot{\phi}}{2}+V(\phi)+\mu\right), \\
\label{ray_mel}
2\dot{H}+3H^{2}=-\left(\frac{\dot{\phi}}{2}-V(\phi)\right),\\
\ddot{\phi}+3H\dot{\phi}+V'(\phi)=0,\\
\dot{\mu}+3H\mu=0.
\eq

Following \cite{mel}, we split the scale factor in a product
form by setting ${a}^{3}=xy$, while the scalar field is
written as $\phi=\sqrt{2/3}\log (y/x)$ with $x$ and $y$ positive.
Setting
\bq
x&=&C[\exp\chi_1\cos\chi_{2}+\exp (-\chi_{1})\cos\chi_{2}], \\
y&=&C[\exp\chi_1\sin\chi_{2}+\exp (-\chi_{1})\sin(-\chi_{2})]
\eq
(here, $\chi_1 =w_{1}(t-t_{0}), \chi_{2}=w_{2}(t-T_{0})$,
$C>0$, $t_{0}$ is an arbitrary constant, $T_{0}$ is the
`initial' time and $w_{1}$, $w_{2}$ positive parameters such that
$w_{1}^{2}-w_{2}^{2}=3/4 W_{0},\,\,
2w_{1}w_{2}=3/4 V_{0}$),
we find that for large and positive values of the time parameter $t$ the scale
factor becomes
\be
a=C^{2/3} \exp (2/3 \chi_{1})\cos ^{1/3}
\chi_{2}\sin^{1/3}\chi_{2},
\ee
which is obviously divergent.

At the finite time value $T_{s}=\pi/(2w_{2})+T_{0}$,
we have that $a\rightarrow 0$, $H\rightarrow w_{2}(\cot\chi_2 -\tan\chi_2 )/3
\rightarrow\infty$,
and consequently we find that
\be
|D|\rightarrow \infty,
\ee
and also
\be
|E|\rightarrow |2w_{2}(2(\cot\chi_{2}-\tan\chi_{2})-
\frac{w_{2}}{9}(\cot^{2}\chi_{2}+\tan^{2}\chi_{2}))|\rightarrow\infty.
\ee
This model therefore serves as a nice example that illustrates the structure of
the future-collapse type singularity predicted earlier in our classification.
\subsection{Graduated inflation}
In the original inflationary scenario the equation of state $p=-\mu$ that
describes the scalar field during the slow-roll period is eventually responsible for
an accelerating exponential expansion. However, milder forms of solutions, for
example a power-law expansion $a\propto t^{p}$ with $p>1$, can still suffice
for inflation.
Different forms of solutions in which the expansion follows a power-law with an
arbitrary in general exponent lead to the so-called \emph{intermediate}
or \emph{graduated-inflationary} universes (cf. \cite{ba} and references therein). These
models were first studied by Barrow who introduced an equation of state of the
form
\be
\label{grad_ifl}
\mu+p=\gamma\mu ^{\lambda},
\ee
where $\lambda$, $\gamma$ are constants and $\gamma\neq 0$,
in a flat model. The set of solutions that Barrow found (which naturally
depended on the $\lambda$ and
$\gamma$ parameters) included not only the traditional power-law and
exponential inflationary solutions but also new ones. In these solutions the
scale factor is an exponential of an arbitrary power of the time and in some
cases describes a universe that begins
with a past-collapse singularity. In the original references, however,
no complete discussion of the structure of singularities predicted in these models is
included and so we give below such an analysis.

To see how a past-collapse singularity develops in the history of these
cosmologies we substitute the equation
of state (\ref{grad_ifl}) in the continuity equation and integrate to find
\be
\mu=\frac{(3\gamma (\lambda -1))^{1/(1-\lambda)}}{(\ln(a/a_{0}))^{1/(\lambda -1)}},
\ee
where $\lambda\neq 1$ and $a_{0}$ is a constant. Using then the Friedmann equation
$3H^{2}=\mu$, we obtain the scale factor in the following explicit form \cite{ba}:
\be
\ln\left(\frac{a}{a_{0}}\right)=
\frac{3^{\lambda/(1-2\lambda)}\gamma^{1/(1-2\lambda)}
(\lambda-1/2)^{2(1-\lambda)/(1-2\lambda)}}{\lambda-1}
t^{2(1-\lambda)/(1-2\lambda)},
\ee
with $\lambda\neq 1,1/2$. Setting
\bq
A&=&\frac{3^{\lambda/(1-2\lambda)}\gamma^{1/(1-2\lambda)}
(\lambda-1/2)^{2(1-\lambda)/(1-2\lambda)}}{\lambda-1}, \\
\label{infparam}
f&=&2(1-\lambda)/(1-2\lambda),
\eq
we can write the scale factor in the form
\be
\label{scale_grad}
a(t)\propto \exp(A t^{f}).
\ee
Considering the case that $1/2<\lambda <1$ and $\gamma>0$, we have that $A<0$ and
$f<0$. This means that a past-collapse singularity
develops as $t\rightarrow 0$, since then $a\rightarrow 0$,
$H=Af t^{f-1}\rightarrow \infty$ and
consequently
\be
|D|^{2}=3H^{4}\rightarrow\infty,
\ee
and also $\mu\rightarrow \infty$,
which in turn implies that
\be
|E|^{2}=\frac{1}{12}(-2\mu+3\gamma\mu^{\lambda})^{2}\rightarrow\infty.
\ee
The asymptotic strength of this singularity is therefore:
$a<<(|E|\leftrightarrow |D|)<<H$.
Actually this type of behaviour is found to occur more generally due to the
following result.

%
\begin{theorem} \label{5}
A necessary and sufficient condition for a past-collapse
singularity \\ $(S_{1},N_{1},B_{1})$ to occur at time $t_{1}$ in an open or flat
universe filled with a fluid with equation of state
$p+\mu=\gamma\mu^{\lambda}$, $\gamma>0$ and $\lambda<1$, is that
$\mu\rightarrow\infty$ as $t\rightarrow t_{1}$.
\end{theorem}
\begin{proof}
The continuity equation Eq. (\ref{cons_fluid}) gives directly the form of
the scale factor,
\newpage
\be
\label{scale}
a=a_{0}\exp{\left(\frac{\mu^{-\lambda+1}}{3\gamma(\lambda-1)}\right)},
\ee so that $a\rightarrow 0$ as $t\rightarrow t_{1}$. Since

\[ H^{2}=\frac{\mu}{3}-\frac{k}{a^{2}}=\frac{\mu}{3}-
k{a_{0}}^{-2}\exp{\left(\frac{-2\mu^{-\lambda+1}}{3\gamma(\lambda-1)}\right)}>0,
\]
\[ |E|^{2}=\frac{1}{12}(-2\mu+3\gamma\mu^{\lambda})^{2}, \]
\[|D|^{2}=\frac{\mu^{2}}{3},\]
we see that as $t\rightarrow t_{1}$ $H$, $|E|$ and $|D|$ diverge
provided that $k=0$ or $k=-1$. The converse is straightforward.
\end{proof}

As it was pointed out in \cite{ba} the matter content in a graduated
inflationary universe can be reinterpreted as
a scalar field with the following potential:
\be
\label{potential}
V(\phi)=3\left(Af\left(\frac{f-1}{8A(1-f)}\right)^{(f-1)/f}\phi^{2(f-1)/f}\right)^{2}
\\
-Af(1-f)\left(\frac{f-1}{8A(1-f)}\right)^{(f-2)/f}\phi^{2(f-2)/f},
\ee
where
\be
\phi=\left(\frac{8A(1-f)}{f}\right)^{1/2}.
\ee
and where $f$ is given by Eq. (\ref{infparam}).
It is worth noting here that it is only recently that this model has regained
interest since, as it was shown in \cite{ba06}, a specific choice for the
$f$ parameter offers a good fit to the WMAP observations \cite{spergel} within
the framework of the slow-roll approximation.

In order to illustrate this result,
let us briefly review the equations of the slow-roll approximation. For large
$\phi$ the potential (\ref{potential}) becomes
\cite{ba06}
\be
V(\phi)=\frac{48A^{2}}{(\beta+4)^{2}}(2A\beta)^{\beta/2}\phi^{-\beta},
\ee
where $\phi=(2A\beta t^{f})^{1/2}$ and $\beta=4(f^{-1}-1)$.
In the slow-roll approximation it is assumed that $\dot{\phi}^{2}\ll V$ and
$\ddot{\phi}\ll V'$ restricting the field equations to the form
(cf. \cite{liddle}, \cite{lyth0})
\bq
\label{slow_roll1}
3H^{2}&\simeq&V\\
\label{slow_roll2}
3H\dot{\phi}&\simeq&-V'.
\eq
Inflation then occurs if
\be
\label{inflation}
\frac{\ddot{a}}{a}=\dot{H}+H^{2}>0\Leftrightarrow -\frac{\dot{H}}{H^{2}}<1.
\ee
Differentiating Eq. (\ref{slow_roll1}) with respect to $t$ and using Eq.
(\ref{slow_roll2}) we find that
\be
\dot{H}=-\frac{(V')^{2}}{18 H^{2}},
\ee
so that
\be
-\frac{\dot{H}}{H^{2}}=\frac{1}{2}{\left(\frac{V'}{V}\right)} ^{2},
\ee
while condition (\ref{inflation}) holds if and only if the \emph{first slow-roll
parameter $\epsilon$}, defined by \cite{liddle}, \cite{lyth0}
\be
\label{epsilon}
\epsilon\equiv\frac{1}{2}\left(\frac{V'}{V}\right)^{2},
\ee
satisfies the condition
\be
\label{epsilon}
\epsilon \ll  1.
\ee
If we further differentiate Eq. (\ref{slow_roll2}) and use Eq. (\ref{epsilon})
together with $\ddot{\phi}\ll V'$, we conclude that the
\emph{second slow-roll parameter $\eta$}, defined by \cite{liddle},
\cite{lyth0}
\be
\eta\equiv \frac{V''}{V},
\ee
satisfies the condition
\be
|\eta|\ll 1.
\ee
However, as it is stated in \cite{Lyth1}, it is not necessary to demand that
the derivative of an approximation equation is \emph{itself} a valid
approximation equation and so this last condition can be regarded as independent
from the first one given by Eq. (\ref{epsilon}).
In the intermediate-inflationary model of \cite{ba06} these parameters are shown
to be $\epsilon=\beta^{2}/2{\phi}^{2}$ and $\eta=\beta(1+\beta/2)/\phi^{2}$.
As it is further shown in \cite{ba06}, the resulting values of two very important
spectral indices that are closely related to these slow-roll
parameters come in good agreement with their corresponding values measured by
observations, provided that $\beta=2$, i.e, $f=2/3$ and hence
$a(t)=\exp(At^{2/3})$.
A graduated-inflationary model with these parameters and scale factor
offers therefore an appropriate theoretical framework within which the data from
the WMAP observations is explained.
\section{Big-rip singularities}
Let us now move on to the discussion of a less standard type of model,
nevertheless one which is singular in a way predicted by our scheme.
The notion of dark matter was introduced in cosmology in order to explain a series
of observations in the area of galactic dynamics (cf. \cite{peebles} pp. 417-424).
These observations
showed that the velocities of galaxies in clusters are so high
and the luminous masses enclosed within the galaxies so low that these
structures ought to have been ripped apart. A similar situation holds for
the case of spiral galaxies. In particular, it was found in \cite{Albada}
that at a distance $r$
outside the central luminous part of such a galaxy the rotational velocity
$v(r)$ does not follow the anticipated third law of Kepler
$v(r)\propto r^{-1/2}$ but is instead measured to be constant. This leads to
the counter-intuitive fact that the mass at the outskirts of such a galaxy grows
proportionally to distance (cf. \cite{peebles} pp. 46-47). Such a pattern of matter distribution could be
attributed to the existence of a halo surrounding these galaxies and consisting
of low-luminosity matter, thereby called dark matter, that contributes the
missing amount of matter needed to fit the observed profile of these galaxies
(see \cite{peebles} pp. 45-53).

Dark matter is also considered in another context. Observations strongly indicate
that the value of $\Omega$, the ratio of the present density of the universe
to the critical density $\mu_{c}=3H^{2}_{0}/8\pi G$ ($H_{0}$ being the present
value of $H$), is unity. However, the luminous
mass of the universe is currently measured to have a value of $\Omega_{Lum}$
only around $0.04$ \cite{padma}. On the other hand, measurements of the amount
of dark matter that is believed to exist in clustered areas of the universe give
$\Omega_{Dm}\simeq 0.26$ \cite{padma}. This missing part, nearly $70\%$ of the
critical density, is contributed by an unclustered component the dark energy which has a
ratio of pressure to density $w=p/\mu<-0.78$ \cite{padma}.  
Taking this into account we see that dark matter and dark energy may play
an important role in the understanding of the composition of the universe.
Even though the nature of dark energy remains largely a mystery it is an issue
worth examining 
as it offers us the opportunity to discover new and exciting
possible fates of model universes with such forms of material content.

The simplest source of dark energy can be described by simply postulating the
existence of a cosmological constant $\Lambda$ which satisfies a perfect fluid
equation of state, $p=w\mu$ with $w=-1$ \cite{weinberg}. However, the consideration of a
cosmological constant $\Lambda$ in this context leads to the so-called
cosmological constant problem - the \emph{expected} value of the cosmological
constant from quantum field theories is some 120 orders of magnitude larger
than the observed one \cite{cald03}.
Another candidate for dark energy is the \emph{quintessence}, a
scalar field which satisfies the equation of state for $-1<w<-1/3$
(cf. \cite{k1}-\cite{k6}).
Because recent observations indicate that the universe is not
only expanding but is also accelerating a lot of attention has been drawn to
use such matter components (or even more exotic ones,
cf. \cite{ph_1}-\cite{ph_42}) in an effort to explain the
observed acceleration \cite{page}\, \cite{spergel}. Note that in order to obtain
an accelerating FRW universe filled with a perfect
fluid equation of state $p=w\mu$, it is necessary to have a value of the
$w$ parameter that is less than $-1/3$ as this clearly follows from the field
equations (cf. \ref{acc_frw}) below).

Cosmological observations do not exclude the fact that the $w$ parameter be less
than $-1$ \cite{cald03}. An entity satisfying such an exotic equation of state is dubbed
\emph{phantom dark energy}; this is a becoming name for the description of something
that is physically unorthodox but which nevertheless seems to provide an
explanation to the observations \cite{cald99}. To begin with, phantom dark energy immediately violates the dominant
energy condition since $p+\mu=\mu (w+1)<0$. This leaves room for the existence
of wormholes and times machines (see \cite{cald03} for a discussion). Another
interesting albeit unphysical property
of phantom dark energy is that it exhibits a sound speed that it is both
imaginary and superluminal 
\cite{ph_1}.

Despite the strange properties of phantom-dark energy it is still an interesting
type of matter field to study in cosmology
(see Refs. \cite{ph_1}-\cite{ph_42}, \cite{cald99}, for phantom filled cosmological models)
since its presence influences the future of the universe through
the development of a big-rip singularity as is shown in the
next Subsection. In what follows, we will use the term
\emph{phantom cosmology} to refer to any model that includes a phantom-matter
component.
\subsection{Phantom dark energy}

Consider the case of a flat FRW universe with metric
$ds^{2}=-dt^{2}+a^{2}(t)d\bar{x}^{2}$,
filled with a fluid with equation of state $p=w\mu$. The Einstein field equations
for such a universe are
\bq
\label{Frd_Eq_flat}
\mu &=& 3H^{2},\\
\dot{\mu}&+&3H\mu(1+w)=0,\\
\label{acc_frw}
\frac{\ddot{a}}{a}&=&-\frac{1}{6}\mu(1+3w).
\eq
Integrating the continuity equation we have
\be
\mu=Ca^{-3(w+1)},
\ee
where $C$ is the integration constant.
For ordinary matter ($w\geq 0$), or for dark matter with $w>-1$, the energy density
decreases as the universe expands. In contrast, phantom dark energy has a
density that \emph{increases} as the expansion proceeds in an FRW universe
(irrespectively of its spatial geometry). This implies that it will
at first overcome and finally dominate over all components of ordinary matter.

Solving the Friedmann equation (\ref{Frd_Eq_flat}) for $a$, we find
\cite{gonzales}
\be
\label{phantom_scale}
a=\left[{a _{0}}^{3(1+w)/2}+\frac {3(1+w)\sqrt
{A}}{2}(t-t_{0})\right]^{\frac {2}{3(1+w)}},
\ee
where $A=C/3$ and $a _{0}$ is an integration constant. The behaviour
that would have been expected had we
considered only ordinary matter is that $a\rightarrow \infty$ as
$t\rightarrow \infty$, and the Hubble distance $H^{-1}$ to grow
more rapidly than the scale factor (as it follows from Eq. (\ref{phantom_scale})
and Eq. (\ref{extr_ph}) below), which means that more and more galaxies would
become visible asymptotically \cite{gonzales}. In contrast to this situation,
in a phantom dark energy model a big-rip singularity arises within the finite
value of time given by \cite{gonzales}
\be
t_{rip}=t_{0}+\frac{2}{3\sqrt {A}
(|w|-1)a_{0}^{3(|w|-1)/2}},
\ee
and called the \emph{big-rip time}. Note that $t_{rip}$ depends on the initial size of the
universe, $a_{0}$, and the exact value of the $w$ parameter in such a way that
the bigger the values of $a_{0}$ and $w$, the briefer the existence of
the phantom universe.
Indeed, from Eq. (\ref{phantom_scale}) we see that the scale factor becomes
divergent as $t\rightarrow t_{rip}$. Also the Hubble rate and hence the norms the
electric parts of the Bel-Robinson energy become diverge there since we can
obtain the following forms:
\bq
\label{extr_ph}
H^{2}&=&A\left[{a _{0}}^{3(1+w)/2}+\frac {3(1+w)\sqrt
{A}}{2}(t-t_{0})\right]^{-2},\\
|E|^{2}&=&\frac{3}{4}H^{4}(1+3w)^{2},\\
|D|^{2}&=&3H^{4}.
\eq
Eq. (\ref{extr_ph}) implies in particular that the Hubble \emph{distance} $H^{-1}$ now
tends to zero as $t\rightarrow t_{rip}$.
Combining this with the infinite growth of the scale factor
as $t\rightarrow t_{rip}$, we see that this
type of accelerated expansion forces galaxies to disappear beyond the cosmological
horizon rendering the end of the universe dark and rather violent. The latter
attribute follows from the fact that it can be shown (see \cite{cald03}) that as
the big-rip time approaches gravitationally bound structures such as galaxies,
planetary systems or even atoms, become unbound due to the enormous repulsive
pressure.

The following theorem \cite{iklaoud_3}, \cite{iklaoud_4} characterises
the future singularity in phantom cosmologies \textit{irrespective} of the
value of the curvature $k$ and implies that the singularities in such
models can be milder than the standard all-encompassing future-collapse
singularities and have necessarily diverging pressure. This theorem generalises
and unifies previous results (cf. \cite{cald03}, \cite{steph}, \cite{sud_10},
\cite{sud_12}, \cite{gonzales}, \cite{noj1}, \cite{cald99}, \cite{ch04},
\cite{ph_33}).

%
\begin{theorem} \label{2}
Necessary and sufficient conditions for a big-rip
singularity $(S_{2},N_{3},B_{1})$ to occur at the finite future time $t_{s}$
in a universe filled with a fluid with  equation of state
$p=w\mu$, are that $w<-1$ and $|p|\rightarrow\infty$ at $t_{s}$.
\end{theorem}
\begin{proof}
Substituting the equation of state $p=w\mu$ in the
continuity equation $\dot{\mu}+3H(\mu+p)=0$, we have
\be
\label{phantom} \mu\propto a^{-3(w+1)}, \ee and so if $w<-1$ and
$p$ blows up at $t_{s}$, $a$ also blows up at $t_{s}$. Since
\be
H^{2}=\frac{\mu}{3}-\frac{k}{a^{2}},\quad
|D|^{2}=\frac{\mu^{2}}{3}, \quad
|E|^{2}=\frac{1}{12}\mu^{2}(1+3w)^{2}, \ee we conclude that at
$t_{s}$, $H$, $a$, $|D|$ and $|E|$ are divergent.
Conversely, assuming a big-rip singularity $(S_{2}, N_{3}, B_{1})$
at $t_{s}$ in an \textsc{FRW} universe with the equation of
state $p=w\mu$, we have from the $(B_{1})$ hypothesis that
$\mu\rightarrow\infty$ at $t_{s}$ and so $p$ also blows up at
$t_{s}$. Since $a$ is divergent as well, we see from
(\ref{phantom}) that $w<-1$.
\end{proof}

It is interesting to note that the relative strength of this big-rip singularity
is strongly influenced by the choice of the $w$ parameter. In particular, there are
three possible behaviours depending on the given ranges of the $w$ parameter:
\begin{itemize}
\item \textrm{if} \quad $-4/3<w<-1$, \quad
\textrm{then} \quad $H<<(|E|\sim |D|)<<a$,

\item \textrm{if} \quad $-5/3<w<-4 /3$,  \quad \textrm{we have} \quad
$H<<a<<(|E| \sim |D|)$, \quad\textrm{whereas}

\item \textrm{if}\quad $w<-5/3$, \quad\textrm{we have}\quad
$a<<H<<(|E|\sim |D|)$.
\end{itemize}

We note that for all three ranges of $w$, $|E|$ and $|D|$ diverge similarly while
$H$ remains always much smaller compared to these two quantities. On the other
hand, the scale factor diverges less strongly relatively to $H$ and $|E|$, $|D|$
as $w$ decreases below $-1$.
\subsection{Chaplygin gas}
As we saw in the previous Section a phantom equation of state drives the universe
towards a big-rip singularity. In this section we will see that
a big-rip singularity
can still arise in models with
forms of dark energy different from phantoms such as the so-called
\emph{generalised} Chaplygin gases;
these fields play a central role in
dark energy cosmology and specifically
in the unification of 
dark energy and dark matter
(see Refs. \cite{ch_1}-\cite{ch_25}).

For example, consider a flat FRW universe
containing a dark-energy component which at late times satisfies a general
equation of state of the form \cite{steph}
\be
p=-\mu_{d}-A\mu_{d}^{\alpha}
\ee
where $A$ and $\alpha$ are real parameters. Integrating the continuity equation
\be
\dot{\mu}_{d}-3HA\mu_{d}^{\alpha}=0,
\ee
we find that the energy density ${\mu}_{d}$ of the dark component behaves as
follows,
\be
\label{mu_d}
\mu_{d}=\mu_{d0}\left(1+3\bar{A}(1-\alpha)\ln\frac{a}{a_{0}}\right)
^{\frac{1}{1-\alpha}},
\ee
where $\bar{A}=A\mu_{d0}^{\alpha-1}>0$ and $\mu_{d0}$ is an integration constant.
Note that
\be
\frac{d\mu_{d}}{da}=\frac{3\bar{A}\mu_{d0}}{a}
\left(1+3\bar{A}(1-\alpha)\ln\frac{a}{a_{0}}\right)
^{\frac{1}{1-\alpha}-1},
\ee
so that the parameter $\bar{A}$ monitors the behaviour
of the dark energy as the universe expands: if $\bar{A}>0$ the density of the
dark energy increases with expansion if $\bar{A}=0$, the density is constant,
whereas if $\bar{A}<0$, the density decreases \cite{steph}. Recall that for the
occurrence of a big-rip singularity the density needs to grow with expansion and
hence for the study of such a singularity we may focus here in the case
$\bar{A}>0$.

Substituting in the Friedmann
equation (\ref{frw_H}) we find the scale factor in the following implicit form
\be
\left(1+3\bar{A}(1-\alpha)\ln \frac{a}{a_{0}}\right)^
{\frac{1-2\alpha}{2(1-\alpha)}}-
\left(1+3\bar{A}(1-\alpha)\ln \frac{a_{d}}{a_{0}}\right)^
{\frac{1-2\alpha}{2(1-\alpha)}}=\frac{3}{2}\bar{A}(1-2\alpha)C(t-t_{d}).
\ee
Here $B=\sqrt{\mu_{d0}/3}$, $t_{d}$ is the time at which the dark energy
becomes dominant and $a_{d}$ is the corresponding value of the scale factor
$a_{d}=a(t_{d})$.
The Hubble expansion rate is then given by
\be
\label{eqH}
H=B{\left({\left(1+3\bar{A}(1-\alpha)\ln{\frac{a_{d}}{a_{0}}}\right)
^{\frac{1-2\alpha}{2(1-\alpha)}}+
\frac{3}{2}\bar{A}(1-2\alpha)C(t-t_{d})}\right)}^{\frac{1}{1-2\alpha}}.
\ee
We now show that the future behaviour of this model is quite sensitive to the
values of the parameter $\alpha$. Although the case with $\alpha\in (1/2,1)$
leads to the development of a big-rip singularity in finite time,
when $\alpha<1/2$ we are led instead a singularity-free universe.
Indeed, taking $\alpha\in (1/2,1)$ we allow for the dark energy density, the
scale factor and the Hubble rate $H$ all to increase with time until they diverge
at the finite time $t_f$ where
\be
t_{f}=t_{d}+\frac{2}{3\bar{A}(2\alpha-1)C}
{\left(1+3\bar{A}(1-\alpha)\ln{\frac{a_{d}}{a_{0}}}\right)}^
{\frac{1-2\alpha}{2(1-\alpha)}}.
\ee
In addition, the Bel-Robinson energy blows up there since the electric parts
take the forms:
\bq
|E|^{2}&=&\frac{1}{12}(2\mu+3A\mu^{a})^{2},\\
|D|^{2}&=&3H^{4}.
\eq
If instead one
assumes that $\alpha<1/2$, we see from Eq. (\ref{eqH}) that $H$ is
always finite (so that the Bel-Robinson energy is also finite) and therefore the
model is geodesically complete. (The case with $\alpha>1$ leads to an
$(S_{2},N_{2},B_{1})$ singularity since as it follows from (\ref{mu_d})
$\mu_{d}$ and consequently
$p_{d}$, $H$, $|E|$ and $|D|$ all diverge at the finite value
$a_{s}=a_{0}e^{1/(3\bar{A}(\alpha-1))}$.)

Summing up the results of this Section, the consideration of these unconventional
dark energy components in the framework of FRW cosmology leads to the possibility
of a big-rip singularity. Even though Theorem \ref{sing_H} predicts such a
singular behaviour in condition $S_{2}$, it is only very recently
that the study of physically motivated models described here has offered
sufficient \emph{physical} grounds for the appearance of such a structure
distinct from the standard future-collapse types considered so far
(cf. \cite{cald03}, \cite{steph}, \cite{sud_10}, \cite{sud_12}, \cite{gonzales},
\cite{noj1}, \cite{cald99}, \cite{ch04}, \cite{ph_33}).
\section{Sudden singularities}
In the previous Section, we saw that the violation of the dominant energy
condition is sufficient for the occurrence of future finite-time
singularities during the expansion of the universe.
But is this also necessary? Can a finite-time singularity develop
during the expansion of the universe \emph{without} violating the
dominant energy condition? If yes, then what are the qualitative differences
between the resulting types of singularities?

Barrow initiated a study of these questions in \cite{ba04}. He found a new type
of singularity that can arise in a universe where matter does not violate the
dominant energy condition. Actually in his example matter
satisfies even the strong-energy conditions $\mu>0$ and $\mu+3p>0$.
This fact implies that the violation of the dominant-energy condition is
sufficient but \emph{not} necessary for the occurrence of a finite-time
singularity during expansion.
Moreover, the nature of the resulting singularity is substantially
different from the big-rip singularities discussed earlier.
In particular, the finite-time singularity of Barrow is a sudden singularity
according to our Definition \ref{sing_names} occurring at an arbitrary
but fixed value of finite time $t_{s}$ at which 
the scale factor, the Hubble rate and consequently $|D|$, all remain finite
but 
$|E|$, or equivalently the acceleration, diverges. This becomes possible
since by the Raychaudhuri equation at $t_{s}$ we have
\be
\label{ray_sud}
\frac{\ddot{a}}{a}=-\frac{p}{2}-\frac{H^{2}_{s}}{2}-\frac{k}{2a^{2}_{s}},
\ee
where $a_{s}\equiv a(t_{s})$, $H_{s}\equiv H(t_{s})>0$, which
implies that a singularity with
$\ddot{a}/a\rightarrow -\infty$ can occur provided the pressure blows up
at $t_{s}$.

In the following Subsections we analyse how a sudden singularity
arises in universes filled with a fluid having $\mu$ and $p$ unconnected
(Subsection 4.4.1), or satisfying an equation of state (Subsection 4.4.2).
In Subsection 4.4.3 we explore model universes that exhibit singularities
that are milder than the sudden ones.
\subsection{Fluids with $p$ and $\mu$ unconnected}
A special exact solution acquiring all the characteristics needed for the
development of a sudden singularity is given in \cite{ba04} and reads
\begin{equation}
a(t)=1+\left(\frac{t}{t_{s}}\right)^{q}
(a_{s}-1)+(1-\frac{t}{t_{s}})^n,
\end{equation}
with $1<n<2, 0<q\leq 1$. For the more detailed analysis of the properties of
this solution given below we follow \cite{ck05}, \cite{iklaoud_2}.
This solution exists smoothly only on the interval
$(0,t_{s})$, with $a_{s}$ and $H_{s}$ and therefore $|D|$ is finite at the
right end. By setting $A=q(a_s-1)/t^q_s, B=n/t_s^n$, we find that
\be
\dot{a}(t)=At^{q-1}+B(t_s-t)^{n-1},
\ee
which means that unless $q=1$, $\dot{a}$ blows up as $t\rightarrow 0$,
making $H$ continuous only on $(0,t_s)$. Also $a(0)$ is finite
and we can extend $H$ and define it to be finite also at $0$, $H(0)\equiv H_0$,
so that $H$ is defined on $[0,t_s]$. However, since $\lim_{t\rightarrow
0^+}H(t)=\pm\infty$, $H$ is non-integrable on $[0,t_s]$, $t_s$ arbitrary and thus
this model universe implements exactly Condition $S1$ of Theorem \ref{sing_H}.

Another basic feature of this solution is that the second derivative of the
scale factor is
\be
\ddot{a}(t)=q(q-1)Bt^{q-2}-\frac{n(n-1)}{t^{2}_{s}(1-t/t_{s})^{2-n}}
\rightarrow-\infty, \quad \textrm{as} \quad t\rightarrow t_{s},
\ee
so that $|E|$, and consequently $p$, diverges at $t_{s}$ as it is expected for a
sudden singularity. Note that $\mu_{s}>0$ for
$q^{2}t_{s}^{-2}(a_{s}-1)^{2}>-k$, and since $(\mu+3p)=-6\ddot{a}/a$
(see Eq. (\ref{accfrw})), we can have $\mu>0$ and $\mu+3p>0$
throughout the evolution of such an expanding universe.

We can further
prove that under the additional assumption $p\geq 0$, a universe
exhibiting a sudden singularity is geodesically incomplete \cite{ck05}.
Firstly, we know (cf. \cite{he73}, pp. 141-2) that
in an FRW universe with $\mu > 0$ and $p\geq 0$, given any vector $X$ at any
point $q$, the geodesic $\gamma (\upsilon)$ which passes through the point
$q=\gamma(0)$ in the direction of $X$ is such that
\begin{itemize}
\item either $\gamma (\upsilon)$ can be extended to arbitrary
values of its affine parameter $\upsilon$, \emph{or}
\item there is some value $\upsilon_{0}>0$
such that the scalar invariant $(R_{\a\b}-1/2
Rg_{\a\b})(R^{\a\b}-1/2Rg^{\a\b})$ is unbounded on $\gamma
([0,\upsilon])$.
\end{itemize}
For a RW spacetime, the invariant
$(R_{\a\b}-1/2 Rg_{\a\b})(R^{\a\b}-1/2Rg^{\a\b})$ is found to be
\begin{equation}
\frac{12k}{{a}^{4}}+\frac{24\dot{a}^{2}k}{a^{4}}+\frac{12\dot{a}^{4}}{a^{4}}+\frac{12\ddot
a}{a ^{3}}+\frac{12\dot{a} ^{2}\ddot{a}}{a^{3}}+\frac{12\ddot{a}^{2}}{a^{2}},
\end{equation}
and since $a\rightarrow a (t_{s})$, $ H(t)\rightarrow H_{s}$,
$p(t)\rightarrow \infty $, $\ddot{a}/ a \rightarrow -\infty$ as
$t\rightarrow t_{s}$, we see that $(R_{\a\b}-1/2
Rg_{\a\b})(R^{\a\b}-1/2Rg^{\a\b})$ is unbounded at $t_s$. Hence, we
find that such a spacetime is geodesically incomplete. Note that this holds true
irrespectively of the exact form of the scale factor, we only use the
qualitative properties that define a sudden singularity
(Definition \ref{sing_names}).

Further, the relative asymptotic strength of a sudden singularity
is immediately implied by the definition of this type of singularity and it is
necessarily of the form $(H \sim  |D| \sim  a)<<|E|$.
\subsection{Fluids with $p$ and $\mu$ connected}
Although Barrow in \cite{ba04} kept the pressure and the density unconnected in his
example of a sudden singularity, Nojiri et al. in
\cite{noj1} focused on finding specific equations of state appropriate for the
emergence of such a structure in a flat universe. An example of such an
equation of state is given in \cite{noj1},
\be
p+\mu=-C(\mu_{0}-\mu)^{-\gamma},
\ee
where $\gamma$ is a positive constant. We are interested below in the asymptotic
behaviour of the resulting solution of the cosmological equations in the limit of
$\mu\rightarrow \mu_{0}$. Integrating the
continuity equation, we find that the scale factor, given by
\be
a=a_{0}\exp \left(-\frac{(\mu_{0}-\mu)^{\gamma +1}}{3C(\gamma +1)}\right)
\ee
($a_{0}$ is an integration constant), remains finite as
$\mu\rightarrow\mu_{0}$. From the Friedmann equation $3H^{2}=\mu$, it follows
that $H$ and therefore $|D|$ must also remain finite. However, in the limit
$\mu\rightarrow\mu_{0}$, the pressure, and hence $|E|$ blow up. Solving the
Friedmann equation for $H$, substituting in the continuity
equation $\dot{\mu}=3HC(\mu_{0}-\mu)^{-\gamma}$, and integrating we find
\be
t=t_{0}+\int\frac{d\mu}{\sqrt{3\mu}C(\mu_{0}-\mu)^{-\gamma}},
\ee
which remains finite (since the integrand decreases to zero) even when
$\mu\rightarrow\mu_{0}$ (recall that $\gamma>0$). This fact leads to
the conclusion that singular behaviour is met within a finite interval of proper
time $t$ and, as
%
the following theorem shows, this type of behaviour continues to be valid
also in curved models.
\begin{theorem} \label{4}
A necessary and sufficient condition for a sudden singularity of type
$(S_{3},N_{2},B_{3})$ to occur at $t_{0}$ in a universe filled with a
fluid with equation of state $p+\mu=-C(\mu_{0}-\mu)^{-\gamma}$,
$\gamma>0$, is that $\mu\rightarrow\mu_{0}$ at $t_{0}$.
\end{theorem}
\begin{proof}
Integrating the continuity equation $\dot{\mu}-3CH(\mu_{0}-\mu)^{-\gamma}=0$,
we find
\be
\label{scale2}
a\propto\exp\left\{{-\frac{{(\mu_{0}-\mu})^{\gamma+1}}{3C(\gamma+1)}}\right\},
\ee
which is finite as $t\rightarrow t_{0}$. Also since
\be
H^{2}=\frac{\mu}{3}-\frac{k}{a^{2}},
\ee
\newpage
\bq
|E|^{2}&=&3\frac{\ddot{a}^{2}}{a^{2}}=
\frac{1}{12}(\mu+3p)^{2}=\frac{1}{12}(2\mu+3C(\mu_{0}-\mu)^{-\gamma})^2, \\
|D|^{2}&=&3\left(H^{2}+\frac{k}{a^{2}}\right)^{2}=\frac{\mu^{2}}{3},
\eq
we see that as $t\rightarrow t_{0}$, $H$ and $|D|$ remain finite
whereas $|E|$ diverges, leading to a sudden singularity.
The converse is immediate.
\end{proof}
\subsection{Milder-than-sudden singularities}
A sudden singularity is a very mild type of singularity in the sense that the only
diverging physical quantity that characterises it is the $|E|$ component of the
Bel-Robinson energy. The only other milder type of singularity that could occur
would be one with a bounded Bel-Robinson energy as this is described by the
$B_{4}$ case in our classification scheme. This can indeed happen.
Consider a flat universe filled with a fluid with equation of state
(cf. \cite{noj1})
\be
p+\mu=-\frac{AB\mu^{2\beta-1}}{A\mu^{\beta-1}+B},\quad
0<\beta<1/2.
\ee
Such a fluid, admits the exact solution
$a=a_{0} e^{\tau^{8/3}}$ when $\beta=1/5$. Then $H=(8/3)\tau^{5/3}$,
$\dot{H}=(40/9)\tau^{2/3}$ and $\ddot{H}=(80/27)\tau^{-1/3}$, and thus
as $\tau\rightarrow 0$, $a$, $\dot{a}$, $\ddot{a}$, $H$, $\dot{H}$
all remain finite whereas $\ddot{H}$ becomes divergent. We can see
that in this model the Bel Robinson energy at
the initial time, $\mathcal{B}(0)$, is finite whereas \textit{its
time derivative} is given by
\be
\dot{\mathcal{B}}(\tau)=3\left[2\frac{\ddot{a}}{a}(\ddot{H}+2H\dot{H})+
4\left(\frac{k}{a^{2}}+H^{2}\right) \left(-\frac{k
H}{a^{2}}+H\dot{H}\right)\right],
\ee
hence $\dot{\mathcal{B}}(\tau)\rightarrow\infty$ at $\tau\rightarrow 0$.
Since the derivative of the Bel-Robinson energy diverges we may interpret
this singularity geometrically as one \emph{in the velocities of
the Bianchi (frame) field}. At $t_s$, the Bianchi field encounters
a cusp and its velocity diverges there \cite{iklaoud_4}. We believe this to be
the mildest type of singularity known in the literature to date.

Another example of a very mild singularity
is given by Borde \emph{et al} in \cite{borde}. These authors
define an averaged-out Hubble function,
$H_{av}$ (given by Eq. (\ref{H_av}) below), along a null or timelike geodesic
and showed that every past-directed null or timelike geodesic satisfying
$H_{av}>0$ throughout the past is not past complete.
To see this consider first
a null geodesic with affine parameter $\lambda$.
Then using Eq. (\ref{frw_geod}) and 
substituting the Christoffel symbols (see 
footnote (\ref{footnote}) in p. 14), we
find that the geodesic equation is given by 
\be
\frac{d^{2}t}{d\lambda ^{2}}+\frac{\dot{a}}{a}
\left(a^{2}\left(\frac{dr}{d\lambda}\right)^{2}+f^{2}a^{2}
\left(\frac{d\theta}{d\lambda}\right)^{2}+a^{2}f^{2}\sin^{2}\theta
\left(\frac{d\phi}{d\lambda}\right)^{2}\right)=0.
\ee
Using the fact that the geodesic is null ($ds^{2}=0$) we have that
\be
\frac{d^{2}t}{d\lambda ^{2}}+
\frac{\dot{a}}{a}\left(\frac{dt}{d\lambda}\right)^{2}=0
\ee
and by integration we find that
\be
d\lambda \propto a(t) dt.
\ee
Setting $d\lambda=(a(t)/a(t_{s}))dt$, the affine parameter can be normalised
so that $d\lambda/dt=1$ for $t=t_{s}$.
The averaged-out Hubble function is defined by
\begin{equation}
\label{H_av}
H_{av}=\frac{1}{\lambda(t_{s})-\lambda(t_{i})}
\int_{\lambda(t_{i})}^{\lambda(t_{s})}H(\lambda)d\lambda.
\end{equation}
Now by considering those null geodesics that satisfy $H_{av}>0$, we find
\begin{equation}
\label{H_avineq}
0<H_{av}=\frac{1}{\lambda(t_{s})-\lambda(t_{i})}
\int_{\lambda(t_{i})}^{\lambda(t_{s})}H(\lambda)d\lambda
=\frac{1}{\lambda(t_{s})-\lambda(t_{i})}
\int_{a(t_{i})}^{a(t_{s})}\frac{da}{a(t_{s})}<\frac{1}
{\lambda(t_{s})-\lambda(t_{i})}.
\end{equation}
This shows that the affine parameter of \emph{every} past-directed null geodesic
having $H_{av}>0$ must take values only in a finite interval,
which implies geodesic incompleteness. A similar proof is obtained for the case
of a timelike geodesic. Hence, all causal past-directed geodesics
obeying the condition $H_{av}>0$ are past incomplete.

Note that for null geodesics, Eq. (\ref{H_avineq}) implies that $H$ is finite but
\emph{only} in a finite time interval (an analogous conclusion holds for a
timelike geodesic) and this is no other than the $S_{3}$ (necessary) condition
for a singularity of Theorem \ref{sing_H}. We therefore conclude that under the
assumptions of \cite{borde}, privileged observers cannot exist in such expanding
universes for an infinite proper time
because had they existed Theorem \ref{frwcomp} would imply that
spacetime ought to be timelike and null geodesically complete.
\section{Sudden-rip singularities}
A sudden-rip singularity arises once we allow not only $|E|$ but also $|D|$
to diverge at a finite time, while the scale factor
remains finite there. This of course implies that $H$, $p$ and $\mu$
become also divergent.
This possibility was explored
in some examples in \cite{noj1}. An appropriate equation of state to
illustrate this type of asymptotic behaviour is to take \cite{noj1}
\be
\label{eos_noj}
p+\mu=-B\mu^{\beta}, \quad \beta>1.
\ee
Integration of the continuity equation $\dot{\mu}-3BH\mu^{\beta}=0$ then gives
directly \cite{noj1}
\be
a=a_{0}\exp\left(\frac{\mu^{1-\beta}}{3(1-\beta)B}\right),
\ee
which is finite everywhere since we have chosen $\beta>1$.
We can find the relation between the proper time and the density
by first substituting Eq. (\ref{eos_noj}) in the continuity equation and then
using the Friedmann equation $3H^{2}=\mu$. We obtain
\be
dt=\int\frac{d\mu}{\sqrt{3\mu}B\mu^{\beta}},
\ee
or equivalently \cite{noj1}
\be
t\sim
t_{0}+\frac{2}{\sqrt{3} B} \frac{\mu^{-\beta+1/2}}{1-2\beta}.
\ee
Therefore, when $\mu\rightarrow \infty$, $t\rightarrow t_{0}$ and so
$|p|$, $H$, $|D|$ and $|E|$ all diverge within the finite time $t_{0}$, meaning
that this is an $(S_{2},N_{2},B_{1})$ singularity. The asymptotic strength of
this singularity is found to be: $a<<H<<|D|<<|E|$.

This behavior is a special case of the
following general result which holds also in curved models.

\begin{theorem} \label{3}
A necessary and sufficient condition for a sudden-rip singularity
(of type $(S_{2},N_{2},B_{1})$) to occur at $t_{s}$ in an isotropic universe
filled with a fluid with equation of state $p+\mu=-B\mu^{\beta}$, $\beta>1$, is
that $\mu\rightarrow\infty$ as $t\rightarrow t_{s}$.
\end{theorem}
\begin{proof}
From  the continuity equation we have
\be
\label{scale1}
a=a_{0}\exp{\left(\frac{\mu^{1-\beta}}{3B(1-\beta)}\right)}, \ee
and so $a\rightarrow a_{0}$ as $t\rightarrow t_{s}$. Since
\be
H^{2}=\frac{\mu}{3}-\frac{k}{a^{2}},\quad
|E|^{2}=\frac{1}{12}(2\mu+3B\mu^{\beta})^{2},\quad
|D|^{2}=\frac{\mu^{2}}{3}, \ee we see that as $t\rightarrow
t_{s}$, $H$, $|E|$ and $|D|$ diverge leading precisely to a
sudden-rip singularity.
The converse is obvious.
\end{proof}

%
\section{Discussion}

We have illustrated in this chapter that our methods of classifying
singularities through the behaviour of the Hubble parameter the scale factor
and the Bel-Robinson energy provide a rigorous basis for a useful
description of the dynamics of a variety of cosmological
models that exhibit finite-time singularities. Even though these
model universes come from completely
different motivations yet they reveal their tendencies to have finite-time
singularities in the sense of geodesic incompleteness for primarily
one reason, the non-integrability of $H$. Finite-time singularities such
as collapse, big-rip or sudden singularities occur because $H$ fails to be
integrable for an infinite time interval.
In such universes, privileged observers cannot exist
for an infinite proper time. We know from
Theorem \ref{frwcomp} that if such observers
existed for an infinite proper time, then the universe would be
timelike and null geodesically complete.
Indeed, this is the case in various families of complete models ranging from
dark energy, to inflationary to braneworlds that we have already studied.

Naturally, the fact that the non-integrability condition for $H$ can be sustained
in different models stems from the various ways that these models are constructed
and the specific features they share. This non-integrability may be supported
either by the presence of phantom-dark energy driving the observed cosmic
acceleration and leading to a divergence of both $|E|$ and $|D|$, or by a fluid
with only a diverging $|E|$ 
as in \cite{ba04}, or in the inflationary character of the specific model as
in \cite{borde}, etc. All these different particularities are
acknowledged and brought to light with the use of the Bel-Robinson energy.
We have shown that with the combination of our classifying methods we have
succeeded to obtain a complete picture of the reasons underlying
a wide variety of possible dynamical behaviours of the universe on approach
to a spacetime singularity.

\chapter{Braneworld singularities}
\section{Introduction}
In previous chapters we studied singularities that arise within
the framework of a four-dimensional RW-spacetime
that models the whole entity of the observed universe.
In this chapter we study possible singularities that develop when we consider
the different case in which the observed universe is
modeled by a four-dimensional hypersurface situated at a fixed position of an extra spatial
dimension. The whole spacetime is five dimensional: there are four
dimensions of space out of which only three are spanned by the hypersurface and
one dimension of time; such a hypersurface is called a \emph{3-brane} while the
full higher dimensional space is called the \emph{bulk}.

Although a brane can trap particles and forces making it
impossible for them to escape, it \emph{does} interact with the matter and forces in the
bulk. This happens because gravity naturally extends over all dimensions of
space and time and it can therefore influence the fields that are constrained to
live within the boundaries of the brane \cite{liza}. Other matter fields in the bulk can also
interact with the fields on the brane, and the strength of this interaction
is controlled by a coupling function that is model-dependent \cite{liza}.

Branes were discovered within the framework of string theory as locations
where open strings are constrained to end. Their importance was reinforced
when it was realised that they can play a catalytic role in proving duality
between the different versions of string theory. Driven by these new discoveries,
a great deal of recent research focuses on building model universes that include
branes (see Refs. \cite{br_1}-\cite{br_59} for a partial list).

All these different model universes that flourish in the recent
literature are termed \emph{braneworlds} 
and in many cases provided satisfactory explanations
to some of the most challenging unresolved issues of cosmology and particle
physics such as the hierarchy problem and the cosmological constant problem.
Below, we refer briefly to the ideas behind three of the most recognised
braneworld cosmological scenarios and to the implications they offer
for the singularity problem.
The first scenario is the one introduced
by Nima Arkani-Hamed, Savas Dimopoulos and
Gia Dvali in \cite{ADD} dubbed as the ADD scenario and the two scenarios
that follow are those introduced by Lisa Randall and Raman Sundrum in \cite{RS1}
and \cite{RS2} known as RS$^{1}$ and RS$^{2}$.

According to the ADD scenario all standard particles and
forces except gravity are confined on a brane while the extra dimension is
as big as a tenth of a millimeter; this is surprising since until then it was
believed that extra dimensions were significantly smaller. In this scenario
the hierarchy problem, that poses the question of why gravity at small distances
appears to be so much weaker than the other forces, \emph{is} solved:
the effect of gravity on particles could fade away because of the large extra
dimension so that gravity would appear much weaker than the brane-restricted
electromagnetic strong and weak forces whose strength
remains unaffected exactly because they cannot experience the extra
dimension.

A different approach to the hierarchy problem is provided in \cite{RS1} and
\cite{RS2}. In the RS$^{1}$ scenario the extra dimension is considered to be
small in size and bounded between two branes called the Gravitybrane and the
Weakbrane. In this model the authors make essential use of a warped geometry
where the \emph{warp factor} 
that defines how space, time, energies and masses
change with distance (measured along the fifth extra dimension)
decreases exponentially along the fifth dimension.
A consequence of this configuration is that the energy of the bulk and the
branes warps the spacetime in such a way that gravity is felt strong on the
Gravitybrane and weak on the Weakbrane which accommodates all standard particles
and forces. In contrast, in the RS$^{2}$ scenario where there is only one brane,
the extra dimension is infinite in size but it can nonetheless remain invisible
provided that the spacetime is suitably warped.

In this chapter we use the concept of a warped geometry in our braneworld
models but we will not focus so much on the exploration of the
phenomenological implications that these models offer,
rather on the interesting features of finite-distance singularities that
they exhibit as well as on very basic questions (to which we refer below)
that the development of such singularities erases.

In particular, we will focus on the analysis of a new version of a braneworld
model that was originally proposed by Arkani-Hamed \emph{et al.} in
\cite{nima}. In the original version \cite{nima} the model consists of a flat
three-brane embedded in a five dimensional bulk space that is filled with a
scalar field minimally coupled to the bulk and conformally coupled to the fields
on the brane.
The metric describing the five dimensional spacetime is a warped product
as in the RS scenarios. A basic feature of this model is the development of a
singularity at a finite distance from the
brane around which the warp factor, its derivative and the energy density of the
scalar field all diverge. This singularity was used in \cite{nima} in order to
provide an explanation to the cosmological constant problem while preserving four
dimensional flatness.
In particular, the authors in \cite{nima} argue that this singularity
on the one hand successfully cuts off the size of the extra dimension
giving rise to the observed four dimensional flatness, and on the other hand it
acts as a reservoir through which all of the vacuum energy is emptied so that
its value is brought to agrement with the observed one.

Our interest in this particular braneworld model springs precisely from the fact
that it exhibits a finite-distance singularity which bares resemblance to
the finite-time singularities that are the core of our research. Our quest is
to provide answers to the following questions that naturally arise: How general is the
behaviour found in \cite{nima}? Is the finite-distance singularity the
only possible outcome when we pass on to a curved instead of a flat brane?
How strongly are the dynamical features of the model affected by a different
form of bulk matter such as a perfect fluid? 
Is the type of singularity found in \cite{nima} the only possible one in
these new braneworld models, and if not, what other types of finite-distance
singularities can arise?

In this chapter we show that in a curved brane the finite-distance
singularity found in \cite{nima} shifts to an infinite distance away from the
brane, while with a perfect fluid 
in the bulk new types of finite-distance singularities emerge. On approach to
these singularities the warp factor can do only two things: It will either tend
to zero (collapse-brane I or II), or diverge
(big rip-brane singularities\footnote{See Def. \ref{brane_sing_names}})
while the energy density of the
bulk-matter components usually diverges. Our analysis implies that the effect
of an infinite energy leak to the extra dimension
is indeed generic.
For a qualitative analysis of such models the method of asymptotic splittings,
expounded in Appendix E, is a most useful tool.

The plan of this chapter is as follows: In Section 5.2, we review the basic
aspects of the embedding of a three-brane in a five dimensional bulk space. 
In Section 5.3, we give the field equations that describe the dynamics of a
three-brane embedded in a five dimensional bulk with a scalar field; in
Subsection 5.3.1, we analyse the possible
asymptotic behaviours of that model using the method of asymptotic splittings.
In Section 5.4, we explore the consequences of altering the matter in the
bulk. In particular, the case where the bulk is filled with a
perfect fluid has a dramatic effect on the global dynamics.
We no longer encounter just the singularities of the
type found in \cite{nima} but now we have dynamics characterised by the
finite-distance singularities referred to in the previous paragraph.
In Section 5.6, we conclude and discuss our results.
\section{Brane geometry}
Consider a four-dimensional hypersurface 
$\Sigma$ that is situated at the arbitrary but fixed position $Y_{\ast}$
where $Y$ is the coordinate along the fifth dimension of the embedding space.
The brane has negligible thickness compared to the extra
dimension (and also the other dimensions) extending over the higher dimensional
bulk space and can be viewed as a thin surface (surface layer).
The brane therefore naturally divides the bulk in two regions
$V^{+}$ and $V^{-}$ each containing one of the boundary surfaces
$\Sigma ^{+}$ and $\Sigma ^{-}$ that confine the brane. The two
regions $V^{+}$ and $V^{-}$ are expected to look identical from $\Sigma ^{+}$ and
$\Sigma ^{-}$, the only difference being that their normal vector should be
reversed; in other words our braneworld  admits a $Y\rightarrow -Y$ symmetry. The extrinsic
curvatures of the two embeddings, $K_{ab}^{+}$ and
$K_{ab}^{-}$, therefore satisfy
\be
\label{jump_extr}
K_{ab}^{+}=-K_{ab}^{-}, \quad \textrm{everywhere on}\quad \Sigma.
\ee

Using the Gauss-Codazzi formalism, analysed in Appendix B.2, we can write the
Ricci $5$-tensor in a form analogous to (\ref{ricciepsilon})
(note that $\epsilon (n)=1$ ), i.e.,
\be
R_{ab}=
^{(4)}R_{ab}-K_{ab}K+2K_{an}K^{n}_{b}-\frac{\partial K_{ab}}{\partial Y},
\ee
or
\be
\label{ricci5}
R_{ab}=-\frac{\partial K_{ab}}{\partial Y}+Z_{ab},
\ee
where
\be
Z_{ab}=^{(4)}R_{ab}-K_{ab}K+2K_{an}K^{n}_{b}.
\ee
Tracing the Einstein equations
\be
\label{einst_5d}
R_{ab}-\frac{1}{2}g_{ab}R=\kappa_{5}^{2}T_{ab},
\ee
where $\kappa^{2}_{5}=8\pi G_{5}=M_{5}^{-3}$ with $G_{5}$ being the five
dimensional Newton's constant and $M_{5}$ the five
dimensional Planck mass, we find that the scalar curvature is given by
\be
R=-\frac{2}{3}\kappa_{5}^{2}T.
\ee
Substituting back in (\ref{einst_5d}) leads to the following equivalent form
of the Einstein equations:
\newpage
\be
\label{einstricci}
R_{ab}=\kappa_{5}^{2} \left(T_{ab}-\frac{1}{3}g_{ab}T\right).
\ee
Inserting Eq. (\ref{ricci5}) in Eq. (\ref{einstricci}) and integrating over the
layer defined by the thickness of the brane we obtain \cite{israel1},
\cite{israel2}
\be
-K_{ab}^{+}+K_{ab}^{-}+
\int_{Y_{\ast}-\epsilon}^{Y_{\ast}+\epsilon}Z_{ab}dY=
\kappa^{2}_{5} \left(T_{ab}-\frac{1}{3}g_{ab}T\right).
\ee

In the passage to the limit $\epsilon\rightarrow 0$ and
for fixed $-K_{ab}^{+}$, $K_{ab}^{-}$ we naturally expect that
$K_{ab}$ remains bounded inside the layer and hence
the integral of $Z_{ab}$ over the layer should vanish \cite{israel1},
\cite{israel2}.
Taking thus the limit of the above equation for
$\epsilon\rightarrow 0$ and using
Eq. (\ref{jump_extr}), we find the following \emph{junction condition}
\cite{israel1}, \cite{israel2}  (see also \cite{mtw} pp. 551-556)
\be
\label{junction_cond}
[K_{ab}]\equiv K_{ab}^{+}-K_{ab}^{-}=2K_{ab}^{+}=-
\kappa^{2}_{5}\left(S_{ab}-\frac{1}{3}g_{ab}S\right),
\ee
that relates the jump of the extrinsic curvature $[K_{ab}]$
across the brane to the surface-energy momentum tensor $S_{ab}$
(this follows from the discussion below) defined by
\be
S_{ab}\equiv\lim_{\epsilon\rightarrow 0}
\int_{Y_{\ast}-\epsilon}^{Y_{\ast}+\epsilon} T_{ab} dY,
\ee
with $S= g^{ab}S_{ab}$ being the trace of $S_{ab}$. It can be proved
(cf. \cite{israel1}, \cite{israel2}, and \cite{mtw} pp. 551-556)
that integration of the rest of the components of the Einstein equations gives
$S_{55}$ and $S_{5a}$ with $a=1,2,3,4$ both zero and
and hence, we conclude that there is no flow of energy from the fields
restricted on $\Sigma$ into the bulk. It is
thus implied 
that the quantity $S_{ab}$ acts as a delta function source of matter existing
only on $\Sigma$ and vanishing off $\Sigma$. This last property suggests that
$S_{ab}$ plays the role of a surface-energy momentum tensor. Combining this with
Eq. (\ref{junction_cond}), it follows that the
extrinsic curvature of the brane is completely determined by the matter confined
inside the brane \cite{binetruy}.
\newpage
\section{Asymptotics of scalar field-brane configuration}
Consider the braneworld model consisting of the 
three-brane configuration
analysed in the previous Section and filled with
a scalar field minimally coupled to the
five-dimensional bulk and conformally
coupled to any fields restricted on the brane. The total action
$S_{total}$ then splits in two parts, namely, the bulk action
$S_{bulk}$ and the brane action $S_{brane}$, i.e.,
\be
\label{s_tot}
S_{total}=S_{bulk}+S_{brane},
\ee
with
\bq
S_{bulk}&=&\int d^{4}x
dY\sqrt{g_{5}}\left(\frac{R}{2\kappa^{2}_{5}}-
\frac{\lambda}{2}(\nabla\phi)^{2}\right),\\
S_{brane}&=&-\int d^{4}x\sqrt{g_{4}}f(\phi),\,\, \textrm{at}\,\,Y=Y_{\ast},
\eq
where
$\lambda$ is a parameter defining the type of scalar field
$\phi$, $\kappa^{2}_{5}=M_{5}^{-3}$, $M_{5}$ is the five
dimensional Planck mass, and $f(\phi)$ denotes the tension of the brane
depending on the scalar field. The tension of the brane can be understood
physically as the resistance that one encounters when trying to distort
the brane: zero tension means that the brane has no resistance so that any
distortion would influence its form dramatically, while infinite tension
would deprive the brane's liberty to evolve dynamically \cite{liza}.

Varying the total action (\ref{s_tot}) with respect to $g^{AB}$,
we find the five-dimensional Einstein field equations in the form \cite{nima}
\be
\label{einst5d}
R_{AB}-\frac{1}{2}g_{AB}R=\lambda \kappa_{5}^{2}
\left(\nabla _{A}\phi\nabla_{B}\phi-\frac{1}{2}g_{AB}(\nabla\phi)^{2}\right)
-\frac{2\kappa_{5}^{2}}{\sqrt{g_{5}}}\frac{\delta S_{brane}}
{\delta g^{\a\b}}\delta_{\a}^{A}\delta_{\b}^{B}\delta(Y),
\footnote{Note here that by working in a similar way as in Appendix D,
we may interpret the energy-momentum tensor of the scalar field
as that of a perfect fluid in the bulk having density $\rho=\rho(Y)$ and
pressure $P=P(Y)$, 
i.e., $T_{AB}=(\rho+P)u_{A}u_{B}-Pg_{AB}$, where $u_{A}=(0,0,0,0,1)$, so that
$P=1/2(\nabla\phi)^{2}=\rho$.}
\ee
while the scalar field equation is obtained by variation of the action
(\ref{s_tot}) with respect to $\phi$ \cite{nima}
\footnote{The left-hand side of this
equation is the analogous left-hand side of the equation of motion
(\ref{emtscalar}) of a scalar field in a four-dimensional spacetime with
$V(\phi)=0$ (see appendix D).}
\newpage
\be
\label{scalarbr}
\lambda \Box _{5}\phi=-\frac{1}{\sqrt{g_{4}}}\frac{\delta S_{brane}}
{\delta\phi}\delta(Y),
\ee
where $A,B=1,2,3,4,5$ and $\a,\b=1,2,3,4$ while
$\delta (Y)=1$ at $Y=Y_{\ast}$ and vanishing everywhere else.

We assume a bulk metric of the form
\be
\label{warpmetric}
ds^{2}=a^{2}(Y)d\tilde{s}^{2}+dY^{2},
\ee
where $d\tilde{s}^{2}$ is the four dimensional flat, de Sitter or anti de Sitter
metric, i.e.,
\be
\label{branemetrics}
d\tilde{s}^{2}=
\left\{
         \begin{array}{l}
           -dt^{2}+dr^{2}+r^{2}d\sigma^{2} \\
           -dt^{2}+\frac{\cosh^{2}(H t)}{H^{2}}
           (dr^{2}+\sin^{2} r d\sigma^{2})\\
           -dt^{2}+ \frac{\cos^{2} (H t)}{H^{2}}
           (dr^{2}+\sinh^{2} r d\sigma^{2}),
         \end{array}
       \right.
\ee
where $d\sigma^{2}=d\theta^{2}+\sin^{2}\theta d\varphi^{2}$.
The field equations then take the form
\bq
\label{feq1}
\frac{a'^{2}}{a^{2}}&=&\frac{\lambda\kappa^{2}_{5}\phi'^{2}}{12}+\frac{k
H^{2}}{a^{2}} \\
\label{feq2}
\frac{a''}{a}&=&-\frac{\lambda\kappa^{2}_{5}\phi'^{2}}{4} \\
\label{feq3}
\phi''+4\frac{a'}{a}\phi'&=&0,
\eq
where $k=0,\pm 1$ and $H^{-1}$ is the de Sitter curvature radius and
the prime denotes differentiation with respect to $Y$.
The variables to be determined are $a$, $a'$ and $\phi$.
These three equations are not independent since Eq. (\ref{feq2}) was
derived after substitution of Eq. (\ref{feq1}) in the field equation
$G_{\a\a}=\kappa_{5}^{2}T_{\a\a}$, $\a=1,2,3,4$,
\be
\frac{a''}{a}+\frac{a'^{2}}{a^{2}}-\frac{kH^{2}}{a^{2}}
=-\lambda\kappa_{5}^{2}\frac{\phi'^{2}}{6}.
\ee
In our analysis we use the independent equations (\ref{feq2}) and
(\ref{feq3}) to determine the unknown variables $a$, $a'$ and $\phi'$
and the Eq. (\ref{feq1}) will then play the role of a constraint equation
for our system.

The extrinsic curvature of the three-brane is calculated to be
\be
\label{extrbr}
K_{\alpha\beta}=\frac{1}{2}\frac{\partial}{\partial Y} g_{\a\b}=a'(Y)a(Y),
\ee
while the surface energy-momentum tensor (defined only on the brane and
vanishing off the brane) is now given by
\be
\label{s_ij}
S_{\a\b}=-\frac{2\kappa^{2}_{5}}{\sqrt{g_{5}}}
\frac{\delta S_{brane}}{\delta g^{\alpha\beta}}
\delta^{A}_{\a}\delta^{B}_{\b}\delta(Y)=-a^{2}(Y_{\ast})f(\phi(Y_{\ast})).
\ee
Solving Eq. (\ref{einst5d}) on the brane is equivalent to
substituting Eqs. (\ref{extrbr}) and (\ref{s_ij}) in condition
(\ref{junction_cond}) which leads to the following relation to be satisfied by
$a'(Y_{\ast})$:
\be
\label{bound1}
a'(Y_{\ast})=-\frac{\kappa_{5}^{2}}{6}f(\phi(Y_{\ast}))a(Y_{\ast}).
\ee
On the other hand, Eq. (\ref{scalarbr}) gives directly
\be
\frac{\lambda}{\sqrt{g_{5}}}\frac{\partial}{\partial Y}(\sqrt{g_{5}}\phi')=
\frac{1}{\sqrt{g_{4}}}\int dx^{4}\sqrt{g_{4}}f'(\phi)\delta (Y).
\ee
Integrating over the interval $[Y_{\ast}-\epsilon,Y_{\ast}+\epsilon]$
and taking the limit as $\epsilon\rightarrow 0$, we finally find
\be
\label{bound2}
\phi'(Y_{\ast})=\frac{f'(\phi(Y_{\ast}))}{2\lambda},
\ee
where we have used the $Y\rightarrow -Y$ symmetry, i.e.,
$\phi'^{+}=\phi'(Y_{\ast}+\epsilon)=-\phi'(Y_{\ast}-\epsilon)=-\phi'^{-}$.

Our purpose now is to find all possible asymptotic behaviours around the assumed
position of a singularity, denoted by $Y_{s}$, emerging from general or particular
solutions of the system of equations (\ref{feq1})-(\ref{feq3}) describing our
braneworld cosmology. As we mentioned in the Introduction a most useful
tool for this analysis is the method of asymptotic splittings
(\cite{skot}, see also Appendix E).
According to this method, we start by setting
\be
x=a, \quad y=a', \quad z=\phi',
\ee
where the differentiation is considered with
respect to $\Upsilon=Y-Y_{s}$. 
The field equations (\ref{feq2}) and (\ref{feq3}) become the following
system of ordinary differential equations:
\bq
\label{syst1_1}
x'&=&y \\
\label{syst1_2}
y'&=&-\lambda Az^{2}x \\
\label{syst1_3}
z'&=&-4y\frac{z}{x},
\eq
where $A=\kappa^{2}_{5}/4$. Hence, we have the vector field
\be
\mathbf{f}=\left(y,-\lambda Az^{2}x,-4y\frac{z}{x}\right)^{\intercal}.
\ee
Equation (\ref{feq1}) does not include any terms containing derivatives
with respect to $\Upsilon$; it is to be regarded as a constraint equation
for the above system (\ref{syst1_1})-(\ref{syst1_3}). In terms of the new
variables the constraint is written in the form
\be
\label{constraint1}
\frac{y^{2}}{x^{2}}=\frac{A\lambda}{3} z^{2}+\frac{k H^{2}}{x^{2}}.
\ee
Substituting the forms $$(x,y,z)=(\alpha\Upsilon^{p},\beta
\Upsilon^{q},\delta\Upsilon^{r}),$$ where
\be
(p,q,r)\in\mathbb{Q}^{3} \quad \textrm{and} \quad
(\alpha,\beta,\delta)\in \mathbb{C}^{3}\smallsetminus\{\mathbf{0}\},
\ee
in the dynamical system (\ref{syst1_1})-(\ref{syst1_3}), we seek to determine the
possible \emph{dominant balances} in the neighborhood of the singularity, that is
pairs of the form
\be
\mathcal{B}=\{\mathbf{a},\mathbf{p}\}, \quad \textrm{where} \quad
\mathbf{a}=(\alpha,\beta,\delta), \quad \mathbf{p}=(p,q,r).
\ee
For our system (\ref{syst1_1})-(\ref{syst1_3}) we find$^{\dag}$:
\bq
\label{sing}
\mathcal{B}_{1}&=&\{(\alpha,\alpha/4,\sqrt{3}/(4\sqrt{A\lambda})),(1/4,-3/4,-1)\}\\
\label{nonsing} \mathcal{B}_{2}&=&\{(\alpha,\alpha,0),(1,0,-1)\}
\eq
\footnotetext{$^{\dag}$ We have ignored here a third balance
$\mathcal{B}_{3}=\{(\alpha,0,0),(0,-1,-1)\}$ since
it does not give the necessary $-1$ as one of the Kowalevskaya exponents
(see Appendix E for more details).}

Since (\ref{syst1_1})-(\ref{syst1_3}) is a weight-homogeneous system, the scale
invariant solutions given by the above balances are exact solutions of the
system. The balance  $\mathcal{B}_{1}$ satisfies the constraint
Eq. (\ref{constraint1}) only for $k=0$ corresponding thus to a general
(as we will see in the next paragraphs)
solution that allows \emph{only} for a flat spatial brane geometry. In contrast,
$\mathcal{B}_{2}$ corresponds to a particular solution for a \emph{curved brane}
since it satisfies Eq. (\ref{constraint1}) for $k\neq 0$ and
$\alpha^{2}=k H^{2}$ (here we have to sacrifice one arbitrary constant
by setting it equal to $kH^{2}$). 

Let us now focus on each of the two possible balances separately
and build series expansions in the neighborhood of the
singularity. We start by calculating the Kowalevskaya
exponents, eigenvalues of the matrix
$\mathcal{K}=D\mathbf{f}(\mathbf{a})-\textrm{diag}(\mathbf{p}),$
where $D\mathbf{f}(\mathbf{a})$ is the Jacobian matrix of $f$,
\be
D\mathbf{f}(x,y,z)=\left(
                     \begin{array}{ccc}
                       0 & 1             & 0 \\
               -A\lambda z^{2} & 0             & -2A\lambda x z \\
       \dfrac{4y z}{x^{2}} & -\dfrac{4z}{x} & -\dfrac{4y}{x} \\
                     \end{array}
                   \right),
\ee
evaluated on $\mathbf{a}$.
For the first balance we have that
$\mathbf{a}=(\alpha,\alpha/4,\sqrt{3}/(4\sqrt{A\lambda}))$,
and $\mathbf{p}=(1/4,-3/4,-1)$, thus
\be
\mathcal{K}=\left(
                     \begin{array}{ccc}
                       -\dfrac{1}{4} & 1                                   & 0 \\
                      -\dfrac{3}{16} & \dfrac{3}{4}             & -\dfrac{\sqrt{3A\lambda}\a}{2} \\
\dfrac{\sqrt{3}}{4\a\sqrt{A\lambda}} & -\dfrac{\sqrt{3}}{\a\sqrt{A\lambda}} & 0 \\
                     \end{array}
                   \right).
\ee
The $\mathcal{K}$-exponents are then given by
\bq
\textrm{spec}(\mathcal{K})=\{-1,0,3/2\}.
\eq
As it follows from the analysis in Appendix E, these exponents correspond to the
indices of the series coefficients where arbitrary constants first appear.
The $-1$ exponent signals the arbitrary position of the singularity, $Y_{s}$.
This first balance has two non-negative rational eigenvalues
which means that it describes a general solution, i.e., a solution
having the full number of arbitrary constants. In order to construct
an asymptotic expansion of this solution, we substitute in the system
(\ref{syst1_1})-(\ref{syst1_3}) the series expansions
$$\mathbf{x}=\Upsilon^{\mathbf{p}}(\mathbf{a}+
\Sigma_{j=1}^{\infty}\mathbf{c}_{j}\Upsilon^{j/s}),$$ where
$\mathbf{x}=(x,y,z)$, $\mathbf{c}_{j}=(c_{j1},c_{j2},c_{j3})$, and $s$ is
the least common multiple of the denominators of the positive
eigenvalues; here $s=2$, and the corresponding series expansions are
given by,
\be
x=\Sigma_{j=0}^{\infty}c_{j1}\Upsilon^{j/2+1/4},\quad
y=\Sigma_{j=0}^{\infty}c_{j2}\Upsilon^{j/2-3/4},\quad
z=\Sigma_{j=0}^{\infty}c_{j3}\Upsilon^{j/2-1}.
\ee
Following this, we arrive at the asymptotic solution around the singularity
in the form
\bq
\label{Puis_1x}
x&=&\alpha\Upsilon^{1/4}+\frac{4}{7}c_{32}\Upsilon^{7/4}+\cdots \\
y&=&\frac{\alpha}{4}\Upsilon^{-3/4}+c_{32}\Upsilon^{3/4}+\cdots\\
\label{Puis_1z}
z&=&\frac{\sqrt{3}}{4\sqrt{A\lambda}}\Upsilon^{-1}-
\frac{4\sqrt{3}}{7\alpha\sqrt{A\lambda}}c_{32}\Upsilon^{1/2}+\cdots.
\eq

The last step is to check if for each $j$ satisfying
$j/2=\rho$ with $\rho$ a positive eigenvalue, the corresponding
eigenvector $v$ of the $\mathcal{K}$ matrix is such that the \emph{compatibility
conditions} hold, namely,
\be
v^{\top}\cdot P_{j}=0,
\ee
where $P_{j}$ are polynomials in
$\mathbf{c}_{i},\ldots, \mathbf{c}_{j-1}$ given by
\be
\mathcal{K}\mathbf{c}_{j}-(j/s)\mathbf{c}_{j}=P_{j}.
\ee
Here the corresponding relation $j/2=3/2$, is valid only for $j=3$
and the associated eigenvector is
\be
\upsilon=\left(-\frac{\a\sqrt{A\lambda}}{\sqrt{3}},
-\frac{7\a\sqrt{A\lambda}}{4\sqrt{3}},1\right).
\ee
The compatibility condition,
$\upsilon\cdot (\mathcal{K}-(3/2)\mathcal{I}_{3})\mathbf{c}_{3}=0,$ therefore
indeed holds since,
\be
(\mathcal{K}-(3/2)\mathcal{I}_{3})\mathbf{c}_{3}=c_{32}
\left(
  \begin{array}{ccc}
    -\dfrac{7}{4} & 1 & 0 \\ \\
    -\dfrac{3}{16} & -\dfrac{3}{4} & -\dfrac{\a\sqrt{3A\lambda}}{2} \\ \\
    \dfrac{\sqrt{3}}{4\a\sqrt{A\lambda}} & -\dfrac{\sqrt{3}}{\a\sqrt{A\lambda}} & -\dfrac{3}{2} \\
  \end{array}
\right)
\left(
  \begin{array}{c}
    \dfrac{4}{7} \\ \\
    1 \\ \\
     -\dfrac{4\sqrt{3}}{7\a\sqrt{A\lambda}}\\
  \end{array}
\right)=\left(
          \begin{array}{c}
            0 \\ \\
            0 \\ \\
            0 \\
          \end{array}
        \right).
\ee
This shows that a representation of the solution asymptotically with a
Puiseux series as this is given by Eqs. (\ref{Puis_1x})-(\ref{Puis_1z}) is valid.
We thus conclude that near the singularity at finite distance $Y_{s}$
from the brane the asymptotic forms of the variables are
\be
\label{behscI}
a\rightarrow 0, \quad a'\rightarrow\infty, \quad \phi'\rightarrow \infty.
\ee
This is exactly the asymptotic behaviour of the solution found previously by
Arkani-Hammed \emph{et al} in \cite{nima}.

However, the previous behaviour is not the only possible one;
the analysis below shows that a very different outcome results by the
asymptotic analysis of the second balance.
The $\mathcal{K}$-matrix of the second balance is
\newpage
\be
\mathcal{K}=D\mathbf{f}((\a,\a,0))-\textrm{diag}(1,0,-1)=\left(
  \begin{array}{ccc}
    -1 & 1 &  0 \\
    0 & 0 &  0 \\
    0 & 0 & -3 \\
  \end{array}
\right),
\ee
with eigenvalues
\be
\textrm{spec}(\mathcal{K})=\{-1,0,-3\}.
\ee
For the balance $\mathcal{B}_{2}$ we find two distinct negative integer
$\mathcal{K}$-exponents and and hence we expect to find an infinite negative
expansion of a \emph{particular} solution (recall that we had to sacrifice one
arbitrary constant) around the presumed singularity at $Y_{s}$
with the negative $\mathcal{K}$-exponents signaling the positions where the
arbitrary constants first appear \cite{fordy}.
We therefore expand the variables in series with descending powers of $\Upsilon$
in order to meet the two arbitrary constants occurring for $j=-1$ and $j=-3$, i.e.,
\be
x=\Sigma_{j=0}^{-\infty}c_{j1}\Upsilon^{j+1}, \quad
y=\Sigma_{j=0}^{-\infty}c_{j2}\Upsilon^{j}, \quad
z=\Sigma_{j=0}^{-\infty}c_{j3}\Upsilon^{j-1}.
\ee
Substituting these series expansions back in the system
(\ref{syst1_1})-(\ref{syst1_3}),
we find the following asymptotic behaviour
\bq
\label{Puis_2x}
x&=&\alpha\Upsilon+c_{-1\,1}+\cdots \\
y&=&\alpha+\cdots \\
\label{Puis_2z}
z&=&c_{-3\,3}\Upsilon^{-4}+\cdots .
\eq
We now check the compatibility condition for $j=-3$
(the associated eigenvector here is $\upsilon=(0,0,1)$). We find that
\be
P_{-3}=(\mathcal{K}+3\mathcal{I}_{3})\mathbf{c}_{-3}=\left(
                                                \begin{array}{ccc}
                                                  2 & 1 & 0 \\
                                                  0 & 3 & 0 \\
                                                  0 & 0 & 0 \\
                                                \end{array}
                                              \right)
\left(
  \begin{array}{c}
    0 \\
    0 \\
    c_{33} \\
  \end{array}
\right)=\left(
          \begin{array}{c}
            0 \\
            0 \\
            0 \\
          \end{array}
        \right),
\ee
so that the compatibility condition is indeed satisfied.
The expansions given by Eqs. (\ref{Puis_2x})-(\ref{Puis_2z}) are therefore
valid and we can say that as $\Upsilon\rightarrow 0$, or equivalently as
$S=1/\Upsilon\rightarrow \infty$, we have that
\be
\label{behscII}
a\rightarrow \infty, \quad a'\rightarrow \infty, \quad \phi'\rightarrow \infty.
\ee
We thus conclude that there exist two possible outcomes given by
Eqs. (\ref{behscI}) and (\ref{behscII})
for these braneworld models, the dynamical behaviours of which strongly depend
on the spatial geometry of the brane. For a flat brane the model
experiences a finite-distance singularity through which all the
vacuum energy decays (since $\phi'\rightarrow\infty$, as $Y\rightarrow Y_{s}$),
whereas for a de Sitter or anti-de Sitter
brane the singularity is now located at an infinite distance.
Flat solutions correspond to the particular coupling
used by Arkani-Hammed \emph{et al} in \cite{nima}, the only
choice to make this possible. This easily follows by using equations
(\ref{bound1}) and (\ref{bound2}) and solving the FRW equation
(\ref{feq1}) on the brane for $kH^{2}$:
$$kH^{2}=\frac{a^{2}(Y_{\ast})\kappa^{2}_{5}}{12}
\left(\frac{\kappa_{5}^{2}}{3}f^{2}(\phi(Y_{\ast}))-
\frac{f'^{2}(\phi(Y_{\ast}))}{4\lambda}\right).$$ Clearly, $k$ is
identically zero if and only if:
$$\frac{f'(\phi)}{f(\phi)}=2\sqrt{\frac{\lambda}{3}}\kappa_{5},$$ or equivalently, if and
only if $f(\phi)\propto e^{2\sqrt{\lambda/3}\kappa_{5}\phi}$ (Arkani-Hammed
\emph{et al} in \cite{nima} have $\lambda=3$, and hence the appropriate choice for
the brane tension is $f(\phi)\propto e^{2\kappa_{5}\phi}$).

By working with other couplings we can allow for non-flat, maximally symmetric
solutions to exist and avoid having the singularity at a finite
distance away from the position of the brane.

In the following Section we replace the scalar field in the bulk with a perfect
fluid and explore the outcome of the resulting dynamical system.
\section{Perfect fluid in the bulk}
We consider here the case that the bulk matter takes the form of a perfect fluid
with equation of state $P=\gamma \rho$, where $P$ and $\rho$ are functions of $Y$ only.
This fluid can be viewed as an entity in its own right but we can also use it
to represent an average contribution of the various scalar fields living in
the bulk.

The field equations describing the braneworld model in this case are derived
by replacing the energy-momentum tensor in the first term of the
RHS of Eq. (\ref{einst5d}) by that of a perfect fluid,
$T_{AB}=(\rho+P)u_{A}u_{B}-Pg_{AB}$, $u_{A}=(0,0,0,0,1)$. They read:
\bq
\label{syst2i}
\frac{a''}{a}&=&-\kappa^{2}_{5}\frac{(1+2\gamma)}{6}\rho, \\
\label{syst2ii}
\rho'&+&4(1+\gamma)H\rho=0,\\
\label{syst2iii}
\frac{a'^{2}}{a^{2}}&=&\frac{\kappa^{2}_{5}}{6} A \rho+\frac{k H^{2}}{a^{2}}.
\eq
Introducing the new variables
\be
x=a, \quad y=a', \quad w=\rho,
\ee
Eqs. (\ref{syst2i})-(\ref{syst2iii}) can be written in the form
\bq
\label{syst2a}
x'&=&y, \\
y'&=&-2A\frac{(1+2\gamma)}{3}w x, \\
\label{syst2c}
w'&=&-4(1+\gamma)\frac{y}{x}w,
\eq
with the constraint equation given by
\be
\label{constraint3}
\frac{y^{2}}{x^{2}}=\frac{2}{3} A w+\frac{k H^{2}}{x^{2}}, \quad
A=\kappa_{5}^{2}/4.
\ee
The system (\ref{syst2a})-(\ref{syst2c}) is weight homogeneous with vector
field
\newpage
$$\mathbf{f}=\left(y,-2A\frac{(1+2\gamma)}{3}w x,-4(1+\gamma)
\frac{y}{x}w\right)^{\intercal},$$
and by setting $(x,y,w)=(\alpha\Upsilon^{p},\beta \Upsilon^{q},\delta
\Upsilon ^{s})$, we find the following set of five possible balances:
\bq
_{\gamma}\mathcal{B}_{1}&=&
\left\{\left(\alpha,\alpha p,\frac{3}{2A}p^{2}\right),(p,p-1,-2)\right\},
\quad p=\frac{1}{2(\gamma+1)}, \, \gamma \neq -1/2,-1,\\
_{\gamma}\mathcal{B}_{2}&=&\{(\alpha,\alpha,0),(1,0,-2)\},\quad \gamma \neq -1/2,\\
_{-1/2}\mathcal{B}_{3}&=&\{(\alpha,\alpha,0),(1,0,r)\},
\\
_{-1/2}\mathcal{B}_{4}
&=&\{(\alpha,\alpha,\delta),(1,0,-2)\},\\ 
_{-1/2}\mathcal{B}_{5}&=&\{(\a,0,0), (0,-1,r)\},
\eq
where $_{-1/2}\mathcal{B}_{i}\equiv_{\gamma=-1/2}\mathcal{B}_{i}$.
These balances are \emph{exact} solutions of the system and they should
therefore also satisfy the constraint equation (\ref{constraint3}).
This fact alters the generality of the solution represented by each balance
and determines the type of spatial geometry that it admits:
The balances $_{\gamma}\mathcal{B}_{1}$ and $_{-1/2}\mathcal{B}_{5}$ describe
a potentially general solution corresponding to a flat brane, while the balances
$_{\gamma}\mathcal{B}_{2}$ and $_{-1/2}\mathcal{B}_{3}$ describe
particular solutions of curved branes since we already have to sacrifice the
arbitrary constant $\a$ by accepting that $\a^{2}=kH^{2}$. The balance
$_{-1/2}\mathcal{B}_{4}$ on the other hand
describes a particular solution of a curved or flat brane since we have to set
$\d=(3/(2A))(1-kH^{2}/\a^{2})$ in order to satisfy Eq.(\ref{constraint3}).

Before we proceed to the analysis of each of these balances in the next
Subsections, we find it useful to introduce the following definition of
finite-distance singularities:
\begin{definition}
\label{brane_sing_names}
Let $a$ be the warp factor in the metric (\ref{warpmetric}), $a'$
its derivative and $\varrho$ the energy density of the matter component
in the bulk. We call a finite-distance brane-singularity occurring with:

\item[i)] $a\rightarrow 0$, $a'\rightarrow \infty$ and $\varrho\rightarrow \infty$:
a collapse-brane I singularity.

\item[ii)] $a\rightarrow 0$, $a'\rightarrow c$,
where $c$ is an arbitrary constant, and $\varrho\rightarrow\infty$:
a collapse-brane II singularity.

\item[iii)] $a\rightarrow \infty$, $a'\rightarrow \infty$ and
$\varrho\rightarrow\infty$: a big rip-brane singularity.
\end{definition}
Note that with this terminology, the finite-distance singularity we met in the
previous Section is a collapse-brane I singularity.
\subsection{The balance $_{\gamma}\mathcal{B}_{1}$}
We begin our analysis by calculating the Kowalevskaya-matrix
for the first balance. The Jacobian matrix of $\mathbf{f}$ is given by,
\be
D\mathbf{f}(x,y,z)=\left(
                     \begin{array}{ccc}
                       0 & 1             & 0 \\ \\
               -\dfrac{2}{3}(1+2\gamma)A w & 0             & -\dfrac{2}{3}(1+2\gamma)A x  \\ \\
       4(1+\gamma)\dfrac{y w}{x^{2}} & -4(1+\gamma)\dfrac{w}{x} & -4(1+\gamma)\dfrac{y}{x} \\
                     \end{array}
                   \right)
\ee
and 
\beq
\quad\quad\quad_{\gamma}\mathcal{K}_{1}&=&D\mathbf{f}\left(\a,\a p,\frac{3}{2A}p^{2}\right)-\textrm{diag}(p,p-1,-2)\\
&=& D\mathbf{f}\left(a,\frac{a}{2(1+\gamma)},\frac{3}{8A(1+\gamma)^{2}}\right)
-\textrm{diag}\left(\frac{1}{2(1+\gamma)},-\frac{1+2\gamma}{2(1+\gamma)},-2\right)
\eeq
\be
=\left(
                     \begin{array}{ccc}
                       -\dfrac{1}{2(1+\gamma)} & 1             & 0 \\ \\
               -\dfrac{1+2\gamma}{4(1+\gamma)^{2}} & \dfrac{1+2\gamma}{2(1+\gamma)}           & -\dfrac{2}{3}(1+2\gamma)A \a  \\ \\
       \dfrac{3}{4(1+\gamma)^{2}A \a} & -\dfrac{3}{2(1+\gamma)A \a} & 0 \\
                     \end{array}
                   \right).
\ee
The eigenvalues of this matrix are
\be
\textrm{spec}(_{\gamma}\mathcal{K}_{1})=
\left\{-1,0,\frac{1+2\gamma}{1+\gamma}\right\}.
\ee
The last eigenvalue is a function of the $\gamma$ parameter and it is positive
when either $\gamma<-1$ or $\gamma>-1/2$. We are going to focus on each of these
two ranges of $\gamma$ and explore the possible outcome of the system and the
character of the finite-distance singularities that occur each time.

Consider first the case with $\gamma<-1$. For purposes of illustration, we take
$\gamma=-2$. Then the
balance and the $_{-2}\mathcal{K}_{1}$-exponents read respectively,
\bq
_{-2}\mathcal{B}_{1}&=&\{(\a,-\a/2,3/(8A)),(-1/2,-3/2,-2)\},\\
\textrm{spec}(_{-2}\mathcal{K}_{1})&=&\{-1,0,3\}.
\eq
Substituting the value $\gamma=-2$ in the system given by
Eqs. (\ref{syst2a})-(\ref{syst2c}) and also the forms
\be
x=\Sigma_{j=0}^{\infty}c_{j1}\Upsilon^{j-1/2},\quad
y=\Sigma_{j=0}^{\infty}c_{j2}\Upsilon^{j-3/2},\quad
z=\Sigma_{j=0}^{\infty}c_{j3}\Upsilon^{j-2},
\ee
we expect to meet the third arbitrary constant at $j=3$. Indeed we find,
\bq
\label{ph_B1x}
x&=& \a\Upsilon^{-1/2}+\frac{2}{3}A\a c_{3\,3}\Upsilon^{5/2}+\cdots,\\
y&=& -\frac{\a}{2}\Upsilon^{-3/2}+\frac{5}{3}A \a c_{3\,3}\Upsilon^{3/2}+\cdots,\\
\label{ph_B1w}
w&=& \frac{3}{8A}\Upsilon^{-2}+c_{3\,3}\Upsilon+\cdots.
\eq
The compatibility condition is trivially satisfied for $j=3$ since
\be
(\mathcal{K}-3\mathcal{I}_{3})\mathbf{c}_{3}=
\left(
  \begin{array}{ccc}
    -\dfrac{5}{2} & 1 & 0 \\ \\
    \dfrac{3}{4} & -\dfrac{3}{2} & 2A\a \\ \\
    \dfrac{3}{4A\a} & \dfrac{3}{2A\a} & -3 \\
  \end{array}
\right)c_{3\,3}\left(
                 \begin{array}{c}
                   \dfrac{2}{3}A\a \\ \\
                   \dfrac{5}{3}A\a  \\ \\
                   1 \\
                 \end{array}
               \right)=\left(
                         \begin{array}{c}
                           0 \\
                           0 \\
                           0 \\
                         \end{array}
                       \right).
\ee
The series expansions given by Eqs. (\ref{ph_B1x})-(\ref{ph_B1w}) are therefore
valid asymptotically for $\Upsilon\rightarrow 0$ so that
\be
a\rightarrow \infty,\quad a'\rightarrow -\infty, \quad\rho\rightarrow \infty.
\ee
We therefore conclude that the balance $_{\gamma}\mathcal{B}_{1}$ 
describes a general solution (corresponding to a flat brane) around a big
rip-brane singularity when the perfect fluid
satisfies a phantom-like equation of state, i.e., $\gamma<-1$.
Note that using the analogy between the warp factor and the scale factor we
can say that this singularity bares many similarities to the
analogous one arising in the classical theory (see Section 4.3) since it is
characterized by all quantities $a$, $\rho$ and consequently $P$ and $H$ becoming
divergent asymptotically.

In order to illustrate the second range of values of $\gamma$, $\gamma>-1/2$,
we assume that $\gamma=-1/4$. Then
\bq
_{-1/4}\mathcal{B}_{1}&=&\{(\a,-2\a/3,2/(3A)),(2/3,-1/3,-2)\},\\
\textrm{spec}(_{-1/4}\mathcal{K}_{1})&=&\{-1,0,2/3\}.
\eq
Substituting in the system (\ref{syst2a})-(\ref{syst2c}) the particular value
$\gamma=-1/4$ and the forms
\be
x=\Sigma_{j=0}^{\infty}c_{j1}\Upsilon^{j/3+2/3}, \quad
y=\Sigma_{j=0}^{\infty}c_{j2}\Upsilon^{j/3-1/3}, \quad
z=\Sigma_{j=0}^{\infty}c_{j3}\Upsilon^{j/3-2},
\ee
we arrive at the following asymptotic expansions:
\bq
\label{g_B1x}
x&=& \a\Upsilon^{2/3}-\frac{A\a}{2}c_{2\,3}\Upsilon^{4/3}+\cdots,\\
y&=& \frac{2}{3}\a\Upsilon^{-1/3}-\frac{2}{3}A \a c_{2\,3}\Upsilon^{1/3}+\cdots,\\
\label{g_B1w}
w&=& \frac{2}{3A}\Upsilon^{-2}+c_{2\,3}\Upsilon^{-4/3}+\cdots.
\eq
Again, the compatibility condition is bound to be satisfied since in this case
\newpage
\be
(\mathcal{K}-(2/3)\mathcal{I}_{3})\mathbf{c}_{2}=
\left(
  \begin{array}{ccc}
    -\dfrac{4}{3} & 1 & 0 \\ \\
    -\dfrac{2}{9} & -\dfrac{1}{3} & -\dfrac{A\a}{3} \\ \\
    \dfrac{4}{3A\a} & -\dfrac{2}{A\a} & -\dfrac{2}{3} \\
  \end{array}
\right)c_{2\,3}\left(
                 \begin{array}{c}
                   -\dfrac{A\a}{2} \\ \\
                   -\dfrac{2A\a}{3}  \\ \\
                   1 \\
                 \end{array}
               \right)=\left(
                         \begin{array}{c}
                           0 \\ \\
                           0 \\ \\
                           0 \\
                         \end{array}
                       \right).
\ee
Eqs. (\ref{g_B1x})-(\ref{g_B1w}) then imply that as $\Upsilon\rightarrow 0$,
\be
a\rightarrow 0, \quad a'\rightarrow \infty, \quad \rho\rightarrow \infty.
\ee
This second asymptotic behaviour corresponds to a general solution of a flat
brane that is valid around a collapse-brane I singularity.
We thus 
regain a behaviour similar to that of the balance in (\ref{sing}).
However, an important outcome of our analysis is that there is now in
addition to that previously met, 
an equally general distinct \emph{big rip-brane singularity} that
may also arise during the evolution of a flat brane.
\subsection{The balance $_{\gamma}\mathcal{B}_{2}$}
For the second balance $_{\gamma}\mathcal{B}_{2}$, we find that
\be
_{\gamma}\mathcal{K}_{2}=D\mathbf{f}\left(\a,\a,0\right)
-\textrm{diag}\left(1,0,-2\right)=
\left(
  \begin{array}{ccc}
    -1 & 1 & 0 \\
    0 & 0 & -\dfrac{2}{3}A\a (1+2\gamma)\\
    0 & 0 & -2(1+2\gamma) \\
  \end{array}
\right),
\ee
and hence its spectrum is
\be
\textrm{spec}(_{\gamma}\mathcal{K}_{2})=\{-1,0,-2(1+2\gamma)\}.
\ee

We note here that the third arbitrary constant will appear at
$j=-2(1+2\gamma)$. After substituting the forms
\be
x=\Sigma_{j=0}^{\infty}c_{j1}\Upsilon^{j+1},\quad
y=\Sigma_{j=0}^{\infty}c_{j2}\Upsilon^{j},\quad
z=\Sigma_{j=0}^{\infty}c_{j3}\Upsilon^{j-2},
\ee
in the system (\ref{syst2a})-(\ref{syst2c}),
we may try giving different values to $\gamma$:
inserting $\gamma =-3/4$ in the system we meet the third arbitrary constant
at $j=1$ ($\textrm{spec}(_{-3/4}\mathcal{K}_{2})=\{-1,0,1\}$). The
variables $(x,y,w)$ then expand in the following way:
\bq
\label{-3/4_B2x}
x&=& \a\Upsilon+\frac{A\a}{6}c_{1\,3}\Upsilon^{2}+\cdots,\\
y&=& \a+\frac{A\a}{3} c_{1\,3}\Upsilon+\cdots,\\
\label{-3/4_B2w}
w&=&  c_{1\,3}\Upsilon^{-1}+\cdots.
\eq
We ought to check the validity of the compatibility condition for $j=1$. But this
is again trivially satisfied since
\be
(_{-3/4}\mathcal{K}_{2}-\mathcal{I}_{3})\mathbf{c}_{1}=
\left(
  \begin{array}{ccc}
    -2 & 1  & 0 \\
     0 & -1 & A\a/3 \\
     0 & 0  & 0 \\
  \end{array}
\right)c_{1\,3}\left(
                 \begin{array}{c}
                   A\a/6 \\
                   A\a/3  \\
                   1 \\
                 \end{array}
               \right)=\left(
                         \begin{array}{c}
                           0 \\
                           0 \\
                           0 \\
                         \end{array}
                       \right).
\ee
The series expansions Eqs.
(\ref{-3/4_B2x})-(\ref{-3/4_B2w}) are therefore valid ones, and we conclude
that as $\Upsilon\rightarrow 0$,
\be
\label{-3/4B2}
a\rightarrow 0, \quad a'\rightarrow\a, \quad \rho\rightarrow \infty,
\quad \alpha\neq 0.
\ee
It will follow from the analysis below that the behaviour of $\rho$ may vary
according to our choice of $\gamma$.

Indeed for $\gamma=-1$ ($\textrm{spec}(_{-1}\mathcal{K}_{2})=\{-1,0,2\}$)
we find that
\bq
\label{-1_B2x}
x&=& \a\Upsilon+\frac{A\a}{9}c_{2\,3}\Upsilon^{3}+\cdots,\\
y&=& \a+\frac{A\a}{3} c_{2\,3}\Upsilon^{2}+\cdots,\\
\label{-1_B2w}
w&=&  c_{2\,3}+\cdots.
\eq
Note that the compatibility condition is satisfied also here:
\be
(_{-1}\mathcal{K}_{2}-2\mathcal{I}_{3})\mathbf{c}_{2}=
\left(
  \begin{array}{ccc}
    -3 & 1  & 0 \\
     0 & -2 & 2A\a/3 \\
     0 & 0  & 0 \\
  \end{array}
\right)c_{2\,3}\left(
                 \begin{array}{c}
                   A\a/9 \\
                   A\a/3 \\
                   1 \\
                 \end{array}
               \right)=\left(
                         \begin{array}{c}
                           0 \\
                           0 \\
                           0 \\
                         \end{array}
                       \right).
\ee
We see then that as $\Upsilon\rightarrow 0$,
\be
a\rightarrow 0, \quad a'\rightarrow\a, \quad \rho\rightarrow c_{2\,3},
\quad \a\neq 0,
\ee
which shows that distinctly different behaviour than (\ref{-3/4B2}).

A yet different behaviour is met if we choose $\gamma=-5/4$: The
$\mathcal{K}$-exponents are in this case given by
$\textrm{spec}(_{-5/4}\mathcal{K}_{2})=\{-1,0,3\}$, and the series
expansions read:
\bq
\label{-5/4_B2x}
x&=& \a\Upsilon+\frac{A\a}{12}c_{3\,3}\Upsilon^{4}+\cdots,\\
y&=& \a+\frac{A\a}{3} c_{3\,3}\Upsilon^{3}+\cdots,\\
\label{-5/4_B2w}
w&=&  c_{3\,3}\Upsilon+\cdots.
\eq
These expansions are valid locally around the singularity since the
compatibility condition holds true,
\be
(_{-5/4}\mathcal{K}_{2}-3\mathcal{I}_{3})\mathbf{c}_{1}=
\left(
  \begin{array}{ccc}
    -4 & 1  & 0 \\
     0 & -3 & A\a \\
     0 & 0  & 0 \\
  \end{array}
\right)c_{3\,2}\left(
                 \begin{array}{c}
                   A\a/12 \\
                   A\a/3  \\
                   1 \\
                 \end{array}
               \right)=\left(
                         \begin{array}{c}
                           0 \\
                           0 \\
                           0 \\
                         \end{array}
                       \right).
\ee
For $\Upsilon\rightarrow 0$ we have that,
\be
a\rightarrow 0, \quad a'\rightarrow\a, \quad \rho\rightarrow 0,
\quad \a\neq 0.
\ee

This balance describes therefore the asymptotic behaviour of a particular
solution of a curved brane around a finite-distance singularity
where $a\rightarrow 0$ and $a'\rightarrow \alpha$, i.e., $a$ and $a'$ behave
exactly as in the case of a collapse-brane II singularity but
now the behaviour of the energy density of the perfect fluid varies
dramatically from zero to infinity ($\rho\rightarrow 0, \rho_{s},\infty$)
depending on the value of $\gamma$. 
This allows for new dynamical outcomes of the system
(\ref{syst2a})-(\ref{constraint3}):
we can have a constant or even no flow of energy through the
finite-distance singularity and into the extra dimension.

A totally different result is obtained when we consider a family of fluids having
a $\gamma$ greater than $-1/2$.
For $\gamma=0$, we have that $\textrm{spec}(_{0}\mathcal{K}_{2})=\{-1,0,-2\}$ and
hence we may expand $(x,y,w)$ in descending powers in order to meet the arbitrary
constant appearing at $j=-1$ and $j=-2$, i.e.,
\be
x=\Sigma_{j=0}^{-\infty}c_{j1}\Upsilon^{j+1},\quad
y=\Sigma_{j=0}^{-\infty}c_{j2}\Upsilon^{j},\quad
z=\Sigma_{j=0}^{-\infty}c_{j3}\Upsilon^{j-2}.
\ee
We find:
\bq
\label{0_B2x}
x&=& \a\Upsilon+c_{-1\,1}-\frac{A\a}{3}c_{-2\,3}\Upsilon^{-1}+\cdots,\\
y&=& \a+\frac{A\a}{3} c_{-2\,3}\Upsilon^{-2}+\cdots,\\
\label{0_B2w}
w&=&  c_{-2\,3}\Upsilon^{-4}+\cdots.
\eq
The compatibility condition is satisfied since
\be
(_{0}\mathcal{K}_{2}+2\mathcal{I}_{3})\mathbf{c}_{-2}=
\left(
  \begin{array}{ccc}
    1 & 1  & 0 \\
     0 & 2 & -2A\a/3 \\
     0 & 0  & 0 \\
  \end{array}
\right)c_{-2\,3}\left(
                 \begin{array}{c}
                   -A\a /3\\
                   A\a/3  \\
                   1 \\
                 \end{array}
               \right)=\left(
                         \begin{array}{c}
                           0 \\
                           0 \\
                           0 \\
                         \end{array}
                       \right).
\ee
As $S=1/\Upsilon\rightarrow\infty$ we conclude that
\be
a\rightarrow \infty, \quad a'\rightarrow\infty, \quad \rho\rightarrow \infty,
\ee
and we can therefore avoid the finite-distance singularity in this case.
\newpage
\subsection{The balances $_{-1/2}\mathcal{B}_{3}$, $_{-1/2}\mathcal{B}_{4}$
and $_{-1/2}\mathcal{B}_{5}$}

In this Section, we analyse the behaviours that are described by the three
balances which are valid once we set $\gamma=-1/2$ in the system
(\ref{syst2a})-(\ref{syst2c}). We show that the balance
$_{-1/2}\mathcal{B}_{3}$, for $r<-2$,
and the balance $_{-1/2}\mathcal{B}_{4}$ imply the existence of a
collapse-brane II singularity.
However, setting $r>-2$ in the balance $_{-1/2}\mathcal{B}_{3}$ leads to the
avoidance of the finite-distance singularity. The balance
$_{-1/2}\mathcal{B}_{5}$ on the other hand, describes a smooth behaviour since
the warp factor, its derivative and the energy
density of the fluid, remain finite asymptotically.

We begin with the analysis of the balance $_{-1/2}\mathcal{B}_{3}$.
The $\mathcal{K}$-matrix for this balance is
\be
_{-1/2}\mathcal{K}_{3}=D\mathbf{f}\left(\a,\a,0\right)
-\textrm{diag}(1,0,r)=
\left(
  \begin{array}{ccc}
    -1 & 1 & 0 \\
    0 & 0 & 0\\
    0 & 0 & -2-r \\
  \end{array}
\right),
\ee
and hence,
\be
\textrm{spec}(_{-1/2}\mathcal{K}_{3})=\{-1,0,-2-r\}.
\ee
Taking $-2-r>0$, we have two non-negative $\mathcal{K}$-exponents. For $r=-3$
we substitute the forms
\be
x=\Sigma_{j=0}^{\infty}c_{j1}\Upsilon^{j+1},\quad
y=\Sigma_{j=0}^{\infty}c_{j2}\Upsilon^{j},\quad
z=\Sigma_{j=0}^{\infty}c_{j3}\Upsilon^{j-3},
\ee
and arrive at the expansions
\bq
\label{-1/2_B3x}
x&=& \a\Upsilon+\cdots,\\
y&=& \a+\cdots,\\
\label{-1/2_B3w}
w&=&  c_{1\,3}\Upsilon^{-2}+\cdots.
\eq
The compatibility condition is satisfied since
\newpage
\be
(_{-1/2}\mathcal{K}_{3}-\mathcal{I}_{3})\mathbf{c}_{1}=
\left(
  \begin{array}{ccc}
    -2 & 1  & 0 \\
     0 & -1 & 0 \\
     0 & 0  & 0 \\
  \end{array}
\right)c_{1\,3}\left(
                 \begin{array}{c}
                   0 \\
                   0  \\
                   1 \\
                 \end{array}
               \right)=\left(
                         \begin{array}{c}
                           0 \\
                           0 \\
                           0 \\
                         \end{array}
                       \right).
\ee
The expansions (\ref{-1/2_B3x})-(\ref{-1/2_B3w}) are therefore valid ones in the
vicinity of the singularity 
and the general behaviour of the solution is characterised by
\be
a\rightarrow 0, \quad a'\rightarrow \a, \quad\rho \rightarrow\infty,
\quad \a\neq 0.
\ee
The balance $_{-1/2}\mathcal{B}_{3}$ for $r<-2$ implies therefore the
existence of a collapse-brane II singularity.

If on the other hand we take $r>-2$, we have two negative
$\mathcal{K}$-exponents.
If we choose the value $r=0$, then the spectrum is found to be
\be
\textrm{spec}(_{-1/2}\mathcal{K}_{3})=\{-1,0,-2\},
\ee
and so inserting the forms
\be
x=\Sigma_{j=0}^{-\infty}c_{j1}\Upsilon^{j+1},\quad
y=\Sigma_{j=0}^{-\infty}c_{j2}\Upsilon^{j},\quad
z=\Sigma_{j=0}^{-\infty}c_{j3}\Upsilon^{j},
\ee
we obtain
\bq
\label{-1/2_B3x}
x&=& \a\Upsilon+c_{-1\,1},\\
y&=& \a,\\
\label{-1/2_B3w}
w&=&  c_{-2\,3}\Upsilon^{-2}+\cdots,
\eq
which validates the compatibility condition at $j=-2$ since
\newpage
\be
(_{-1/2}\mathcal{K}_{3}+2\mathcal{I}_{3})\mathbf{c}_{-2}=
\left(
  \begin{array}{ccc}
    1 & 1  & 0 \\
     0 & 2  & 0 \\
     0 & 0  & 0 \\
  \end{array}
\right)c_{-2\,3}\left(
                 \begin{array}{c}
                   0 \\
                   0 \\
                   1 \\
                 \end{array}
               \right)=\left(
                         \begin{array}{c}
                           0 \\
                           0 \\
                           0 \\
                         \end{array}
                       \right).
\ee
We see that as $S=1/\Upsilon\rightarrow\infty$,
\be
a\rightarrow c_{-1\,1}, \quad a'\rightarrow\a, \quad \rho\rightarrow \infty,
\quad \a\neq 0,
\ee
so that the balance $_{-1/2}\mathcal{B}_{3}$ for $r>-2$ offers the possibility of
avoiding the finite-distance singularity met previously.

The balance $_{-1/2}\mathcal{B}_{4}$ is one with
\be
_{-1/2}\mathcal{K}_{4}=D\mathbf{f}\left(\a,\a,0\right)
-\textrm{diag}(1,0,-2)=
\left(
  \begin{array}{ccc}
    -1 & 1 & 0 \\
    0 & 0 & 0\\
    \dfrac{2\d}{\a} & -\dfrac{2\d}{\a} & 0 \\
  \end{array}
\right),
\ee
and
\be
\textrm{spec}(_{-1/2}\mathcal{K}_{4})=\{-1,0,0\}.
\ee
We note that the double multiplicity of the zero eigenvalue reflects the fact that
there were already two arbitrary constants,
$\a$ and $\d$ in this balance ($\d$ will have to be later sacrificed in order
for that balance to satisfy the constraint Eq. (\ref{constraint3})). We can thus
write
\bq
\label{-1/2_B4x}
x&=& \a\Upsilon+\cdots,\\
y&=& \a+\cdots,\\
\label{-1/2_B4w}
w&=&  \d\Upsilon^{-2}+\cdots,
\eq
so that as $\Upsilon\rightarrow 0$ a collapse-brane II singularity arises, i.e.,
\be
a\rightarrow 0,\quad a'\rightarrow \a, \quad\rho\rightarrow\infty,
\quad \alpha\neq 0.
\ee

Finally we come to the balance $_{-1/2}\mathcal{B}_{5}=\{(\a,0,0), (0,-1,r)\}$.
This has
\be
_{-1/2}\mathcal{K}_{5}=D\mathbf{f}\left(\a,0,0\right)
-\textrm{diag}(0,-1,r)=
\left(
  \begin{array}{ccc}
    0 & 1 & 0 \\
    0 & 1 & 0\\
    0 & 0 & -r \\
  \end{array}
\right),
\ee
and
\be
\textrm{spec}(_{-1/2}\mathcal{K}_{5})=\{1,0,-r\},
\ee
so we shall have to set $r=1$ in order to have the $-1$ eigenvalue corresponding
to the arbitrary position of the singularity, $t_{s}$. After substitution of
the forms
\be
x=\Sigma_{j=0}^{-\infty}c_{j1}\Upsilon^{j},\quad
y=\Sigma_{j=0}^{-\infty}c_{j2}\Upsilon^{j-1},\quad
z=\Sigma_{j=0}^{-\infty}c_{j3}\Upsilon^{j+1},
\ee
we find that
\bq
\label{-1/2_B4x}
x&=& \a+c_{1\,1}+\cdots,\\
y&=& c_{1\,1}+\cdots,\\
\label{-1/2_B4w}
w&=&  0+\cdots.
\eq
The compatibility condition is satisfied since
\be
(_{-1/2}\mathcal{K}_{5}-\mathcal{I}_{3})\mathbf{c}_{1}=
\left(
  \begin{array}{ccc}
    -1 & 1  & 0 \\
     0 & 0  & 0 \\
     0 & 0  & 0 \\
  \end{array}
\right)c_{1\,1}\left(
                 \begin{array}{c}
                   1 \\
                   1 \\
                   0 \\
                 \end{array}
               \right)=\left(
                         \begin{array}{c}
                           0 \\
                           0 \\
                           0 \\
                         \end{array}
                       \right),
\ee
and we see that as $\Upsilon\rightarrow 0$,
\be
x\rightarrow \a, \quad y\rightarrow c_{1\,1}, \quad w\rightarrow 0,
\quad \alpha\neq 0.
\ee
\section{Conclusions}
The following two tables summarise the overall features of a braneworld
consisting of a flat, or curved three-brane embedded in a five-dimensional bulk
space with a scalar field (Table \ref{scalarfield}),
or a perfect fluid (Table \ref{perfectfluid}). 
In what follows, we denote 
by $f_{s}$ the non-zero constant that is the limit of a
function $f$ on approach to the singularity at $Y_{s}$, and by FDS, IDS
a finite-distance and an infinite-distance singularity respectively.
\begin{table}[h]
\centering
  \caption{Singularities in the case of a scalar field in the bulk}
\label{scalarfield}
\begin{tabular}{l|l|l}
  \hline
  Balance & $k$ & Type of singularity \\ \hline \hline
  $\mathcal{B}_{1}$ & $k=0$ & FDS: $a\rightarrow 0,\,a'\rightarrow \infty,\,\phi\rightarrow\infty$ \\ \hline
  $\mathcal{B}_{2}$ & $k\neq 0$ & IDS: $a\rightarrow \infty,\,a'\rightarrow \infty,\,\phi\rightarrow\infty$ \\
  \hline
\end{tabular}
\end{table}
\begin{table}[h]
\centering
\caption{Singularities in the case of a perfect fluid in the bulk}
\label{perfectfluid}
\begin{tabular}{l|l|l|l|ll|l|l|l}
  \hline
  Balance                    & $k$         & $\gamma$      & $r$    & Type of singularity\\ \hline\hline
  $_{\gamma}\mathcal{B}_{1}$ & $k=0$       & $\gamma>-1/2$ &-       & FDS: $a\rightarrow 0,\,a'\rightarrow\infty,\,\rho\rightarrow \infty$ \\\hline
                             &             & $\gamma<-1$   &-       & FDS: $a\rightarrow \infty,\,a'\rightarrow\infty,\,\rho\rightarrow\infty$ \\\hline
  $_{\gamma}\mathcal{B}_{2}$ & $k\neq 0$   & $\gamma<-1/2$ &-       & FDS: $a\rightarrow 0,\,a'\rightarrow\alpha,\,\rho\rightarrow 0,\rho_{s},\infty$ \\\hline
                             &             & $\gamma>-1/2$ &-       & IDS: $a\rightarrow \infty,\,a'\rightarrow\infty,\,\rho\rightarrow\infty$ \\\hline
  $_{-1/2}\mathcal{B}_{3}$ & $k\neq 0$   & $\gamma=-1/2$ & $r<-2$ & FDS: $a\rightarrow 0,\,a'\rightarrow\alpha,\,\rho\rightarrow 
\infty$ \\
                             &             &               & $r>-2$ & IDS: $a\rightarrow \infty,\,a'\rightarrow\infty,\,\rho\rightarrow\infty$ \\\hline
 $_{-1/2}\mathcal{B}_{4}$ & $k=0,\pm 1$ & $\gamma=-1/2$ &-        & FDS: $a\rightarrow 0,\,a'\rightarrow\alpha,\,\rho\rightarrow\infty$ \\\hline
\end{tabular}
\end{table}

For a flat brane and for both forms of bulk matter we find that there exists a
collapse-brane I singularity as this is implied by the balances $\mathcal{B}_{1}$
and $_{\gamma}\mathcal{B}_{1}$ with $\gamma>-1/2$. When we consider a scalar field this is the
\emph{only} possible type of FDS that may arise. In contrast, when a perfect
fluid is present in the bulk there are in addition two new types of FDS:
The first one is a very distinct big rip-brane singularity
that appears only when a phantom equation of
state is considered ($\gamma<-1$). This new possibility is offered by
the balance $_{\gamma}\mathcal{B}_{1}$ and corresponds to a general solution.
The second new type of FDS is that of a collapse-brane II singularity that is
implied by the balance $_{-1/2}\mathcal{B}_{4}$.
However this feature is less
general than the ones described by the balance $_{\gamma}\mathcal{B}_{1}$. This
follows not only because we had to fix the value of $\gamma$ equal to
$-1/2$ in the balance $_{-1/2}\mathcal{B}_{4}$, but also because we had to
sacrifice one arbitrary constant by setting $\delta=3/(2A)$ in order to satisfy
the constraint Eq. (\ref{constraint3}). This implies that $_{-1/2}\mathcal{B}_{4}$
describes a particular solution with two arbitrary constants.

For a non-flat brane on the other hand,
the dynamical evolution of the braneworlds under consideration depends strongly
on the form of their bulk matter:
although FDSs \emph{cannot} arise when a scalar field fills the bulk, they
\emph{do} become possible once this field is replaced with a perfect fluid.
In particular, both $_{-1/2}\mathcal{B}_{4}$
and $_{-1/2}\mathcal{B}_{3}$, with $r<-2$
offer the possibility of a
collapse-brane II singularity, while
the balance $_{\gamma}\mathcal{B}_{2}$ with $\gamma<-1/2$
implies the existence of a FDS
around which $a\rightarrow 0$ and $a'\rightarrow \alpha$, i.e.,
$a$ and $a'$ behave exactly as in the case of a collapse-brane II singularity but
the energy leak into the extra dimension
is monitored each time by the $\gamma$
parameter that defines the type of fluid (or the dominant exponent).
However, for a non-flat brane there still exists the possibility of escaping the
FDSs and this is suggested by the balances $_{\gamma}\mathcal{B}_{2}$ with
$\gamma>-1/2$ and $_{-1/2}\mathcal{B}_{3}$ with $r>-2$.
(Recall that all possible balances in the case of a non-flat brane correspond
to particular solutions with two arbitrary constants.)

It would be interesting to further investigate if the FDSs found here continue
to emerge when we consider the case in which the scalar field
\emph{coexists} with the
perfect fluid in the bulk. The analysis of this more involved case may also
shed light to the factors that control how these two bulk matter
components compete on approach to the singularity or even predict
new types of FDSs that might then become feasible. 
\chapter{Conclusions and future work}
Two main themes throughout this thesis were completeness and singularities.
In particular, in Chapter 2 we stated and proved the general completeness theorem of
\cite{chc02} for the case of an isotropic spacetime. This theorem indicates
the underlying reasons that lead to completeness in the category of isotropic
models, some examples of which we analysed later on in that Chapter.

We continued in Chapter 3 by formulating necessary conditions for the
appearance of finite-time singularities in isotropic universes based entirely
on the behaviour of the Hubble parameter. These conditions
provided us with a first classification of such singularities, a
scheme which was then further refined and expanded with the use
of the Bel-Robinson energy. 
This energy represents the energy of the gravitational field projected in a sense to a slice
in spacetime and it has been previously used for the proof of global-in-time
results (see \cite{cbm03}, \cite{chm}, \cite{ck93}).
Combining the quantities that constitute the Bel Robinson energy with the
pressure and density of the matter fields through the field equations, we were
thus led to a clear picture of how the matter fields influence the nature of the
singularities. The new classification scheme is based on the
behaviour of the Hubble parameter, the Bel-Robinson energy and the scale factor.
As it was shown later on in Chapter 4, an examination of models from the recent
literature illustrates that this classification not only includes all known
types of singularities but it also predicts new ones. We believe that this
classification is complete in the isotropic category in the sense that it
depicts all the different
qualities in which singularities may differ. Chapter 3 ended with the formulation
of a relation that holds between the existence of a closed trapped surface and
the divergence of the Bel-Robinson energy; this relation provided us with
the information needed to decide on the fate of isotropic cosmologies and
we illustrated this important result with various examples.

In Chapter 4, we analysed under the prism of our classification a wide variety
of isotropic cosmologies that accommodate during their existence one or even more
of the four highlighted singularity types of our classification
scheme, namely a collapse, big-rip, sudden, or a sudden-rip singularity. There
are many more types predicted by our classification
and the search for adequate physical circumstances giving rise to these
types of singularities is a project of future work.

Incorporating these new singularities in the evolution of
particular cosmologies will lead to a precise understanding of the role
that these singularities may play in determining the overall destiny of
such cosmologies. Recall from Section 4.3 that it was only after observations
indicated that dark energy may be described by the phantom equation of state
that the possibility of a big-rip singularity was considered seriously for the
first time. This led to the discovery that our universe may face a violent end
on approach to which every bound structure such as galaxies, stars, planetary
systems or even atoms would be ripped apart.

A similar situation may arise in the consideration of models accommodating
the not-so-far studied singularity types of our classification: in the case that
some model comes in agreement with observations we may be able to make
predictions about the future evolution of the observed universe.

Another project of future work is to extend the present classification by
first considering anisotropic and then inhomogeneous spacetimes within the
framework of general relativity; we have already discussed in Chapter 3
(see Section 3.5) the way that the classification scheme presented in this
work is expected to be enriched with singularities that are by necessity
impossible to occur in isotropic models.

A step towards another direction is to examine how our classification does change
when we abandon general relativity and consider other metric theories of gravity
such as $f(R)$ theories or scalar-tensor gravity. With this examination we will be able to
recognise the way that these different theories of gravity affect the
number and types of singularities that are possible to occur in a given class of
spacetimes.

The main body of this thesis ended with Chapter 5 at which point we took a turn
from the classical case of a four-dimensional spacetime to explore the character
of singularities that arise in the different cosmological landscape of
braneworlds. Our research focused on the asymptotic analysis of singularities that
may arise during the dynamical evolution of a three-brane embedded in a
five-dimensional bulk space that is inhabited by a scalar field or a perfect
fluid. Our analysis, which was performed with the method of asymptotic splittings
of \cite{skot}, shows that there are three main possible types of
finite-distance singularities which we call collapse-brane I, II and big
rip-brane singularities.
A common characteristic of these singularities is that they act as a reservoir
through which all the energy density of the bulk-matter components may be emptied.

This phenomenon of an infinite leak of energy into the extra dimension was
previously used in \cite{nima} as an explanation to the cosmological constant
problem. Their model consisted of a flat three-brane embedded in a five
dimensional bulk with a scalar field and exhibited a collapse-brane I
singularity that successfully cut off the extra dimension and swallowed the
vacuum energy, offering thus an explanation for the observed four-dimensional
flatness and the vanishing value of the cosmological constant.
We showed that this singularity is removed to an infinite distance from the brane
when instead of a flat brane we considered a curved one while by altering the
matter in the bulk we allowed for the new distinct types of finite-distance
singularities mentioned above to emerge.

Future work in the field of braneworld singularities includes the construction of new
versions of our model with the consideration of other forms of bulk
matter such as the combination of a scalar field and a perfect fluid, or a fluid
with unconnected pressure and density. 
Studying different forms of bulk matter
may lead to the discovery of new finite-distance singularities which will
help us to draw conclusions about the possible variety of singularities that
these braneworlds can accommodate and hence help to determine their true range
of validity.

\appendix

\chapter{Notation and conventions}
In this thesis we use the signature $(-,+,+,+)$
for the metric tensor $g_{ab}$. The Riemann tensor is defined by
\be
R_{bcd}^{a}=\partial_{c} \Gamma_{bd}^{a}-\partial_{d} \Gamma_{bc}^{a}
+\Gamma_{cf}^{a}\Gamma_{bd}^{f}-\Gamma^{a}_{fd}\Gamma^{f}_{bc},
\ee
where $\partial_{a}$ stands for the partial derivative $\partial/\partial x^{a}$
and
\be
\Gamma_{bd}^{a}=\frac{1}{2}g^{ac}(\partial_{b}
g_{cd}+\partial_{d}g_{cb}-\partial_{c}g_{bd}).
\ee
are the Christoffel symbols. The Ricci tensor is defined by contracting the first
and the third index of the Riemann tensor, i.e.,
\be
R_{ab}=R_{acb}^{c}.
\ee
The Einstein tensor is given by
\be
G_{ab}=R_{ab}-\frac{1}{2}g_{ab}R,
\ee
where $R$ is the scalar curvature, $R=g^{ab}R_{ab}$.

In Chapters 2,3,4 we use units so that $8\pi G=c=1$. The Einstein
equations then read
\be
G_{ab}=+T_{ab}
\ee
where $T_{ab}$ is the energy-momentum tensor.

The symbol $\nabla_{a}$ stands for the covariant derivative.

We also use the following abbreviations:
RW for Robertson-Walker, FRW for Friedmann-Robertson-Walker, CTS for closed
trapped surface, FDS for finite-distance singularity and IDS for
infinite-distance singularity.

\chapter{Some submanifold geometry formulae}
In this appendix we review some useful formulae from submanifold geometry.
In Section B.1, we calculate the extrinsic curvature of the spacelike slices
of an arbitrary globally hyperbolic spacetime. In Subsection B.1.1, we show
that the extrinsic curvature of an RW space is proportional to the Hubble rate.
In Section B.2, we derive the basic Gauss-Codazzi relations that associate the
Riemann and the Ricci tensor of the spacelike slices
with the Riemann and the Ricci tensor of the spacetime.
\section{Extrinsic curvature}
Consider a globally hyperbolic spacetime $(\mathcal{V},g)$, that is a spacetime
which can be foliated by Cauchy surfaces $\Sigma_{t}\equiv\Sigma\times\{t\}$ parametrised by a global
time function $t$, and a smooth congruence of timelike geodesics in
$(\mathcal{V},g)$ orthogonal to the spacelike hypersurfaces $\Sigma_{t}$. If
$n^{a}$ is the unit tangent vector field ($n^{a}n_{a}=-1$) to these
geodesics, then $n^{a}$ is also the unit normal vector field to the
hypersurfaces $\Sigma_{t}$. The spacetime metric $g_{ab}$ induces the spatial
metric $h_{ab}$ on $\Sigma_{t}$ given by \cite{wald}
\be
\label{projop}
h_{ab}=g_{ab}+n_{a}n_{b}.
\ee
It then follows that $h_{ab}$ satisfies
\newpage
\be
h_{ab}n^{a}=0,\quad h_{ab}h^{b}_{c}=h_{ac}. 
\ee

To understand the role that $h_{ab}$ plays consider a spacetime vector $w^{a}$.
We can split $w^{a}$ into components tangent and normal to $\Sigma_{t}$ so that
\be
w^{a}=w_{\bot}n^{a}+w^{a}_{\|},
\ee
where $w^{a}_{\|}n_{a}=0$. If $w^{a}$ is a tangential vector, then
$w^{a}=w^{a}_{\|}$ and $w_{\bot}=0$, or equivalently
\be
w^{a}=h^{a}_{b}w^{b}.
\ee
This condition can be generalised to apply to an arbitrary tensor in the
following way: we say that the tensor $T^{a_{1}\ldots a_{k}}_{b_{1}\ldots b_{l}}$
is tangential to $\Sigma_{t}$ if
\be
T^{a_{1}\ldots a_{k}}_{b_{1}\ldots b_{l}}=
h^{a_{1}}_{c_{1}}\cdots h^{a_{k}}_{c_{k}}h_{b_{1}}^{d_{1}}\cdots
h_{b_{l}}^{d_{l}} T^{c_{1}\ldots c_{k}}_{d_{1}\ldots d_{l}}.
\ee
This means that $h^{a}_{b}$ is the projection operator onto the
tangential plane perpendicular to $n^{a}$.

The extrinsic curvature $K_{ab}$ of $\Sigma_{t}$ is defined as
\be
\label{def_extr}
K_{ab}=\nabla_{a}n_{b}.
\ee
Therefore the extrinsic curvature describes the
\emph{bending} of $\Sigma_{t}$ in the spacetime. To see this more clearly,
consider the orthogonal deviation vector $\xi^{a}$ representing the
spatial infinitesimal displacement between two neighboring geodesics of the
congruence. Then $\mathcal{L}_{n}\xi^{a}=0$, which means that
$n^{b}\nabla_{b}\xi^{a}=\xi^{b}\nabla_{b}n^{a}=K_{b}^{a}\xi^{b}$. The last
equality implies that the extrinsic curvature measures the failure of $n^{a}$
to be parallelly
transported along a spacelike geodesic with tangent vector $\xi^{a}$. In the
same time, the first equality says that this failure is
equal to the failure of $\xi ^{a}$ to be parallelly transported along the
timelike geodesic with tangent $n^{a}$. This means that an observer on one of
the geodesics would see a nearby geodesic being stretched and
twisted \cite{ck_pr}.

Note that $K_{ab}$ is purely spatial because $K_{ab}n^{b}=0$. Since we have
considered a hypersurface-orthogonal congruence of geodesics, it can be proved
(cf. \cite{wald}, p. 217)
that the rotation $\omega _{ab}$, defined by
$\omega_{ab}=\frac{1}{2}(\nabla_{b}n_{a}-\nabla_{a}n_{b})=K_{[ba]}$, is zero
and hence the extrinsic curvature tensor is symmetric
$K_{ab}=K_{ba}$. We can therefore write
\be
2K_{ab}=\nabla_{a}n_{b}+\nabla_{b}n_{a}.
\ee
The right hand side of the above equation equals the Lie derivative of $g_{ab}$
with respect to $n^{a}$, hence
\be
\label{lie}
K_{ab}=\frac{1}{2}\mathcal{L}_{n}g_{ab}=
\frac{1}{2}\mathcal{L}_{n}(h_{ab}-n_{a}n_{b})=
\frac{1}{2}\mathcal{L}_{n}h_{ab}-\frac{1}{2}\mathcal{L}_{n}(n_{a}n_{b}).
\ee
But
\be
\mathcal{L}_{n}(n_{a})n_{b}=
(n^{c}\nabla_{c}n_{a}+n_{c}\nabla_{a}n^{c})n_{b}
=\left(g_{ad}n^{c}\nabla_{c}n^{d}+\frac{1}{2}\nabla_{a}(n_{c}n^{c})
\right)n_{b}=0,
\ee
because of the geodesic equation $n^{c}\nabla_{c}n^{d}=0$, and the
fact that $n^{c}n_{c}=-1$.
Using this and the Leibnitz rule in Eq. (\ref{lie}), we find the following
formula for the extrinsic curvature:
\be
K_{ab}=\frac{1}{2}\mathcal{L}_{n}h_{ab}.
\ee
Since a Gaussian coordinate system is adapted to $n^{a}$,
we finally get
\be
\label{curv_extr}
K_{ab}=\frac{1}{2}\frac{\partial h_{ab}}{\partial t}.
\ee
This relation states that the extrinsic curvature is the rate of
change of the spatial metric $h_{ab}$ along the congruence.
%
\newpage
\subsection{Extrinsic curvature for a RW space}
Consider now a $\textsc{RW}$ metric, i.e., one with
\be
d{s}^{2}=-dt^{2}+a^{2}(t)\left(dr^{2}+f^{2}(r)(d\theta^{2}+
  \sin^{2}\theta d\phi^{2})\right),
\ee
where
\be
f(r)=\left\{\begin{array}{ll}
  r, & k=0 \\
  \sin r, & k=1  \\
   \sinh r,& k=-1.
\end{array}
\right.
\ee
Then by the use of Eq. (\ref{curv_extr}), we find
\beq
\label{K11}
K_{11}&=&a(t)\dot{a}(t),\\
K_{22}&=&f^{2}(r)a(t)\dot{a}(t),\\
\label{K33}
K_{33}&=&f^{2}(r)\sin^{2}\theta\,a(t)\dot{a}(t).
\eeq
Now, since the spatial norm of a 2-covariant tensor $X_{ij}$ is given by
\be
\label{norm_2}
|X|^{2}=g^{ij}g^{kl}X_{ik}X_{jl}\quad
i,j,k,l=1,2,3,
\ee
we find that the norm of the extrinsic curvature is given by
\be
\label{extr_frw}
|K|^{2}=g^{ij}g^{rs}K_{ir}K_{js}=(g^{11})^{2}K_{11}^{2}+
(g^{22})^{2}K_{22}^{2}+(g^{33})^{2}K_{33}^{2}=3\frac{\dot{a}^{2}(t)}{a^{2}(t)}=
3 H^{2}(t),
\ee
where $H(t)$ is the Hubble expansion rate, $H=\dot{a}/a$.
\section{Gauss-Codazzi formalism} 
As we saw in the previous Section, the metric $g_{ab}$ induces a Riemannian metric
$h_{ab}$ on each hypersurface $\Sigma_{t}$. We know that there exists a unique
derivative operator associated with this metric, which we denote by $D_{c}$. It
can be proved that $D_{c}$ is given by \cite{wald}
\be
\label{spatder}
D_{c}T^{a_{1}\ldots a_{k}}_{b_{1}\ldots b_{l}}=h^{a_{1}}_{d_{1}}\ldots
h_{b_{l}}^{e_{1}}h_{c}^{f}\nabla_{f}T^{d_{1}\ldots d_{k}}_{e_{1}\ldots e_{l}}.
\ee
In turn $D_{c}$ determines a curvature tensor $^{(3)}R_{abc}^{d}$ on
$\Sigma_{t}$. This tensor is defined in analogy with the four dimensional case
by the relation
\be
\label{3Riem}
D_{a}D_{b}u_{c}-D_{b}D_{a}u_{c}=^{(3)}R_{abc}^{d}u_{d}.
\ee

Our goal is to find the relation between the Riemann tensor of the hypersurface
$^{(3)}R_{abc}^{d}$ and the Riemann tensor of the whole spacetime $R_{abc}^{d}$
(Eq. \ref{Riem3} below).
We start by calculating $D_{a}D_{b}u_{c}$:
\bq
D_{a}D_{b}u_{c}&=&D_{a}(h_{c}^{e}h_{b}^{d}\nabla_{d}u_{e})=
h_{b}^{g}h_{c}^{k}h_{a}^{f}\nabla_{f}(h_{k}^{e}h_{g}^{d}\nabla_{d}u_{e})
\nonumber \\
&=& h_{b}^{g}h_{c}^{k}h_{k}^{e}h_{g}^{d}h_{a}^{f}(\nabla_{f}\nabla_{d}u_{e})+
h_{g}^{d}h_{b}^{g}h_{c}^{k}h_{a}^{f}(\nabla_{d}u_{e})\nabla_{f}h_{k}^{e}
\nonumber \\
\label{DaDb}
&+& h_{b}^{g}h_{c}^{k}h_{a}^{f}h_{k}^{e}(\nabla_{d}u_{e})(\nabla_{f}h_{g}^{d}).
\eq
But by the use of
\bq
h_{a}^{f}h_{c}^{k}\nabla_{f}h_{k}^{e}&=&
h_{a}^{f}h_{c}^{k}\nabla_{f}(g_{k}^{e}+n_{k}n^{e})\nonumber \\
&=&h_{a}^{f}h_{c}^{k}(\nabla_{f}n_{k})n^{e}+
h_{a}^{f}h_{c}^{k}n_{k}\nabla_{f}n^{e}
+h_{a}^{f}h_{c}^{k}\nabla_{f}g_{k}^{e}\nonumber \\
&=& h_{c}^{k}K_{ak}n^{e}=K_{ac}n^{e},
\eq
Eq. (\ref{DaDb}) becomes,
\be
D_{a}D_{b}u_{c}=
h_{b}^{g}h_{g}^{d}h_{c}^{k}h_{k}^{e}h_{a}^{f}(\nabla_{f}\nabla_{d}u_{e})+
K_{ac}n^{e}h_{b}^{d}(\nabla_{d}u_{e})+h_{c}^{e}K_{ab}n^{d}(\nabla_{d}u_{e}).
\ee
Now since
\be
h_{b}^{d}(\nabla_{d}u_{e})n^{e}=h_{b}^{d}\nabla_{d}(u_{e}n^{e})-
h_{b}^{d}(\nabla_{d}n^{e})u_{e}=-K_{b}^{e}u_{e},
\ee
\newpage
we find for $D_{a}D_{b}u_{c}$,
\be
D_{a}D_{b}u_{c}=
h_{a}^{f}h_{b}^{d}h_{c}^{e}\nabla_{f}\nabla_{d}u_{e}+
h_{c}^{e}K_{ab}n^{d}\nabla_{d}u_{e}-K_{b}^{e}u_{e}K_{ac}.
\ee
Similarly, we get an analogous expression for $D_{b}D_{a}u_{c}$,
\be
\label{DbDa}
D_{b}D_{a}u_{c}=
h_{b}^{d}h_{a}^{f}h_{c}^{e}\nabla_{d}\nabla_{f}u_{e}+
h_{c}^{e}K_{ba}n^{d}\nabla_{d}u_{e}-K_{a}^{e}u_{e}K_{bc}.
\ee
Subtracting Eqs. (\ref{DaDb}) and (\ref{DbDa}) we have that
\bq
D_{a}D_{b}u_{c}-D_{b}D_{a}u_{c}&=&-K_{ac}K_{b}^{e}u_{e}+K_{bc}K_{a}^{e}u_{e}+
h_{a}^{f}h_{b}^{d}h_{c}^{e}(\nabla_{f}\nabla_{d}u_{e}-\nabla_{d}\nabla_{f}u_{e})
\nonumber\\
&=&-K_{ac}K_{b}^{e}h_{e}^{d}u_{d}+K_{bc}K_{a}^{e}h_{e}^{d}u_{d}+
h_{a}^{f}h_{b}^{d}h_{c}^{e}R_{fde}^{j}u_{j} \nonumber\\
\label{G-C_rhs}
&=&-K_{ac}K_{b}^{d}u_{d}+K_{bc}K_{a}^{d}u_{d}+
h_{a}^{f}h_{b}^{d}h_{c}^{e}h_{j}^{d}R_{fde}^{j}u_{d}.
\eq
Then from Eq. (\ref{G-C_rhs}) and Eq. (\ref{3Riem}) follows the first
Gauss-Codazzi relation:
\be
\label{g-c1}
^{(3)}R_{abc}^{d}=h_{a}^{f}h_{b}^{d}h_{c}^{e}h_{j}^{d}R_{fde}^{j}
-K_{ac}K_{b}^{d}+K_{bc}K_{a}^{d}.
\ee

In Gaussian normal coordinates it holds that
\be
h_{a}^{f}=h^{f\lambda}h_{\lambda a}=\delta_{a}^{f},
\ee
so that Eq. (\ref{g-c1}) becomes
\be
\label{Riem3}
^{(3)}R_{abc}^{d}=
R_{abc}^{d}-K_{ac}K_{b}^{d}+K_{ab}K_{c}^{d},
\ee
where we have used the fact that
$K_{bc}K_{a}^{d}=g_{lc}g^{la}K_{ab}K^{d}_{a}=K_{ab}K_{c}^{d}$.

Now we can use the Eq. (\ref{Riem3}) in order to deduce the analogous
relation that holds between the Ricci tensor of the hypersurface $^{(3)}R_{ab}$
and the Ricci tensor of the spacetime, $R_{ab}$.
For that purpose, we start by calculating the component 
$R^{n}_{anb}$ (cf. \cite{stephani}, pp. 163-164):
\bq
R^{d}_{arb}n_{d}n^{r}h_{q}^{b}h_{m}^{a}&=&
(\nabla_{a}\nabla_{r}n_{b}-\nabla_{r}\nabla_{a}n_{b})n^{r}h_{q}^{b}h_{m}^{a}
=(\nabla_{a}K_{rb})n^{r}h_{q}^{b}h_{m}^{a}-
(\nabla_{r}K_{ab})n^{r}h_{q}^{b}h_{m}^{a} \nonumber \\ &=&
\nabla_{a}(K_{rb}n^{r}h_{q}^{b}h_{m}^{a})-
K_{rb}(\nabla_{a}n^{r})h_{q}^{b}h_{m}^{a}-
K_{rb}n^{r}\nabla_{a}(h_{q}^{b}h_{m}^{a})\nonumber \\
&-&(\pounds_{n}K_{ab}-
K_{rb}\nabla_{a}n^{r}-K_{ar}\nabla_{b}n^{r})h_{q}^{b}h_{m}^{a}
\nonumber \\ &=& (K_{ar}K_{b}^{r}-\pounds_{n}K_{ab})h_{q}^{b}h_{m}^{a},
\eq
so that
\be
\label{Riem4}
R^{d}_{arb}n_{d}n^{r}=K_{ar}K^{r}_{b}-\pounds_{n}K_{ab}.
\ee

Introducing Gaussian-normal coordinates and contracting the first and the third
index, we have that
\be
R^{n}_{anb}=-K_{ar}K^{r}_{b}+\frac{\partial K_{ab}}{\partial n}.
\ee
The Ricci tensor is then given by
\be
\label{ricci_gc}
R_{ab}=
R^{m}_{amb}+ 
R^{n}_{anb}=
^{(3)}R_{ab}+K_{ab}K-2K_{an}K^{n}_{b}+\frac{\partial K_{ab}}{\partial n}.
\ee

We can work in a similar way for the case of a timelike hypersurface with
normal vector $n^{a}$ (now $n^{a}n_{a}$=1).
Both cases of a spacelike and timelike normal vector can be treated at
the same time throughout the previous calculations by using the symbol
$\epsilon (n)=n^{a}n_{a}=\pm 1$. In particular, with the use of $\epsilon$,
the key equations (\ref{projop}), (\ref{Riem3}), (\ref{Riem4}) and
(\ref{ricci_gc}) can be written in the following form:
\bq
h_{ab}&=&g_{ab}-\epsilon(n)n_{a}n_{b},\\
^{(3)}R_{abc}^{d}&=&
R_{abc}^{d}
+\epsilon(n)(K_{ac}K_{b}^{d}-K_{ab}K_{c}^{d}),\\
R^{n}_{anb}\epsilon(n)&=&K_{ar}K^{r}_{b}-\pounds_{n}K_{ab},\\
\label{ricciepsilon}
R_{ab}&=&^{(3)}R_{ab}+
\epsilon(n)\left(-K_{ab}K+2K_{an}K^{n}_{b}-
\frac{\partial K_{ab}}{\partial n}\right).
\eq
\chapter{Bel-Robinson energy}
This appendix deals with the notion of the Bel-Robinson energy.
In Section C.1, we give the definition of the Bel-Robinson energy
for the general form of a sliced space and in Subsection C.1.1, we
calculate the Bel-Robinson energy of a RW space. In the end of the
Subsection C.1.1, we show that the null energy condition in closed or flat
accelerating universes translates into an inequality satisfied by the electric
parts of the Bel-Robinson energy.
\section{Definition}
The Bel-Robinson energy and related quantities play an important role in proving
global properties of spacetimes. For example, they were
used in \cite{ck93}, \cite{chm} to prove global existence results in
the case of asymptotically flat and cosmological spacetimes
respectively. In this thesis, the Bel-Robinson energy is
used for the purposes of classification of spacetime singularities.

The Bel-Robinson energy is a kind of energy of the gravitational
field \emph{projected} in a sense to a slice in spacetime.
Consider the sliced spacetime with metric
\begin{equation}
^{(3+1)}g\equiv -N^{2}(\theta ^{0})^{2}+g_{ij}\;\theta ^{i}\theta
^{j},\quad \theta ^{0}=dt,\quad \theta ^{i}\equiv dx^{i}+\beta
^{i}dt,  \label{2.1}
\end{equation}
where $N=N(t,x^{i})$ is the lapse function and $\beta^{i}(t,x^{j})$
is the shift function. The $2$-covariant spatial electric and
magnetic tensors are defined as follows
\bq
\label{br_E}
E_{ij}&=&R^{0}_{i0j},\\
\label{br_D}
D_{ij}&=&\frac{1}{4}\eta_{ihk}\eta_{jlm}R^{hklm},\\
\label{br_H}
H_{ij}&=&\frac{1}{2}N^{-1}\eta_{ihk}R^{hk}_{0j},\\
\label{br_B}
B_{ji}&=&\frac{1}{2}N^{-1}\eta_{ihk}R^{hk}_{0j},
\eq
where $\eta_{ijk}$ is the volume element of the space metric
$g_{t}\equiv g_{ij}dx^{i}dx^{j}$.
These four time-dependent space tensors comprise what is called a
\emph{Bianchi field}, $(E,H,D,B)$, a very important frame field used
to prove global in time results, cf. \cite{ycb02}.
The \emph{Bel-Robinson energy at time $t$} is given by
\be
\mathcal{B}(t)=\frac{1}{2}\int_{\mathcal{M}_t}\left(|E|^{2}+|D|^{2}+
|B|^{2}+|H|^{2}\right)d\mu_{g_t},
\ee
where by $|X|^{2}=g^{ij}g^{kl}X_{ik}X_{jl}$ we denote the spatial norm of the
$2$-covariant tensor $X$.

\subsection{Bel-Robinson energy of an isotropic universe}

Consider the RW metric (\ref{frw_metric}).
In order to compute the Bel-Robinson energy at time t, $\mathcal{B}(t)$, we
first calculate the two covariant spatial tensors $E_{ij}$ and $D_{ij}$ using
relations (\ref{br_E}) and (\ref{br_D}). These relations give
\beq
E_{11}&=&a\ddot{a},\\
E_{22}&=&f^{2}a\ddot{a},\\
E_{33}&=&f^{2}\sin^{2}\theta\, a\ddot{a},
\eeq
while
\newpage
\beq
D_{11}&=&f^{-2}(1-f'^{2}+f^{2}\dot{a}^{2}),\\
D_{22}&=&f^{2}\left(\dot{a}^{2}-\frac{f''}{f}\right),\\
D_{33}&=&f^{2}\sin^{2}\theta\left(\dot{a}^{2}-\frac{f''}{f}\right),
\eeq
where a prime denotes the differentiation of $f\equiv f(r)$ with respect to $r$,
while a dot denotes the differentiation of $a\equiv a(t)$ with respect to $t$.
The magnetic tensors $H_{ij}$, $B_{ji}$ make no contribution to the Bel-Robinson
energy since for a RW metric we have that $R^{hk}_{0j}=0$ for $h,k,j=1,2,3$.

Next, we make use of Eq. (\ref{norm_2}) in order to calculate the spatial norms
of the two electric tensors. We find
\be
|E|^{2}=(g^{11})^{2}E_{11}^{2}+(g^{22})^{2}E_{22}^{2}+(g^{33})^{2}E_{33}^{2}=
3\frac{\ddot{a}^{2}}{a^{2}},
\ee
and
\be
\label{d_1}
|D|^{2}=(g^{11})^{2}D_{11}^{2}+(g^{22})^{2}D_{22}^{2}+(g^{33})^{2}D_{33}^{2}=
\left(\frac{1-f'^{2}}{a^{2}f^{2}}+
\frac{\dot{a}^{2}}{a^{2}}\right)^{2}+2\left(\frac{\dot{a}^{2}}{a^{2}}-
\frac{f''}{f a^{2}}\right)^{2}.
\ee
It then follows from the form of $f(r)$ (cf. Eq. (\ref{f(r)})) that
\be
\label{d_2}
\frac{1-f'^{2}}{f^{2}}=-\frac{f''}{f}=k,
\ee
where $k$ is the constant curvature of the spatial slice. We can therefore write
$|D|$ in the form
\be
\label{d_3}
|D|^{2}=3\left(\frac{\dot{a}^{2}}{a^{2}}+\frac{k}{a^{2}}\right)^{2}.
\ee
The Bel-Robinson energy at time $t$ is then given by
\newpage
\be
\label{bre_frw}
\mathcal{B}(t)=\frac{C}{2}\left(|E|^{2}+|D|^{2}\right)=
\frac{3C}{2}\left(\frac{\ddot{a}^{2}}{a^{2}}+
\left(\frac{\dot{a}^{2}}{a^{2}}+\frac{k}{a^{2}}\right)^{2}\right),
\ee
where $C$ is the constant volume \emph{of} (or \emph{in}, in the case of a
non-compact space) the three dimensional slice at time $t$.

Finally, we have the following proposition that translates the null energy
condition into an inequality satisfied by the electric parts of the
Bel-Robinson energy.
\begin{proposition}
In an accelerating flat or closed universe the null energy condition is
equivalent to the inequality $|E|\leq |D|$.
\end{proposition}
\begin{proof}
The null energy condition demands that
\be
\label{nullcond}
R_{ab}k^{a}k^{b}\geq 0,
\quad \textrm{for every null vector}\quad k^{a}.
\ee
In an isotropic universe with metric described by Eqs. (\ref{frw_metric}) and
(\ref{f(r)}), the above condition translates to
\beq
R_{ab}k^{a}k^{b}&=&R_{11}(k^{1})^{2}+R_{22}(k^{2})^{2}+R_{33}(k^{3})^{2}+
R_{44}(k^{4})^{2}\\
&=&
\left(2\dot{a}^{2}-2\frac{f''}{f}+a\ddot{a}\right)(k^{1})^{2}+(1-f'^{2}-f
f''+f^{2}(2\dot{a}^{2}+a\ddot{a}))(k^{2})^{2}+
\eeq
\be
\label{re_1}
+\sin^{2}\theta(1-f'^{2}-f
f''+f^{2}(2\dot{a}^{2}+a\ddot{a}))(k^{3})^{2}-
3\frac{\ddot{a}}{a}(k^{4})^{2}.
\ee
Keeping in mind that $1-f'^{2}=-f f''$, we can write Eq.
(\ref{re_1}) as follows
\beq
\label{re_2}
R_{ab}k^{a}k^{b}&=&\left(2\frac{\dot{a}^{2}}{a^{2}}+\frac{\ddot{a}}{a}
-2\frac{f''}{fa^{2}}\right)[a^{2}((k^{1})^{2}+f^{2}((k^{2})^{2}
+\sin^{2}\theta(k^{3})^{2}))-(k^{4})^{2}]+\\
&+&\left(-3\frac{\ddot{a}}{a}+2\frac{\dot{a}^{2}}{a^{2}}
+ \frac{\ddot{a}}{a}-2\frac{f''}{fa^{2}}\right)(k^{4})^{2}
=\left(-2\frac{\ddot{a}}{a}+2\frac{\dot{a}^{2}}{a^{2}}
-2\frac{f''}{fa^{2}}\right)(k^{4})^{2}.
\eeq
Therefore if the null energy condition, Eq. (\ref{nullcond}), holds then
\be
\label{nullcondfrw}
-\frac{\ddot{a}}{a}+\frac{\dot{a}^{2}}{a^{2}}+\frac{k}{a^{2}}\geq 0,
\ee
where we have used the fact that $-f''/f=k$, with $k$ being the constant
curvature of the spatial slices. Since we have assumed a $k\geq 0$ accelerating
universe, the electric parts are given by
\beq
|D|&=&\sqrt{3}\left(\frac{\dot{a}^{2}}{a^{2}}+\frac{k}{a^{2}}\right),\\
|E|&=&\sqrt{3}\frac{\ddot{a}}{a}
\eeq
hence Eq.(\ref{nullcondfrw}) implies that $|E|\leq |D|$.
\end{proof}

%
\chapter{Energy-momentum tensor of a scalar field}
In order to define the energy-momentum tensor for a scalar field we start from
the action
\be
S_{\phi}=\int \mathcal{L}_{\phi}\sqrt{-g}d^{4} x,
\ee
where the Lagrangian density is
\be
\label{lagr_phi}
\mathcal{L}_{\phi}=-\frac{1}{2}(\nabla_{c}\phi\nabla^{c}\phi+2V(\phi))=
-\frac{1}{2}(g^{dc}\nabla_{d}\phi\nabla_{c}\phi+2V(\phi)).
\ee
The energy-momentum tensor for a matter field is defined to be,
cf. \cite{wald}, p. 455,
\be
T_{\mu\nu}\equiv-\frac{2}{\sqrt{-g}}
\frac{\delta(\mathcal{L}_{matter}\sqrt{-g})}{\delta g_{\mu\nu}},
\ee
for the Lagrangian (\ref{lagr_phi}), we find that
\bq
T_{\mu\nu}&=&\frac{2}{\sqrt{-g}}\left((\nabla^{c}\phi\nabla_{c}\phi
+2V(\phi))\frac{\delta(\sqrt{-g})}{\delta g^{\mu\nu}}+
\sqrt{-g}\frac{\delta(g^{\mu\nu}\nabla_{\mu}\phi\nabla_{\nu}\phi)}
{\delta g^{\mu\nu}}\right)=\\
&=&\frac{1}{\sqrt{-g}}\left((\nabla^{c}\phi\nabla_{c}\phi
+2V(\phi))(-\frac{1}{2}\sqrt{-g}g_{\mu\nu})+
\nabla_{\mu}\phi\nabla_{\nu}\phi\sqrt{-g}\right).
\eq
Thus the energy-momentum tensor of a scalar field is
\be
T_{\mu\nu}=\nabla_{\mu}\phi\nabla_{\nu}\phi-
\frac{1}{2}g_{\mu\nu}(\nabla_{c}\phi\nabla^{c}\phi+2V(\phi)),
\ee
or equivalently,
\be
\label{emtscalar}
T_{\mu\nu}=\nabla_{\mu}\phi\nabla_{\nu}\phi-
\frac{1}{2}g_{\mu\nu}(g^{dc}\nabla_{d}\phi\nabla_{c}\phi+2V(\phi)).
\ee
Comparing this with the energy-momentum tensor of a perfect fluid,
\be
T_{\mu\nu}=(\mu+p)u_{\mu}u_{\nu}+pg_{\mu\nu},
\ee
where $u_{\mu}$ is the unit timelike vector ($u^{\mu}u_{\mu}=-1$),
it follows that
\be
\label{pressure}
p=\mathcal{L}_{\phi}=-\frac{1}{2}(g^{dc}\nabla_{d}\phi\nabla_{c}\phi+2V(\phi)),
\ee
and
\be
(\mu+p)u_{\mu}u_{\nu}=\nabla_{\mu}\phi\nabla_{\nu}\phi.
\ee
Contracting the above equation, we find \cite{mad}
\be
-(\mu+p)=\nabla^{\nu}\phi\nabla_{\nu}\phi=-2\mathcal{L}_{\phi}-2V(\phi),
\ee
and substituting for the pressure from Eq. (\ref{pressure}) leads to
\be
\label{density}
\mu=\mathcal{L}_{\phi}+2V(\phi).
\ee
In a homogeneous universe, a scalar field is a function of time only and
hence Eqs. (\ref{pressure}), (\ref{density}) give the standard result,
\bq
p&=&\frac{\dot{\phi}^{2}}{2}-V(\phi),\\
\mu&=&\frac{\dot{\phi}^{2}}{2}+V(\phi).
\eq
The equation of motion of a scalar field can be found either by using the
conservation law of energy, i.e., $\nabla^{\nu} T_{\mu\nu}=0$, or by variation
of the action $S_{\phi}$ with respect to $\phi$. We use the second approach.
We have that:
\bq
\frac{\delta\mathcal{L}_{\phi}}{\delta\phi}=
\frac{\partial\mathcal{L}_{\phi}}{\partial\phi}-
\left(\frac{\partial\mathcal{L_{\phi}}}{\partial(\phi_{,\mu})}\right)_{,\mu}=
2V'(\phi)-2\nabla_{c}\nabla^{c}\phi,
\eq
and so $\delta\mathcal{L}_{\phi}/\delta\phi=0$ gives directly
\be
\label{eqmotscalar}
\Box\phi-V'(\phi)=0,
\ee
where $\Box\phi=(-g)^{-1/2}\nabla_{d}((-g)^{1/2}g^{dc}\nabla_{c}\phi)$.
\chapter{The method of asymptotic splittings}

Consider a system of $n$ first order ordinary differential equations
\be
\label{arb_syst}
\mathbf{\dot{x}}=\mathbf{f}(\mathbf{x}),
\ee
where $\mathbf{x}=(x_{1},\ldots,x_{n})\in \mathbb{R}^{n}$ and
$\mathbf{f}(\mathbf{x})=(f_{1}(\mathbf{x}),\ldots,f_{n}(\mathbf{x}))$.
The \emph{general solution} of the above system contains $n$ arbitrary constants
and hence describes the behaviour of the system for arbitrary initial data. A
\emph{particular solution} on the other hand contains less than $n$ arbitrary
constants and describes the behaviour of the system for a subset of initial data.

We say that a general or particular 
solution of the above dynamical system exhibits a finite-time singularity if
there exists a $t_{s}\in \mathbb{R}$ and a $\mathbf{x}_{0}\in \mathbb{R}^{n}$
such that 
\be
\lim_{t\rightarrow t_{s}}\|\mathbf{x(t;\mathbf{x}_{0})}\|\rightarrow\infty,
\ee
where $\|\centerdot\|$ is any $L^{p}$ norm.

The purpose of \emph{singularity analysis} (cf. \cite{skot},
\cite{gorielybook}, pp. 129-155) is to build series expansions of
solutions around the presumed position of a singularity at $t_{s}$
in order to study the different asymptotic behaviours of the solutions of the
system (\ref{arb_syst}).

In particular, we look for series expansions of solutions that take the form
of a Puiseux series,
\be
\label{Puiseux}
\mathbf{x}=\tau^{\mathbf{p}}\left(\mathbf{a}
+\Sigma_{i=1}^{\infty}\mathbf{c}_{i}\tau^{i/s}\right),
\ee
where $\tau=t-t_{s}$, $\mathbf{p}\in \mathbb{Q}^{n}$, $s\in\mathbb{N}$.

The method of asymptotic splittings for the system (\ref{arb_syst}) is realised
in the following steps\footnote{Here we will refer briefly to the basic steps of
the method of asymptotic splittings. A detailed analysis
can be found in \cite{skot}, \cite{goriely}.}:

$\bullet$ First, we find all the possible \emph{weight-homogeneous decompositions}
of the vector field $\mathbf{f}$ by splitting it into components
$\mathbf{f}^{(j)}$:
\be
\mathbf{f}=\mathbf{f}^{(0)}+\mathbf{f}^{(1)}+\ldots +\mathbf{f}^{(k)},
\ee
with each of these components being \emph{weight homogeneous}, that is to say
\be
\mathbf{f^{(j)}}(\mathbf{a}\tau^{\mathbf{p}})=\tau^{\mathbf{p}
+\mathbf{1}(q^{(j)}-1)}\mathbf{f}^{(j)}(\mathbf{a}) \quad j=0,\ldots,k,
\ee
where $\mathbf{a}\in\mathbb{R}^{n}$ and $q^{(j)}$ are the positive non-dominant
exponents that are defined by (\ref{sub_exp}) below.

$\bullet$ We then substitute the forms
$\mathbf{x}=\mathbf{a}\mathbf{\tau}^{\mathbf{p}}$
in the system $\dot{\mathbf{x}}=\mathbf{f}^{(0)}(\mathbf{x})$ in order to find
all possible \emph{dominant balances}, i.e., sets of the form
$\{\mathbf{a},\mathbf{p}\}$. The \emph{order} of each balance is defined by
the number of the non-zero components of $\mathbf{a}$.

$\bullet$ For each of these balances we check the validity of the following
\emph{dominance condition}
\be
\label{dominance}
\lim_{\tau\rightarrow 0}
\frac{\Sigma_{j=1}^{k}\mathbf{f}^{(j)}(\mathbf{a}\tau^{\mathbf{p}})}
{\tau^{\mathbf{p}-1}}=0,
\ee
and define the non-dominant exponents $q^{(j)}$, $j=1,\ldots, k$ by
the requirement
\be
\label{sub_exp}
\frac{\Sigma_{j=1}^{k}\mathbf{f}^{(j)}(\tau^{\mathbf{p}})}
{\tau^{\mathbf{p}-1}}\sim\tau^{q^{(j)}}.
\ee
The balances that cannot satisfy the condition (\ref{dominance}) are then
discarded.

$\bullet$ Next, we compute the Kovalevskaya matrix $\mathcal{K}$ defined by
\be
\mathcal{K}=D\mathbf{f}^{(0)}(\mathbf{a})-diag \mathbf{p},
\ee
where $D\mathbf{f}^{(0)}(\mathbf{a})$ is the Jacobian matrix of
$\mathbf{f}^{(0)}$ evaluated on $\mathbf{a}$.

$\bullet$ We calculate the spectrum of the $\mathcal{K}$-matrix,
$spec(\mathcal{K})$, that is the set of its $n$ eigenvalues, called also
the \emph{$\mathcal{K}$-exponents}.
The arbitrary constants of any particular or general
solution first appear in those terms in the series (\ref{Puiseux}) whose
coefficients $\mathbf{c}_{k}$ have indices $k=\rho s$, where $\rho$ is a
non-negative $\mathcal{K}$-exponent and $s$ is the least common multiple of
the denominator of the set consisting of the non-dominant exponents $q^{(j)}$
and of the positive $\mathcal{K}$-exponents (cf. \cite{skot}, \cite{goriely}).
The number of non-negative $\mathcal{K}$-exponents equals therefore the number of
arbitrary constants that appear in the series expansions of (\ref{Puiseux}).
There is always the $-1$ exponent that corresponds to an arbitrary
constant that is the position of the singularity at $t_{s}$. A dominant balance
corresponds thus to a general solution if it possesses $n-1$ non-negative
$\mathcal{K}$-exponents (the $n$th arbitrary constant is the position of the
singularity, $t_{s}$).

$\bullet$ Next, we substitute the Puiseux series:
\be
x_{i}=\Sigma_{j=0}^{\infty}c_{ji}\tau^{p_{i}+j/s}, i=1,\ldots, n,
\ee
in the system (\ref{arb_syst}).

$\bullet$ We find the coefficients $\mathbf{c}_{j}$ by
solving the recursion relations
\be
\mathcal{K}\mathbf{c}_{j}-\frac{j}{s}\mathbf{c}_{j}=
\mathbf{P}_{j}(\mathbf{c}_{1},\ldots, \mathbf{c}_{j-1})
\ee
where $\mathbf{P}_{j}$ are polynomials that are read off from the original system.

$\bullet$ We verify that for every $j=\rho s$, with $\rho$ a positive
$\mathcal{K}$-exponent, the following compatibility conditions hold:
\be
\mathbf{\upsilon}^{\intercal}\cdot\mathbf{P}_{j}=0,
\ee
where $\mathbf{\upsilon}$ is the eigenvector associated with the
positive $\mathcal{K}$-exponent $\rho$.

$\bullet$ We repeat the procedure for each possible decomposition.

We note here that if the compatibility condition is violated in some eigenvalue
then the Puiseux series is not valid and instead we have to substitute a
\emph{$\psi$-series}, which is a series of the form
\be
\mathbf{x}=\tau^{\mathbf{p}}\left(\mathbf{a}
+\Sigma_{i=1}^{\infty}\Sigma_{j=1}^{\infty}\mathbf{c}_{ij}\tau^{i/s}(\tau^{\rho}\log\tau)^{j/s}\right),
\ee
where $\rho$ is the first $\mathcal{K}$-exponent for which
the compatibility condition is violated. The rest of the procedure in this case
is the same as before.

\end{document}